\documentclass[twocolumn,showpacs,superscriptaddress]{revtex4}
\usepackage{tikz}
\usepackage[T1]{fontenc}
\usepackage{graphicx}
\usepackage{amssymb}
\usepackage{times}
\usepackage[colorlinks]{hyperref} 
\usepackage{color,amssymb,mathptm,times}
\usepackage{amsfonts,amssymb,latexsym,amsmath,amscd}

\newcommand{\bed}{\[}
\newcommand{\eed}{\]}
\newcommand{\beq}{\begin{equation}}
\newcommand{\eeq}{\end{equation}}
\newcommand{\beqa}{\begin{eqnarray}}
\newcommand{\eeqa}{\end{eqnarray}}
\newcommand{\ket} [1] {\vert #1 \rangle}
\newcommand{\bra} [1] {\langle #1 \vert}
\newcommand{\braket}[2]{\langle #1 | #2 \rangle}

\newcommand{\mean}[1]{\langle #1 \rangle}
\newcommand{\gras}[1]{\bold{#1}}

\newcommand{\Tr}{\mathop{\mathrm{Tr}}}

\newcommand{\Proba}[1] {\textrm{Proba}\big[ #1 \big]} 

\newcommand{\be}{\begin{eqnarray}}
\newcommand{\ee}{\end{eqnarray}}
\newcommand{\bea}{\begin{eqnarray}}
\newcommand{\eea}{\end{eqnarray}}
\newcommand{\bma}{\begin{subequations}}
\newcommand{\ema}{\end{subequations}}

    {
    \smallskip
    \refstepcounter{theorem}
    \noindent
    {\bf Example \Roman{section}.\arabic{theorem}} \ \ }
    {
    \smallskip}

    {
    \smallskip
    \refstepcounter{theorem}
    \noindent
    {\bf Remark \Roman{section}.\arabic{theorem}} \ \ }
    {\hspace*{\fill}
    \smallskip}

    {
    \smallskip
    \refstepcounter{theorem}
    \noindent
    {\bf Definition \Roman{section}.\arabic{theorem}} \ \ }
    {\hspace*{\fill}{\ }
    \smallskip}

    {
    \smallskip
    \refstepcounter{theorem}
    \noindent
    {\bf Definition \Alph{section}.\arabic{theorem}} \ \ }
    {\hspace*{\fill}{\ }
    \smallskip}

    {
    \smallskip
    \refstepcounter{theorem}
    \noindent
    {\bf Scholium \Roman{section}.\arabic{theorem}} \it \ \ }
    {\hspace*{\fill}{\ }
    \smallskip}

\newenvironment{proof}[1][]
    {
    \noindent
    {\bf Proof{#1}:  }
    }
    {\hspace*{\fill}{$\Box$}\smallskip}

    {
    \noindent
    {\bf Sketch{#1}:  }
    }
    {\hspace*{\fill}{$\Box$}\smallskip}

    {
    \noindent
    {\bf Reference:  }
    }
    {\hspace*{\fill}{$\odot$}\smallskip}

\newtheorem{theorem}{Theorem}[section]

\newtheorem{lemma}[theorem]{Lemma}

\usepackage{graphicx}
\usepackage[mathscr]{eucal}
\usepackage{amsmath}
\usepackage{amssymb}
\usepackage{epstopdf}
\usepackage{epsfig}
\usepackage{wrapfig}

\def\one{\ensuremath{\hbox{$\mathrm I$\kern-.6em$\mathrm 1$}}}

\begin{document}

\title{Low depth quantum circuits for Ising models}
\author{S. Iblisdir}
\email{iblisdir@ecm.ub.es}
\affiliation{Dept. Estructura i Constituents de la Mat$\grave{e}$ria, Universitat de Barcelona, 08028 Barcelona, Spain}
\author{M. Cirio}
\affiliation{Centre for Engineered Quantum Systems, Department of Physics and Astronomy, Macquarie University, North Ryde, NSW 2109, Australia}
\author{O. Boada}
\affiliation{Dept. Estructura i Constituents de la Mat$\grave{e}$ria, Universitat de Barcelona, 08028 Barcelona, Spain}
\author{G.K. Brennen}
\affiliation{Centre for Engineered Quantum Systems, Department of Physics and Astronomy, Macquarie University, North Ryde, NSW 2109, Australia}

\date{\today}  


\begin{abstract}

A scheme for measuring complex temperature partition functions of Ising models is introduced. In the context of ordered qubit registers this scheme finds a natural translation in terms of global operations, and single particle measurements on the edge of the array. Two applications of this scheme are presented. First, through appropriate Wick rotations, those amplitudes can be analytically continued to yield estimates for partition functions of Ising models. Bounds on the estimation error, valid with high confidence, are provided through a central-limit theorem, which validity extends beyond the present context. It holds for example for estimations of the Jones polynomial. Interestingly, the kind of state preparations and measurements involved in this application can in principle be made ``instantaneous", i.e. independent of the system size or the parameters being simulated. Second, the scheme allows to accurately estimate some non-trivial invariants of links. A third result concerns the computational power of estimations of partition functions for real temperature classical ferromagnetic Ising models on a square lattice. We provide conditions under which estimating such partition functions allows one to reconstruct scattering amplitudes of quantum circuits making the problem BQP-hard.  Using this mapping, we show that fidelity overlaps for ground states of quantum Hamiltonians, which serve as a witness to quantum phase transitions, can be estimated from classical Ising model partition functions.  Finally, we show that the ability to accurately measure corner magnetizations on thermal states of two-dimensional Ising models with magnetic field leads to fully polynomial random approximation schemes (FPRAS) for the partition function. Each of these results corresponds to a section of the text that can be essentially read independently. 

\end{abstract}

\maketitle


\section{Introduction}

Statistical Mechanics provides formal recipes to study interacting many-body systems. Quantities that can be experimentally probed, such as the free energy or the specific heat, can in principle be derived in a straightforward manner. More often than not, however, computing these quantities turns out to be impossible in a limited time. As can be seen from very idealised systems, our ability to actually apply these recipes is very limited. During the last ten years, significant efforts have been devoted to investigating whether quantum mechanics could help in this respect. Various methods, all involving the superposition principle, have been proposed to compute the Jones polynomial at particular values of its variable \cite{AJL}, partition functions of classical statistical models \cite{vdN,AL,GeraciLidar}, the Tutte polynomial \cite{Tutte}, or more generally to contract tensor networks \cite{AL}. 

In this work, we will mainly focus on a collection of \emph{classical} two-level systems, each attached to a fixed position corresponding to a vertex of some lattice $\Lambda$, with edges $E(\Lambda)$. The state of a particle located at vertex $i$ is associated with a number $\sigma_i$ taking values in $\{-1,+1\}$. The energy of the system is given by an Ising Hamiltonian function, associating an energy with each classical configuration of the system 
$\sigma_{\Lambda}$:
\beq\label{eq:def-Ising}
H( \sigma_{\Lambda} )=-\sum_{i} h_i \sigma_i- \sum_{\langle i, j \rangle} J_{i,j}  \sigma_i  \sigma_j.
\eeq

The first sum in this equation runs over all vertices of $\Lambda$. The quantity $h_i$ models represents some local field felt by a spin located at position 
$i$. The second sum represents interactions between pairs of neighbour particles (edges of the lattice). The strength and sign of these interactions may vary from pair to pair. This model was introduced by Lenz as an idealisation of systems where magnetic interactions prevail \cite{Lenz}. Although innocent looking, it exhibits an extremely rich structure. On a regular lattice, close to a phase transition, its long range behaviour is similar to that of very interesting field theories \cite{DiF} while the problem of computing its partition function,
\beq\label{eq:def-partition-function}
Z(\beta)=\sum_{\{ \sigma \}} \text{exp} \big[-\beta H(\{ \sigma \}) \big],
\eeq
belongs the NP-hard complexity class \cite{Barahona}.

It is the purpose of this paper to present schemes that allows to accurately estimate $Z(\beta)$ for imaginary values of $\beta$ (Section \ref{sect:ctpf}), through manipulation of a suitable quantum mechanical system. Quantum circuits for this task have been previously proposed in Ref.\cite{Cuevas}.  However with our scheme, we will see how to evaluate partition functions of real systems, through analytic continuation (Section \ref{sec:anacon}).  A central-limit theorem is derived that allows to estimate the discrepancy between the partition function we wish to estimate and the estimate provided by the quantum algorithm. Interestingly, this theorem is also valid for a wide class of quantum algorithms, including well-known proposals to use a quantum computer in order to evaluate the Jones polynomial \cite{AJL}.  
As we shall see, the kind of preparation and measurement necessary for this estimation can in principle be made \emph{in constant time}, i.e. independent of the system size or the parameters being simulated. This feature is particularly appealing in view of possible practical implementations. We will then see that imaginary temperature partition functions are interesting in their own right, because they provide non-trivial invariant of knots (Section \ref{sect:ki}). Section \ref{sect:compu-power} deals with computational complexity issues.  We investigate the (quantum) computational power of the Ising model, and show how the ability to estimate real temperature partition functions of this model allows to efficiently simulate a quantum computer.  One application of this is the estimation of the wavefunction overlap, termed fidelity, between ground states of a quantum Hamiltonian in the vicinity of a quantum phase transition. We also show that some much simpler tasks have computational power. In particular, the ability to detect corner magnetisations of disordered Ising models leads to fully polynomial random approximation schemes  thereof.  Many of the quantum algorithms presented here involve repetitions of either constant depth or linear depth circuits and moreover many of the operations can be performed without individual qubit addressability.  This is potentially a real boon to experimental implementations in architectures such as trapped atoms in optical lattices or superconducting qubit arrays where individual addressing is not so easy but many qubits are available.   In additional some of the circuits provide for a trade off in space and time, i.e. one can perform either constant depth circuits in $d+1$ spatial dimensions or linear depth circuits in $d$ dimensions. Constant depth quantum circuits have attracted attention since the discovery of simple examples (depth-1 circuits) that are expected to be difficult to simulate classically \cite{Bremner}.  Furthermore, there is some evidence that fault tolerance thresholds could be improved for constant depth (or more generally logarithmic depth) quantum circuits \cite{Razborov, Preskill}. 

\section{Complex temperature partition functions}\label{sect:ctpf}

We wish to study a classical system defined on some $d$-dimensional lattice $\Lambda$. For that purpose, we consider an associated situation, where a two-level system is located on each vertex of $\Lambda$. The computational basis for each quantum particle, $\{ \ket{+}, \ket{-} \}$, will be associated with classical individual spin configurations. Our construction relies on controlled phase gates acting on nearest neighbours, that is, elements $\langle k,l \rangle$ of $E(\Lambda)$, the set of edges of the lattice. Their action is best described in computational basis:
\beq\label{eq:def-controlled-phase}
C_{k,l}: \ket{\sigma_k,\sigma'_l} \to e^{i \phi_{k,l}(\sigma_k,\sigma'_l)} \ket{\sigma_k,\sigma'_l}.
\eeq

Importantly, these phase gates all commute with each other: 
\beq
\forall \langle k,l \rangle, \langle x,y \rangle \in E(\Lambda), \hspace{0.3cm}
[C_{k,l },C_{x,y }]=0.  
\eeq
Obviously, each function $\phi_{\langle k,l \rangle}$ can be expressed as 
\bed
\phi_{\langle k,l \rangle}(\sigma_k,\sigma'_l)=\sum_{s= \pm 1} \sum_{s'= \pm 1} \phi_{k,l }(s,s') \delta_{s \sigma_k} \delta_{s' \sigma'_l}.
\eed
With the definitions $\kappa_k \equiv \frac{1}{4} \sum_{s_k,s_l} \phi_{ k,l}(s_k,s_l),
J_{k,l} \equiv \frac{1}{4} \sum_{s_k,s_l} \phi_{k,l}(s_k,s_l) s_k s_l,
h_k \equiv \frac{1}{4} \sum_{s_k, s_l} \phi_{k,l }(s_k,s_l) (s_k+s_l)$, we see that a collective action of controlled phase gates across all edges of the lattice can be described in the computational basis as\footnote{Note to a reader interested in reproducing the calculations: the identity $\delta_{\sigma \sigma'}=\frac{1+ \sigma \sigma'}{2}$ has been repeatedly used.}
\beq\label{eq:collective-controlled-phase}
\begin{array}{lll}
\prod_{\langle k,l \rangle \in E} C^{\; \alpha}_{k,l} \prod_{k \in \Lambda} \ket{\sigma_k}
&=&
\textrm{exp}\big[ i \alpha \sum_{k \in \Lambda} \kappa_k\\
&+& i \alpha \sum_{k \in \Lambda} h_k \sigma_k \\
&+& 
i \alpha \sum_{\langle k,l \rangle \in E} J_{k,l } \sigma_k \sigma_l  \big] \\
&&\times \prod_{k \in \Lambda} \ket{\sigma_k}.
\end{array}
\eeq
In particular, if each quantum particle is initialized in the state 
\beq\label{eq:initial-individual}
\ket{+_x} \equiv \frac{1}{\sqrt{2}}(\ket{+}+\ket{-}),
\eeq
we see that the mean value of a product of phase gate operators takes the form of a partition function at imaginary temperature $i \alpha$:
\beq\label{eq:ima-pf-simple}
\begin{array}{lll}
A(\alpha) &\equiv& \bra{+_x^{\otimes |\Lambda|}} \prod_{\langle k,l \rangle \in E} C^{\alpha}_{kl} \ket{+_x^{\otimes |\Lambda|}}\\
&=&\frac{1}{2^{|\Lambda|}} \sum_{\{ \sigma \}} e^{- i \alpha H({\sigma})}, 
\end{array}
\eeq
with $H$ of the form given by Eq.(\ref{eq:def-Ising}). 

It is actually possible to get partition functions of a classical $(d+1)$-dimensional system through evolution of a $d$-dimensional quantum system. For that, we use two additional kinds of gates besides the controlled phase gate. The first kind is single qubit rotations:
\beq\label{eq:def-rotation-gate}
\begin{array}{lll}
U_k: \ket{+} &\to& \cos\theta_k \ket{+} + \sin\theta_k \ket{-}, \nonumber \\
U_k: \ket{-} &\to& -\sin\theta_k \ket{+} + \cos\theta_k \ket{-}. \nonumber 
\end{array}
\eeq
As discussed in Appendix \ref{appendix_a-priori}, other choices are possible. The second is single qubit phase gate:
\beq
P_k(\varphi_k): \ket{\sigma_k} \to e^{i \varphi_k \sigma_k} \ket{\sigma_k}.
\eeq
Next, we observe that the matrix elements of $U_k$ can be expressed in exponential form for almost all values of the parameters $\theta_k$:
\beq
\bra{\sigma'_k} U_k \ket{\sigma_k}=\exp \big[J_k^{\downarrow} \sigma_k \sigma'_k + i \frac{\pi}{4} \sigma'_k - i \frac{\pi}{4} \sigma_k+B(\theta_k) \big],  
\eeq
with $\theta_k \notin \{ k \frac{\pi}{2}: \; k \in \mathbb{Z}\}$ and where: 
\beq\label{eq:down-interaction}
J_k^{\downarrow}=-\frac{1}{2} \ln(\tan \theta_k) -i \frac{\pi}{4}, \hspace{0.7cm} 
\eeq
and
\beq\label{eq:B-factor}
B(\theta_k)=\frac{\ln(\cos(\theta_k))}{2}+\frac{\ln(\sin(\theta_k))}{2}+i\frac{\pi}{4}\;\;.
\eeq
These individual rotations $\{U_k, k \in \Lambda \}$ are applied on all lattice sites simultaneously. For bookkeeping, it is convenient to assume there is an external clock recording the moment $t$ where simultaneous rotations are applied, and ticking at exactly this time. There is nothing particular to this clock, it is just a way to label the change of variables necessary to describe the action of the $U_k$ gates: 
\beq\label{eq:collective-U-gates}
\begin{array}{lll}
\prod_{k \in \Lambda} U_k(t)  \ket{ \sigma(t) }&=& G(t)
\sum_{\{ \sigma(t+1) \} }\\
&&\exp \big[ \sum_{k \in \Lambda} J_k^{\downarrow}(t) \sigma_k(t) \sigma_k(t+1) \\
&&+ i \frac{\pi}{4} \sum_{k \in \Lambda} (\sigma_k(t+1) - \sigma_k(t) ) \big]\\
&&\ket{ \sigma(t+1)},
\end{array}
\eeq
where $G(t)=\exp(\sum_{k \in \Lambda} B(\theta_k(t)))$.

Now let us consider a $d$-dimensional lattice $\Lambda$ of particles  each prepared in the state (\ref{eq:initial-individual}). Let us assume that a layer evolution operator 
\beq
\mathcal{L}(t)  =  \prod_{k \in \Lambda} P_k(-\frac{\pi}{4}) \prod_{k \in \Lambda} U_k(t) \prod_{\langle k,l \rangle \in E} C_{k,l}^{\alpha}(t) \prod_{k \in \Lambda} P_k(\frac{\pi}{4}).
\eeq
is applied $(m-1)$ times on this initial state, leading to the final state $\prod_{t=1}^{m-1} \mathcal{L}(m-t) \ket{+_x}^{\otimes |\Lambda|}$ (see Fig.\ref{fig:optlat1}).

The overlap of this state with the initial state $\ket{+_x^{\otimes |\Lambda|}}$ takes again the form of an Ising partition function, but now defined on an enlarged lattice $\hat{\Lambda}=\Lambda \times \{1, \ldots, m\}$: 
\beq\label{eq:amp2}
\begin{array}{lll}
A(\alpha,\Theta) &\equiv& \bra{+_x^{\otimes |\Lambda|}}\prod_{t=1}^{m-1} \mathcal{L}(m-t) \ket{+_x^{\otimes |\Lambda|}}\\
&=&\frac{1}{2^n}\sum_{\sigma}\exp[-H(\sigma)],
\end{array}
\eeq
where, $\Theta$ denotes collectively all individual rotations performed on the system, and where, up to an additive constant $\sum_{t=1}^m \ln G(t)$, the classical Hamiltonian $H$ with imaginary couplings is 
\begin{equation}
\begin{array}{lll}
-H(\sigma)&=& i \alpha \sum_{t=1}^{m} \sum_{k \in \Lambda} h_k(t)  \sigma_k(t) \\
&&+ i \alpha \sum_{t=1}^{m} \sum_{\langle k,l \rangle \in E} J_{k,l}(t) \sigma_k(t) \sigma_l(t)\\
&&+ \sum_{t=1}^{m-1} \sum_{k \in \Lambda} J_k^{\downarrow}(t) \sigma_k(t) \sigma_k(t+1).
\end{array}
\end{equation}
Eq.(\ref{eq:amp2}) is proven by inserting identity operators and identifying single-particle quantum basis states $\ket{\pm}$ with single particle classical spin configurations $\ket{\sigma}$: 
\bed
\begin{array}{llll}
&\bra{+_x^{\otimes |\Lambda|}}\prod_{t=1}^{m-1} \mathcal{L}(m-t) \ket{+_x^{\otimes |\Lambda|}}& \\
&=\frac{1}{2^n}\sum_{\sigma(1) \ldots \sigma(m)} \prod_{t=1}^{m-1} \bra{\sigma(m-t+1)}\mathcal{L}(m-t) \ket{\sigma(m-t)}.
\end{array}
\eed

\begin{figure}
\begin{center}
\includegraphics[width=\columnwidth]{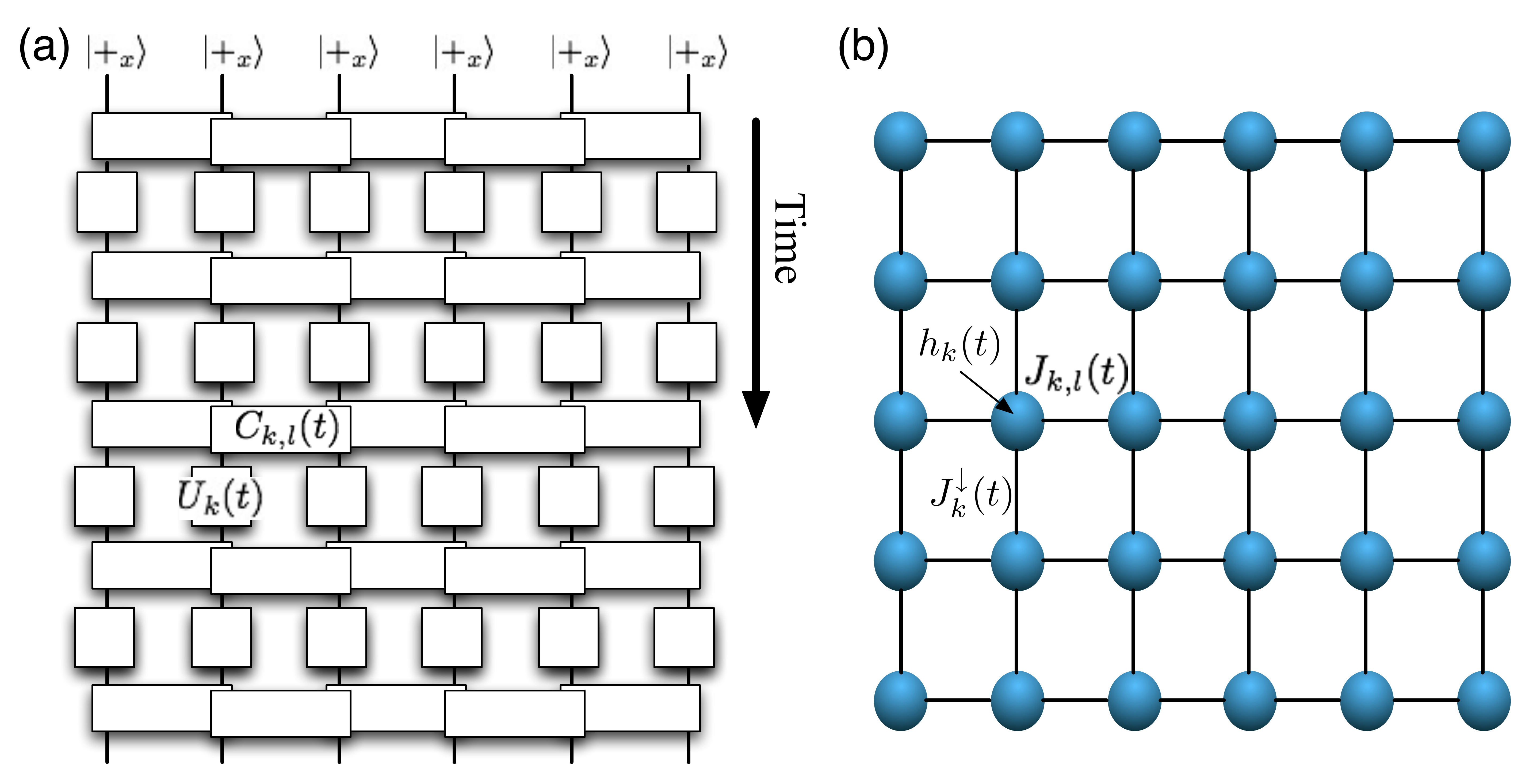}
\end{center}
\caption{Example of the quantum algorithm on a 1D chain of qubits to compute the partition function of a 2D classical Ising model at imaginary temperature.  (a)  The quantum algorithm begins with qubits in the chain initialized in state $\ket{+_x}$ and proceeds with alternating sequences of parallel nearest neighbour two qubit gates $C_k(t)$ diagonal in the computational basis $\{\ket{\pm}\}$ and parallel local rotations $U_k(t)$ (supplemented by single qubit phase gates).  (b)   The corresponding classical Ising model with spatially dependent horizontal and vertical bond strengths and local magnetic fields.}
\label{fig:optlat1}
\end{figure}

\section{Implementation}\label{sect:imp}

At the core of the discussion held in the previous section lies the ability to measure the scalar product between $n$-particle states $\ket{\Phi}$ and $\ket{\Psi}$. We will describe two measurement protocols addressing this problem. The first is the simpler and allows to detect $|\braket{\Phi}{\Psi}|^2$, while the second truly yields $\braket{\Phi}{\Psi}$.

\textbf{Protocol 1}

\begin{enumerate}

\item Prepare an $n$-particle system $A$ in the state $\ket{\Psi}$, and an $n$-particle system $B$ in the state $\ket{\Phi}$.

\item Prepare an ancillary register $R$ of $n$ qubits in the state $\ket{GHZ}=\frac{1}{\sqrt{2}}(\ket{+ \ldots +}+\ket{- \ldots -})$.

\item  Perform a bit-wise controlled swap gate with each qubit $R_j$ of the register as a control and $A_j$, $B_j$ as targets, i.e. if qubit $R_j$ is in the state $\ket{-}$ then apply $\texttt{SWAP}(A_j,B_j)$. We get 
\bed
\begin{array}{lll}
\frac{1}{\sqrt{2}}\Big(\ket{+ \ldots +}_R \ket{\Psi}_A \ket{\Phi}_B + \ket{- \ldots -}_R  \ket{\Phi}_A \ket{\Psi}_B\Big)  
\end{array}
\eed

\item  Measure the first $n-1$ qubits  of $R$ in the basis $\{ \ket{\pm_x}\}=\{ \frac{1}{\sqrt{2}} (\ket{+}\pm \ket{-})  \}$. Denote $m_j=\pm 1$ the (equiprobable) outcomes of measurement on register qubit $j$ and define $\chi=\sum_{j=1}^{n-1} m_j$.  The state for the last qubit of the register and the system $AB$ is 
\beq
\frac{1}{\sqrt{2}}\Big(\ket{+, \Psi, \Phi}+ (-1)^{\chi} \ket{-, \Phi, \Psi}\Big).
\eeq

\item
Measure the Pauli operator $\sigma^x$ of the last ancillary qubit $R_n$.  The expectation value is
\beq
\langle \sigma^x_n\rangle=(-1)^\chi |\braket{\Phi}{\Psi}|^2.
\eeq

\end{enumerate}

\textbf{Protocol 2}

\begin{enumerate}

\item Prepare an $n$-particle system $A$ in the state $\ket{\Phi}$.

\item Prepare an ancillary register $R$ of $n$ qubits in the GHZ state 

\item  Evolve the qubits in register $A$ conditioned on the state of the ancilla to prepare
\bed
\frac{1}{\sqrt{2}}\Big(\ket{+ \ldots +}_R \ket{\Phi}_A +
\ket{- \ldots -}_R  \ket{\Psi}_A\Big)
\eed

This can be done by replacing all instances of quantum gates in the evolution of $ \ket{\Phi} \to \ket{+_x^{\otimes n}} \to \ket{\Psi}$ into bitwise controlled gate operations.  The single qubit phase gates become controlled phase gates: $\ket{+}_{R_k}\bra{+}\otimes {\bf 1}_{A_k}+\ket{-}_{R_k}\bra{-}\otimes P_{A_k}$.  Similarly, the single qubit rotations become: $\ket{+}_{R_k}\bra{+}\otimes {\bf 1}_{A_k}+\ket{-}_{R_k}\bra{-}\otimes U_{A_k}$.  The collisional gates are controlled by one of neighboring ancillary qubits, e.g.:  $\ket{+}_{R_x}\bra{+}\otimes {\bf 1}_{A_x,A_y}+\ket{-}_{R_k}\bra{-}\otimes {C}_{A_x,A_y}$.  Such three qubit diagonal gates can be decomposed into at most $6$ nearest neighbor controlled phase gates \cite{Brennen:03}.

\item  Measure the first $n-1$ qubits  of $R$ in the basis $\{ \ket{\pm_x}\}$. Denote $m_j=\pm 1$ the outcome of measurement on register qubit $j$ and let again  $\chi=\sum_{j=1}^{n-1} m_j$.  The state for the last qubit of the register and the system $A$ is 
\beq
\frac{1}{\sqrt{2}}\Big(\ket{+, \Phi}+ (-1)^{\chi} \ket{-, \Psi}\Big).
\eeq

\item
Measure $\sigma^x$ on the last ancillary qubit $R_n$.  The expectation value is 
\beq
\langle \sigma^x_n\rangle=(-1)^\chi\Re[\braket{\Phi}{\Psi}].
\eeq

\item
Repeat steps 1-4 but on the last qubit $R_n$ measure instead the Pauli operator $\sigma^y$ where the basis $\{\ket{\pm_y}\}=\{\frac{1}{\sqrt{2}}(\ket{+}\pm i\ket{-})\}$.  The expectation value is
 \beq
\langle \sigma^y_n\rangle=(-1)^\chi \Im[\braket{\Phi}{\Psi}].
\eeq

\end{enumerate}

We note that it is actually not necessary to prepare size $n$ ancillary registers in a $GHZ$ state for either measurement protocol, since one ancillary qubit making controlled swaps or controlled interactions like a serial tape head over the quantum registers would suffice.  The penalty is a potentially linear slowdown and the need to transport the ancilla qubit over the register for every gate in the circuit.  The $\ket{GHZ}$ state can be prepared in one plane using global, i.e. spatially homogeneous, pulses in the plane \cite{Benjamin,Raussendorf,Fitzsimons}.  Futhermore, by coupling the quantum register with a common bosonic mode, $\ket{GHZ}$ states can be prepared in constant time \cite{Bren09}.  The idea is to place all the spins inside a high $Q$ cavity (with decay rate $\kappa$) with a resonance field frequency close to the transition between the qubit states and some other excited state.  When the coupling between the field and qubits is spin dependent and dispersive (e.g. a differential light shift induced by polarization section rules or by spin dependent detuning) then the interaction is modelled as:
\begin{equation}
V_{z}=g_{z}a^{\dagger }a\sum_{j}\sigma _{j}^{z},
\label{setofops}
\end{equation}
where $g_{z}$ is the dispersive coupling strength.  Then $\ket{GHZ}$ can be produced either using strong coupling with a quantised state of light or via a geometric phase gate using coherent state displacements.  We outline the latter as follows:
\begin{itemize}
\item
Initialize all the spins in $\ket{+_x}$ and the cavity mode in the vacuum state $\ket{\alpha=0}$.
\item
Perform the following nine step interaction sequence:
\[
\begin{array}{lll}
&&D(-\beta^{-\kappa \tau})e^{-i \tau V_Z}D(-\alpha^{-\kappa \tau})[\prod_{j}\sigma^x_j]e^{-i \tau V_z}\\
&&[\prod_{j}\sigma^x_j]D(\beta)e^{-i \tau
V_z}D(\alpha),
\end{array}
\]
where $D(\alpha)=e^{\alpha a^{\dagger}-\alpha^{\ast}a}$ is a coherent state displacement, 
and $e^{-i \tau V_z}$ is the unitary evolution generated by $V_z$.  When the parameters 
satisfy:  $g_z \tau=\pi/2$, and $|\alpha \beta | (e^{-3\kappa \tau/2}+e^{-\kappa \tau/2})=\pi/4$, 
then the cavity returns to the vacuum and the global rotation $U=e^{-i\frac{\pi}{4}\prod_j \sigma^z_j}$ is applied to the qubits.
\item 
Apply the global operation $\prod_j e^{i\frac{\pi}{2\sqrt{2}}(\sigma^x_j+\sigma^z_j)}$
to the spins.
\end{itemize}
The state of the qubits is then $\frac{1}{\sqrt{2}}(\ket{++\ldots +}-i\ket{--\ldots -}$ which is locally equivalent to $\ket{GHZ}$ and functions just as well for the simulation protocols above.  The overall process fidelity, which measures how close the lossy process is to the target unitary $U=e^{-i\frac{\pi}{4}\prod_j \sigma^z_j}$, satisfies  \cite{BrenII09}  
\[F_{\rm pro}\geq 1-\frac{\pi^2\kappa}{%
2|g_z|} \Big(1+\frac{\pi\kappa}{2|g_z|}\Big).
\]
 Note that this is a constant depth circuit thanks to the non-local coupling of the field to the qubits.  Of course as the number of spins increases the size of the cavity must also increase, and the strength of the field, spin coupling decreases as $1/\sqrt{Vol}$ where $Vol$ is the cavity volume. Consequently, there is ultimately a process time which scales as $\sqrt{n}$ where $n$ is number of qubits. However, in practice this could be quite fast compared to a sequential circuit for generating $\ket{GHZ}$.

Since the measurement of $\braket{\Phi}{\Psi}$ is informationally more complete than that of $|\braket{\Phi}{\Psi}|^2$, the reader might wonder why we have bothered describe a separate procedure to measure the latter quantity. The reason is that Protocol 2 is experimentally more demanding than Protocol 1 since all the gates must be promoted to controlled gates based on the state of the ancilla.  For most of the discussion to follow we assume information is obtained from Protocol 2, while results for partition function reconstructions using Protocol 1 are presented in Appendix \ref{appendix_a-priori}.

So far, we have considered the case of planar boundary conditions. If the classical system is periodic in space (i.e. the lattice $\Lambda$ is periodic) then the above quantum algorithm is simply modified in the couplings $J_{k,l}(t)$ to account for this. If the classical system is periodic in the time direction, then a few modifications are needed.  To relate the measurement of the quantum system to the classical partition function, the boundaries states $\ket{\sigma_k(m)}$ and $\ket{\sigma_k(1)}$ must be identified.  So rather than computing the scattering matrix element $\bra{+_x^{\otimes |\Lambda|}}W\ket{+_x^{\otimes |\Lambda|}} $, where the unitary $W$ is defined as $W= \prod_{t=1}^{m-1} \mathcal{L}(m-t) $, as we have described so far, we want the trace:  $\Tr [W]$. This is found by using the measurement Protocol 2 but with the register $A$ prepared in the completely mixed state $\frac{{\bf 1}}{2^n}$.  The polarization measurements of the last ancilla of the register then yield the real and imaginary parts of $\frac{\Tr[W]}{2^n}$.  Also note by the cyclic property of the trace, the phase gates $P_k$ are no longer needed in the quantum evolution.

Consider the implementation of this measurement for a 3D classical Ising model using a quantum register encoded in a plane.  For Protocol 1 three parallel planes are needed, one (the top plane) prepared in a $\ket{GHZ}$ state, and the centre (c) and bottom (b) planes both prepared in $\ket{+_x^{\otimes n}}$.  The centre plane is prepared in $\prod_{\langle k,l \rangle \in E} C^{\alpha}_{k,l} \ket{+_x^{\otimes |\Lambda|}}$ or evolved in $\prod_{t=1}^{m-1} \mathcal{L}(m-t) \ket{+_x^{\otimes |\Lambda|}}$, and the subsequent $\texttt{C-SWAP}$ gates between registers can be implemented in parallel bitwise between pairs $(c_k,b_k)$ using a sequence of at most $12$ nearest neighbor collisional gates \cite{Brennen:03}.  Finally the measurement of the top register only requires collecting the parity of measurement outcomes of $n-1$ qubits in the bulk (without addressability) and an addressable measurement of $X_n$ for one qubit on a corner.  For Protocol 2 two registers are needed:  the top one prepared in $\ket{GHZ}$ state and the bottom prepared in $\ket{+_x^{\otimes n}}$.  During the quantum evolution all gates acting on the bottom register (say qubit $b_k$) are to be controlled by the neighbouring qubit on the top plane (qubit $t_k$).  For a rotation gates $U_k(t)$ this means to instead apply the controlled gate $\ket{+}_{t_k}\bra{+}\otimes {\bf 1}_{b_k}+\ket{-}_{t_k}\bra{-}\otimes U_k(t)$.   Such a gate can be done using at most 3 controlled collision gates between $t_k$ and $b_k$.  For the two qubit gates $C_{k,l}(t)$ we need to apply $\ket{+}_{t_k}\bra{+}\otimes {\bf 1}_{b_k,b_{l}}+\ket{-}_{t_k}\bra{-}\otimes C_{k,l}(t)$.  This three qubit diagonal gate can be realized using using at most $12$ collisional gates between nearest neighbors $t_k,b_k$ and $b_k,b_{l}$.  Since not all the gates now commute it is necessary to do this in two stages over non overlapping pairs of nearest neighbors in the bottom register.   Measurement of the top register proceeds as for Protocol 1.   
 Regarding addressability, it is necessary to be able to address the different planes along $\hat{z}$ but addressability can be relaxed in the $\hat{x}-\hat{y}$ direction. 

\section{Partition functions}\label{sec:anacon} 

The schemes of Section \ref{sect:ctpf} can be used to provide estimates for real temperature partition functions of classical models. We proceed by analytic continuation of the quantum amplitudes (or their modules) provided by the protocols described in Section \ref{sect:ctpf}.  The general idea is to view the partition function as a polynomial of order linear in the system size whose coefficients are the same as the those obtained from the quantum amplitude estimation but with real instead of complex variables, and then to Wick rotate these variables.

Let $\alpha$ and $\theta$ denote two \emph{complex} variables, and consider a function $F$ of the form
\beq\label{eq:complex-function}
\begin{array}{lll}
F: \mathbb{C} \times \mathbb{C} &\to& \mathbb{C}: (\alpha,\theta) \to F(\alpha,\theta)\\
&=&\sum_{\nu_1=-N_1}^{N_1} \; \sum_{\nu_2=-N_2}^{N_2} c_{\nu_1,\nu_2} e^{i \nu_1 \alpha} e^{i \nu_2 \theta},
\end{array}
\eeq 
where $N_1,N_2 < \infty$. Clearly, $F$ is an \emph{analytic} function, so the coefficients $\{ c_{\nu_1,\nu_2} \}$ define $F$ on the whole complex plane. If $F$ is known for $\alpha_{j_1}= \alpha^{(j_1 )}= 2 \pi \; \frac{j_1}{N_1}, \;  j_1=0 \ldots 2 N_1, \theta_{j_2}=\theta^{(j_2)}=2 \pi \; \frac{j_2}{N_2},  \; j_2=0 \ldots 2 N_2$, then a Fourier transform yields
\beq\label{eq:Fourier-coeff}
\begin{array}{lll}
c_{\nu_1 \nu_2}&=& \frac{1}{(2N_1+1)(2N_2+1)} \sum_{j_1=0}^{2N_1} \sum_{j_2=0}^{2N_2} \\
&&e^{-2 i \pi j_1 \nu_1/(2N_1+1)}e^{-2 i \pi j_2 \nu_2/(2N_2+1)} \\
&&\times~F(\alpha_{j_1},\theta_{j_2}).
\end{array}
\eeq

Plugging this expression in Eq.(\ref{eq:complex-function}), one finds sums of geometric series. Summing them yields
\beq\label{eq:complex-function2}
\hat{F}(\alpha,\theta)=\sum_{j_1=0}^{2N_1} \sum_{j_2=0}^{2N_2}  F(\alpha_{j_1},\theta_{j_2}) \; w^{(N_1)}(\alpha-\alpha_{j_1}) w^{(N_2)}(\theta-\theta_{j_2}) ,
\eeq
where 
\beq\label{def:omega}
w^{(N)}(x) \equiv \frac{1}{2 N+1} \frac{\sin((2 N+1) \frac{x}{2})}{\sin \frac{x}{2}}. 
\eeq

Now consider the quantum amplitudes introduced in Section \ref{sect:ctpf}, in the case where $h_k(t), J_{k,l}(t) \in \{-1,+1 \}$, $\forall k \in \Lambda, \; \forall \langle k,l \rangle \in E(\Lambda), \; \forall t=1 \ldots m$, and where all ``vertical" couplings $J^{\downarrow}$ are set equal. (For the case of non-uniform vertical couplings, see Appendix \ref{appendix_a-priori}). In that case, these quantum amplitudes are certainly of the form (\ref{eq:complex-function}), with $N_1, N_2$ growing at most polynomially with the number of vertices of the classical model being under consideration.  For suitable \emph{complex} values of $\alpha,\theta$, the probability amplitude $A(\alpha,\theta)$ of the $d$-dimensional quantum system can be put in correspondence with the \emph{real} partition function of the $(d+1)$-dimensional classical system. Namely, for
\beq\label{eq:alpha-theta-desired}
\begin{array}{lll}
&&\alpha^{\star}=i \beta, \hspace{1cm}
\theta^{\star}=\frac{1}{i} \ln \sqrt{\frac{1+e^{2 \beta J^{\downarrow} }} {1-e^{2 \beta J^{\downarrow}}}}, \hspace{1cm}\\
&&g(\theta^{\star}) \equiv \frac{1}{2} \ln{\sin 2 \theta^{\star}}+\frac{i \pi}{4}-\frac{1}{2} \ln 2,
\end{array}
\eeq
one finds that $A(\alpha^{\star},\theta^{\star})=e^{|\Lambda| m g(\theta^{\star})} \; Z^{\text{Ising}}(\beta)/2^{|\Lambda|}$. In the definition of $\alpha^{\star}$, we recognise the familiar Wick rotation. The role of the other parameter, $\theta^{\star}$, is to analytically continue the \emph{unitary} quantum mechanical transfer matrix, between successive times, to the (non-unitary) statistical mechanical transfer matrix. In summary, in order to get information about the partition function of a $d$-dimensional classical system, we estimate the probability amplitude $A(\alpha,\theta)$ for well-chosen values of $\alpha$ and $\theta$. From the collected data, we reconstruct the dependence of the function $A$ on its variables $(\alpha,\theta)$, as just explained. Finally, analytic continuation of the variables $(\alpha,\theta)$ to the suitable values (\ref{eq:alpha-theta-desired}) yields an estimate for the desired partition function.

Let us analyse the errors appearing when the values $A(\alpha_{j_1},\theta_{j_2})$ are not known exactly but estimated by some quantities $\varphi_{j_1 j_2}$. The identity (\ref{eq:complex-function2}) allows to get \emph{a priori} error estimate. To simplify the discussion, let us start with the case where partition functions are estimated using a one-time-step protocol. Then, $m=0, N_1= \textrm{poly}(|\Lambda|) \equiv N$ and $N_2=0$. Defining $\delta \varphi^{\star}= \max \{ |\varphi_j-A(\alpha_j)|, j=0 \ldots 2N \}$, we see, through error propagation, that the error at inverse temperature $\beta$, $\Delta A(i \beta)$ satisfies
\beq\label{eq:apriori-error}
\Delta A(i \beta) \leq \sum_{j=0}^{2 N} |w^{(N)}( i \beta-\alpha_{j})| \delta \varphi^{\star}.
\eeq
In the limit of large values of $\beta$, the r.h.s of this equation essentially behaves as $\delta \varphi^{\star} e^{\beta N}$, indicating that the measurement accuracy should shrink exponentially, with the inverse temperature and the size of the system, in order to maintain the error over our estimate for partition function below some fixed prescribed threshold. 

A bound on the error independent of $\beta$ can also be derived easily. Indeed, for the Hamiltonians we are considering, the partition function can be written as 
\bed
Z^{\text{Ising}}(\beta)=\sum_{k=-N}^{N} \xi_k \; e^{-k \beta},
\eed
where all coefficients $\xi_k$ are non-negative integers whose magnitude is at most $2^{m |\Lambda|}$ (number of classical configurations associated with the system). It would therefore be sufficient to be able to estimate these coefficients with a relative accuracy of $2^{-m|\Lambda|}$ in order to be able to reconstruct $Z^{\text{Ising}}(\beta)$ perfectly. The bound appearing on the r.h.s of (\ref{eq:apriori-error}) is independent of the actual values for the link couplings and magnetic fields of the precise Ising model being simulated. We therefore expect it to be pretty loose. 

To get a sharper understanding of how errors behave, we made some numerical simulations. In Fig.\ref{fig:simulations} we show how the error behaves by studying different quantities such as the logarithm of the partition function, the energy and the specific heat. In particular we simulated a model with uniform couplings and zero magnetic fields and a model with $\pm 1$ couplings (with $50\%$ probability) and uniform magnetic field. One can appreciate how, in the uniform case, the error over the partition function goes to zero for zero and infinite temperature. In Appendix \ref{appendix_a-priori}, we show that error over each Fourier coefficient $\xi_k$ is well behaved for large and small values of $k$ (close to $\pm N$), but blows up for intermediate values $k$ (close to $0$). This fact is consistent with our numerical observations and the well known duality present in this model \cite{Bathia}. For the non-homogeneous case, we have found that the errors in the partition function starts by growing exponentially with $\beta$, then remains constant. This observation is consistent with the fact that there is no known low temperature/high temperature duality relation. The errors we have found are also much larger. Our numerics indicate that, in the non-homogeneous case, the magnitude of the partition function is dominated by those $\xi_k$ corresponding to intermediate values of $k$, much more so than for the homogeneous case. 

\begin{figure}
$
\begin{array}{ccc}
\includegraphics[width=\columnwidth]{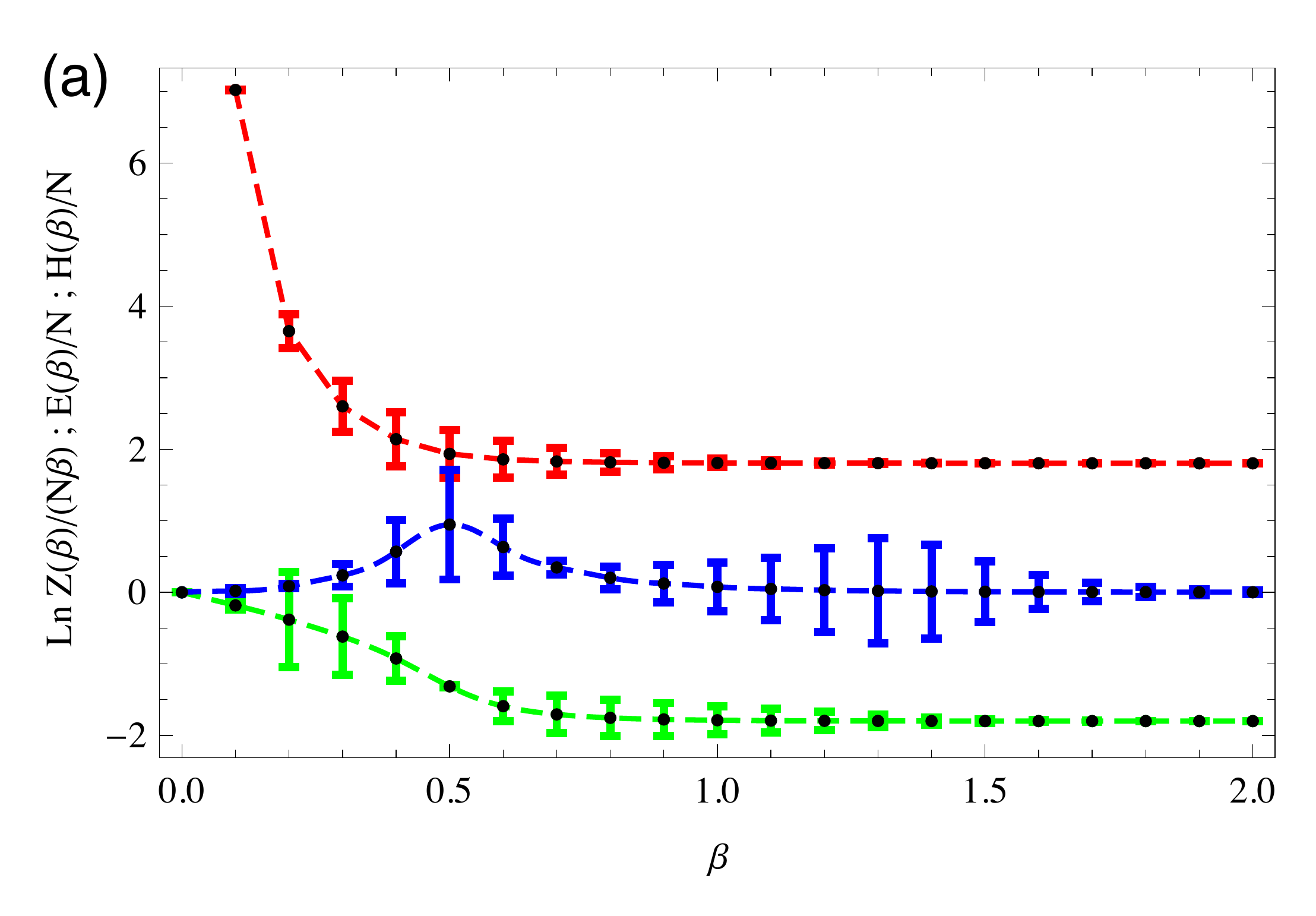}\\
\includegraphics[width=\columnwidth]{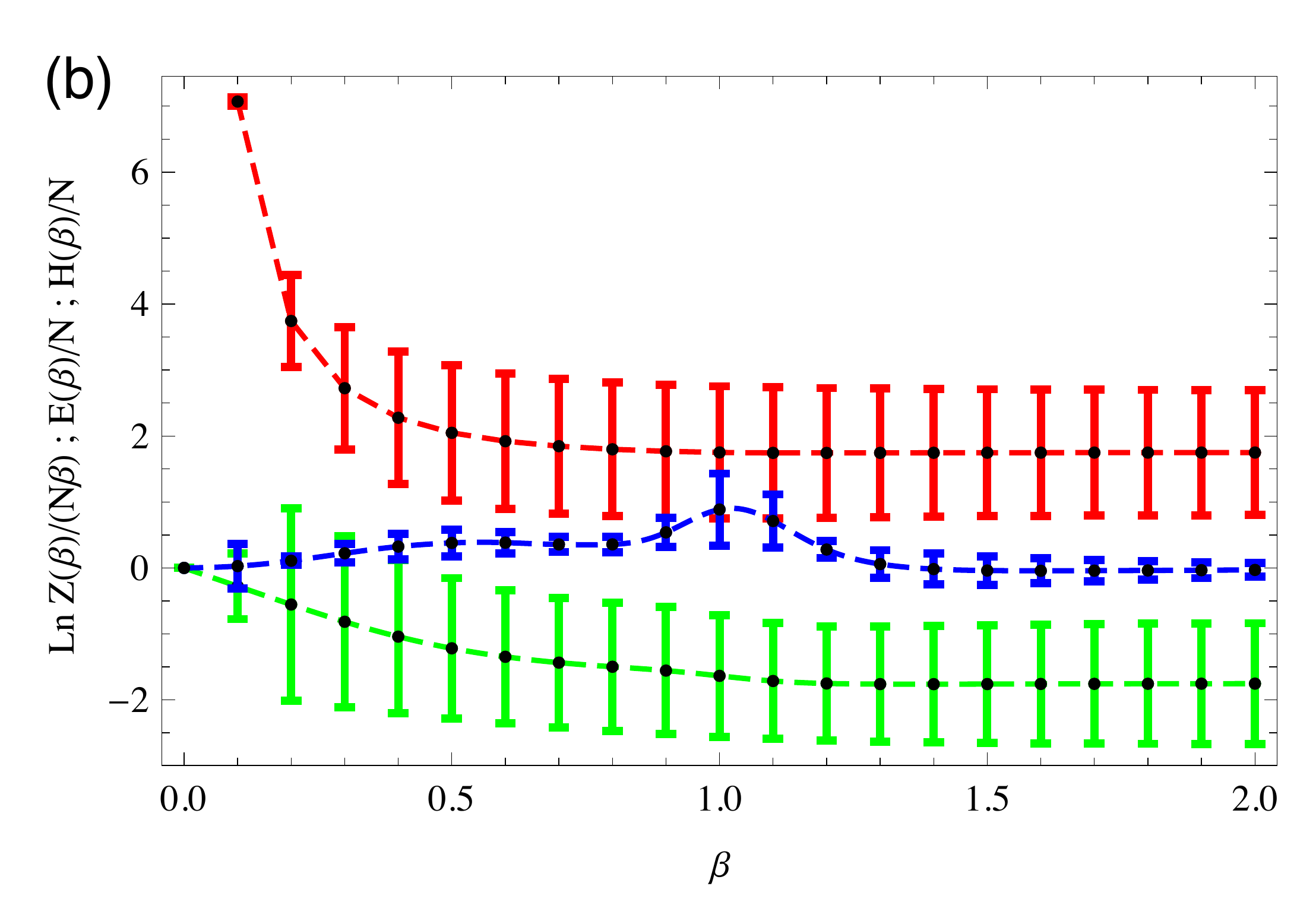}
\end{array}
$
\caption{Example of reconstructed partition functions for Ising model. The reconstructed quantities are the negative of free energy per spin $\ln(Z(\beta)/(N\beta)$ (red), energy per spin $E/N$ (green), and specific heat per spin $H(\beta)/N$ (blue) as a function of temperature and normalized by the number of spins. The plots show the average value of the quantities mentioned above, which is identical to the true value up to numerical machine precision, with error bars representing the a-priori standard deviation. (a) $10\times 10$ classical Ising model with uniform ferromagnetic couplings ($J=1$) and zero magnetic field.  (b) $8\times 8$  classical Ising model with non uniform couplings ($J=\pm 1$ with equal probability) and uniform magnetic field $h=1$.  For the simulation, we supposed to have experimental data with standard deviation equal to $10^{-3}$.}
\label{fig:simulations}
\end{figure} 

Previous attempts at using quantum mechanics to compute approximations of partition functions exhibit errors comparable to ours.  A quantum algorithm based on Fourier sampling was introduced in \cite{Master} to estimate partition functions and free energies of quantum Hamiltonians, which includes the classical Ising model in the case of all diagonal interactions.  There it was found that the number Fourier components needed to be sampled scales polynomially with the lattice size, but in order to obtain a mulitplicative approximation of the partition function, the requisite accuracy of estimation of each coefficient scaled exponentially with the system size. An algorithm, based on using a quantum computer to contract tensor networks yields similar approximation scales \cite{AL}.  Even preparing a quantum state which coherently encodes a classical thermal state of an Ising appears to be difficult, e.g. in Ref. \cite{Yung} the authors provide an algorithm which does so but is exponential in the square root of the system size (see also \cite{DVdN}).  

To conclude this section, we study the possibility to use the data provided by the quantum experiments in order to construct a bound for the error on the estimated partition function.  Our motivation is that, possibly, the a posteriori error analysis might be finer than the error bounds provided by plain error propagation. To simplify the discussion, we will again restrict ourselves to one-step protocols. Extension to the general case is straightforward. Let us expand the quantity $A(i \beta)$ as
\beq
\begin{array}{ll}
A(i \beta)=&\sum_{j=0}^{2N}  
(\Re w^{(N)}( i \beta- \alpha_j)\\
&+ i \; \Im w^{(N)}( i \beta- \alpha_j))(\Re A(\alpha_j) + i \; \Im A(\alpha_j)),
\end{array}
\eeq
and focus on, say,  

\beq
A_{RR}( i \beta) \equiv \sum_{j=0}^{2N} 
\Re w^{(N)}( i \beta- \alpha_j) \; \Re A(\alpha_j).
\eeq
The three other bits of $A(i \beta)$ are treated likewise. As was shown in the previous section, each quantity $\Re A(\alpha_j)$ is obtained by measuring the polarisation of a qubit in a precise direction. Such a measurement process can be viewed as drawing a random variable whose outcomes are $\{+1,-1\}$, and whose mean value is the polarisation we are interested in. Let $M$ denote the number of Bernoulli trials involved in determining each probability amplitude, and let us denote $X_{j}(k)$ the outcome of the $k$-th trial used in the determination of $\Re A(\alpha_j)$. For fixed $j$, the random variables $X_{j}(k)$ have the same distribution for all $k$, characterised by $\text{Prob}[X_{j}(k)=-1]=p_{j}$. 

Our estimate for $A_{RR}( i \beta)$ is
\beq
\widehat{A}_{RR}( i \beta)=\frac{1}{M} \sum_{k=1}^M  \sum_{j=0}^{2 N} \Re w^{(N)}( i \beta- \alpha_j)  \; X_{j}(k).
\eeq
If we assume there is no (uncontrolled) systematic error in the quantum experiments, then the true value of $A_{RR}( i \beta)$ is of course given by
\beq
A_{RR}( i \beta)=
\sum_{j=0}^{2 N} \Re w^{(N)}( i \beta- \alpha_j) (1-2p_{j}).
\eeq
Let $\mathsf{E}_2(\widehat{p}_j)$ and $\mathsf{E}_3(\widehat{p}_j)$ denote appropriate estimates for $\mean{\big( \Re A(\alpha_j)-X_j(k) \big)^2}$ and $\mean{|\Re A(\alpha_j)-X_j(k)|^3}$ respectively, constructed from an appropriate estimate $\widehat{p}_j$ for $p_j$.

With such estimates, we define two random variables as follows:
\bed
\widetilde{D}_M(\epsilon)=\frac{1}{\sqrt{M}}
\frac{\sum_{j=0}^{2N} |\Re w^{(N)}( i \beta- \alpha_j)|^3 \big( \mathsf{E}_3(\widehat{p}_j)+8 \epsilon_j \big)}
{ \big(\sum_{j=0}^{2N} |\Re w^{(N)}( i \beta- \alpha_j)|^2 \big( \mathsf{E}_2(\widehat{p}_j)-4 \epsilon_j \big)  \big)^{3/2}},
\eed
\bed
\widetilde{\lambda}_M(\epsilon)=
\frac{\sqrt{M}}
{\sqrt{ \sum_{j=0}^{2N} |\Re w^{(N)}( i \beta- \alpha_j)|^2 \big( \mathsf{E}_2(\widehat{p}_j)+4 \epsilon_j \big)}}
\eed
where the deviations $\epsilon_j$ are of the form 
\bed
\epsilon_j=\frac{1}{4+s}  \mathsf{E}_2(\widehat{p}_j).
\eed
In this definition, $s$ is a parameter we are free to choose at our convenience.

The following central limit theorem holds for the statistics of errors:

\begin{theorem}[Central limit]\label{thm:error-stat} 
Let $\mathcal{F}_*$ denote the cumulative distribution of a zero-mean, unit-variance Gaussian probability distribution, and let $\Delta$ denote some strictly positive real number. The (composite) \emph{random variable}
\bed
\mathcal{L}(\{X_j(k)\}) \equiv
\big[
 1-2 \mathcal{F}_*(-\widetilde{\lambda}_M(\epsilon) \Delta)-1.12 \; \widetilde{D}_M(\epsilon)
 \big]
\eed
takes a finite value and lower bounds the \emph{quantity} $\textrm{Prob} \big[   |\widehat{A}_{ RR}(i \beta)-A_{RR}(i \beta)| < \Delta \big] $ with probability at least
\bed
\mathcal{P}(\{\epsilon_j \},M,N) \equiv \prod_{j=0}^{2N} \big( 1-2 e^{-\epsilon_j^2 M} \big)-\prod_{j=0}^{2N} \big( p_j^M+(1-p_j)^M  \big).
\eed 
\end{theorem}

The proof of this result builds on the Berry-Ess\'een theorem \cite{BE} and is given in Appendix \ref{thm:proof-central}. Interestingly, the only essential ingredient involved in this proof is the fact that we are trying to estimate a quantity (here a piece of a partition function) as a finite linear combination of Bernoulli random variables. For that reason, this proof  and a similar central-limit theorem are equally valid for \emph{any} quantum algorithm that aims at approximating a quantity $Q$ by an estimate of the form $\sum_{y} \Gamma_y X_y$, where each $X_y$ is a Bernoulli random variable. In particular, our analysis carries through to the algorithm proposed in Ref.\cite{AJL} to compute the Jones polynomial at non-trivial values of its parameter.

This result is interesting in that it actually allows to estimate with tunable statistical confidence and \emph{a posteriori}, i.e. after the quantum experiment is performed, the discrepancy between our estimate and the value we are trying to estimate.

\section{Link Invariants}\label{sect:ki}

Prior work \cite{Tutte} has provided polynomial time quantum algorithms for the Tutte polynomial including the calculation of the Jones polynomial at the specific values considered here.  In this section we note that  in fact these link invariants can be estimated with repeated application of \emph{constant} depth quantum circuits.

There exists several well-established connections between knot theory and statistical mechanics \cite{Wu}. One of them is the following. For every knot it is possible to construct a graph such that the partition function of a Potts model defined on that graph is a link invariant for certain (imaginary) temperatures. This invariant turns out to be the Jones polynomial evaluated at specific values, modulo a known calculable factor. As the quantum algorithm for computing partition functions described in Section \ref{sect:ctpf} is efficient for imaginary temperatures, it follows that it may also be used to distinguish among different link, when the associated statistical model only involves nearest neighbour interactions. In this section we outline the method to compute the statistical-mechanics knot invariant for any given link. We also compute these invariants for some primary linka with few crossings for which the Potts model involves only a few sites and is within reach of current technology.

\begin{figure}[ht]
\begin{minipage}[b]{0.5\linewidth}
\centering
\includegraphics[scale=0.05]{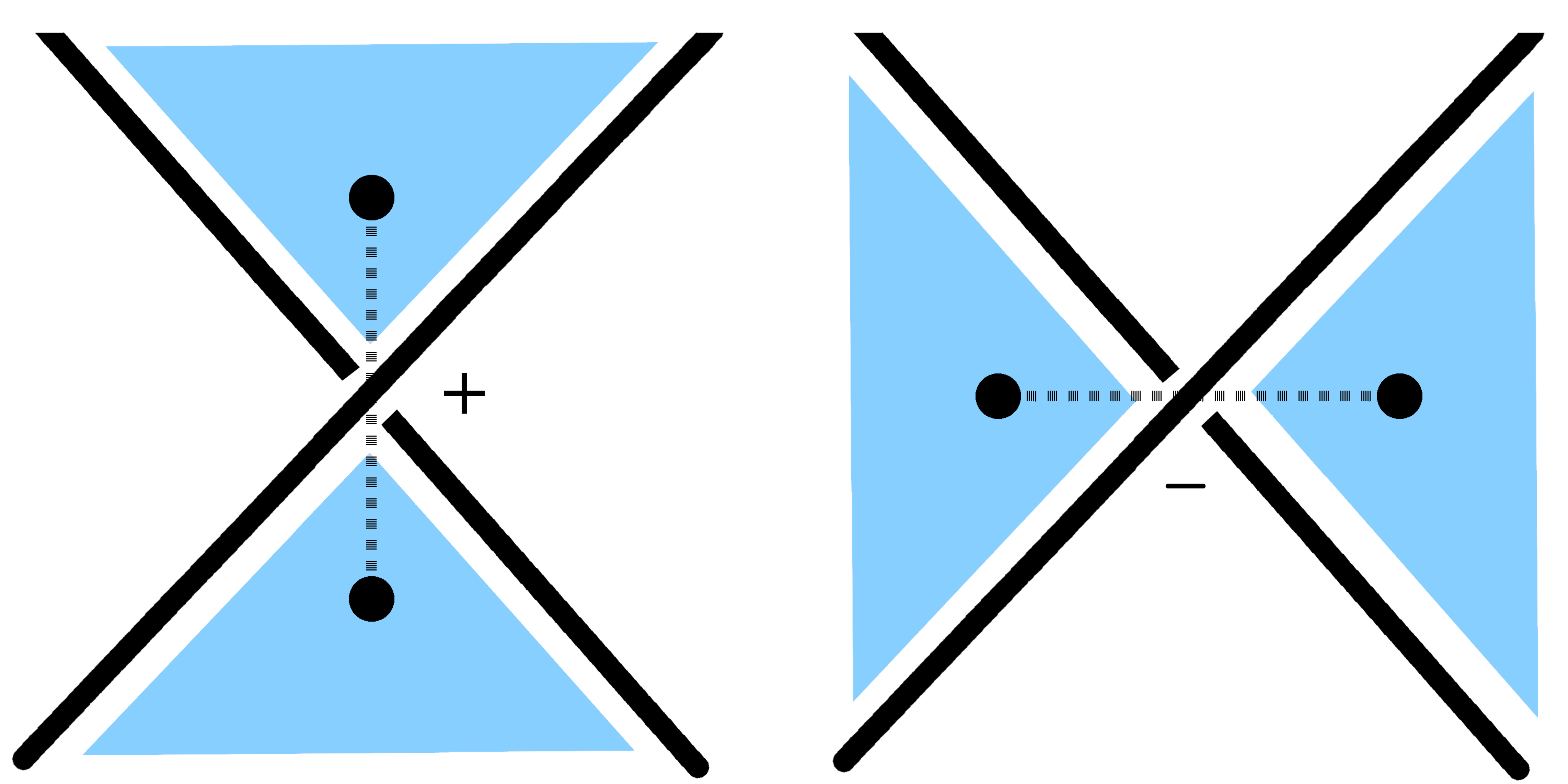}
\end{minipage}
\hspace{0.5cm}
\begin{minipage}[b]{0.5\linewidth}
\centering
\includegraphics[scale=0.18]{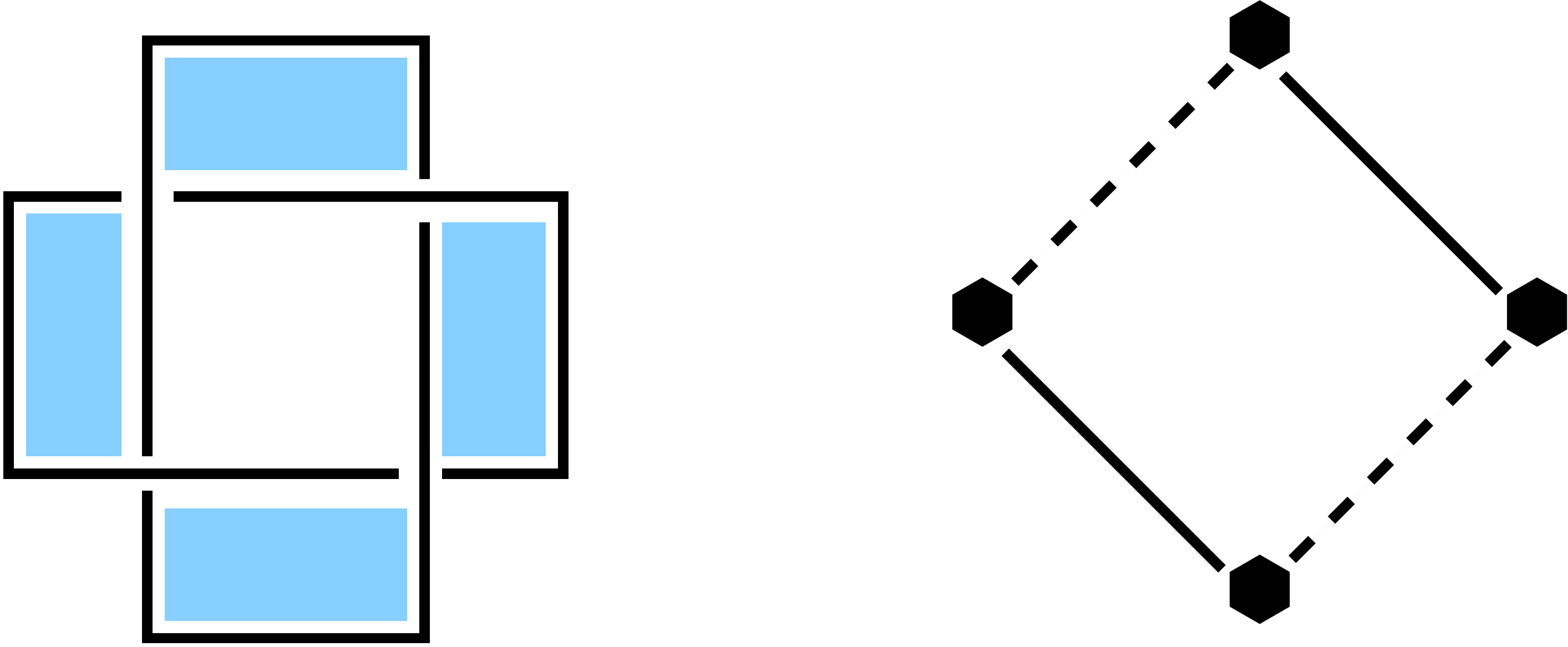}
\end{minipage}
\caption{Convention for determining the sign of the edge coupling assigned to each crossing (top). Example of the lattice obtained following the procedure outlined in the text for one of the possible shadings (bottom). Plain lines represent, say, postive couplings, while dashed lines represent negative couplings.}\label{convention}
\end{figure}
      
Let us start with a brief reminder on a recipe to construct statistical mechanical invariants, given a single component knot or a multicomponent link. We consider the planar projection of a given knot and shade the regions of the diagram in an alternating way such that there are no adjacent shaded regions (there are two ways to do this for any knot). We associate a lattice with vertices $\mathcal{V}$ and signed edges $\mathcal{E}$, $\Lambda=\left(E,V\right)$ to the diagram in the following way. Every shaded region of the diagram will be a vertex of $\Lambda$ and every crossing of the diagram that separates two shaded regions will be an edge linking the two vertices associated with those regions. The sign for the coupling of the edge is determined by the convention in Fig \ref{convention}. For every edge $i \in \mathcal{E}$ we associate a weight $\mathcal{W}^{\pm}_i\left(\sigma,\sigma '\right)$, where $\sigma,\sigma '$ are $q$-valued spins located at the vertices joined by the edge. Let us define a partition function given a set of weights $\mathcal{W}_i$ on $L$,

\begin{equation}
Z_L=\sum_{\left\lbrace \sigma\right\rbrace } \prod_{i\in E} \mathcal{W}_i\,,
\end{equation}\label{partitionknot}
where the sum is over all possible configurations of the spins on the vertices.

\begin{figure}
$
\begin{array}{ccc}
\includegraphics[width=0.9\columnwidth]{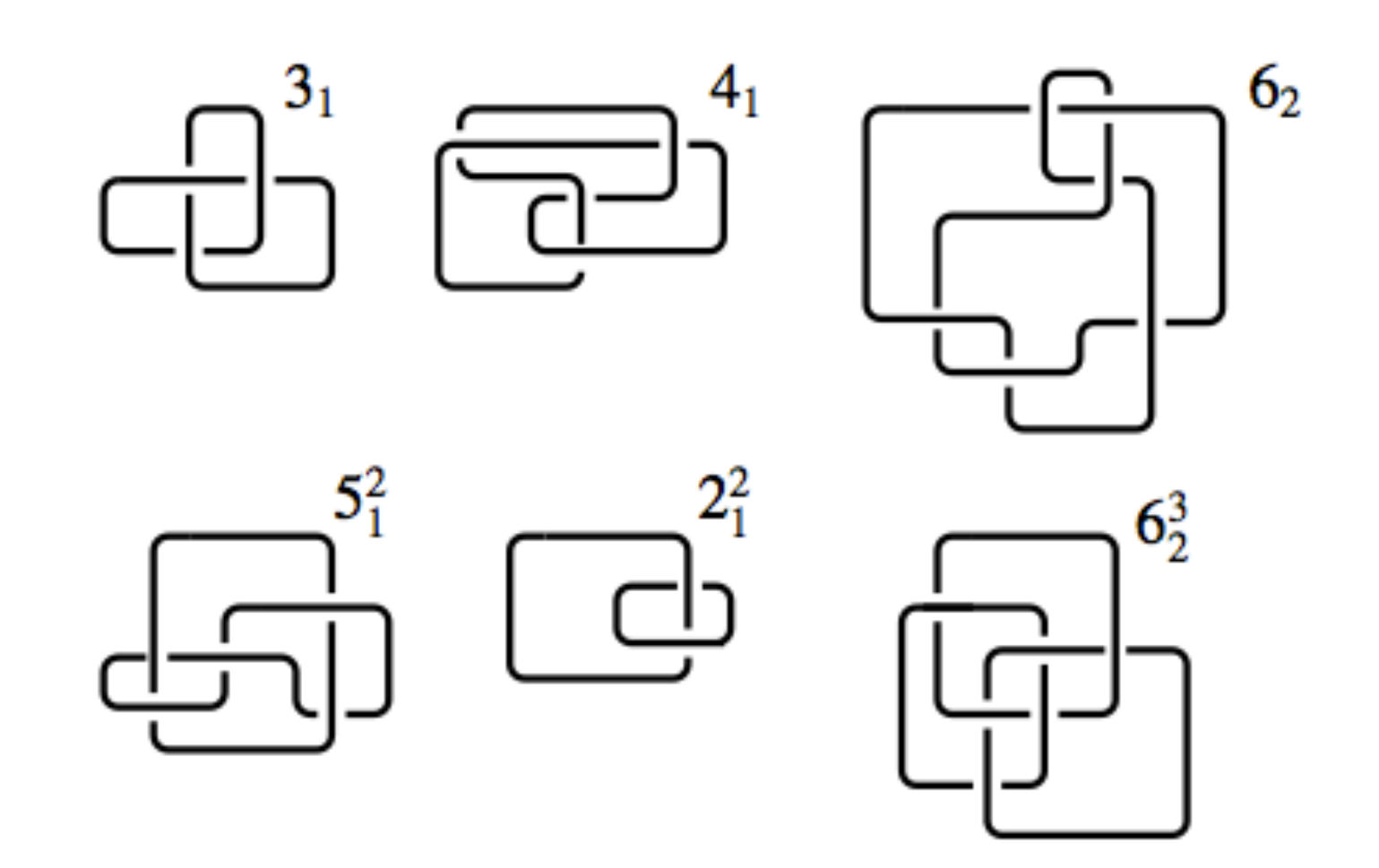}\\
\includegraphics[width=0.9\columnwidth]{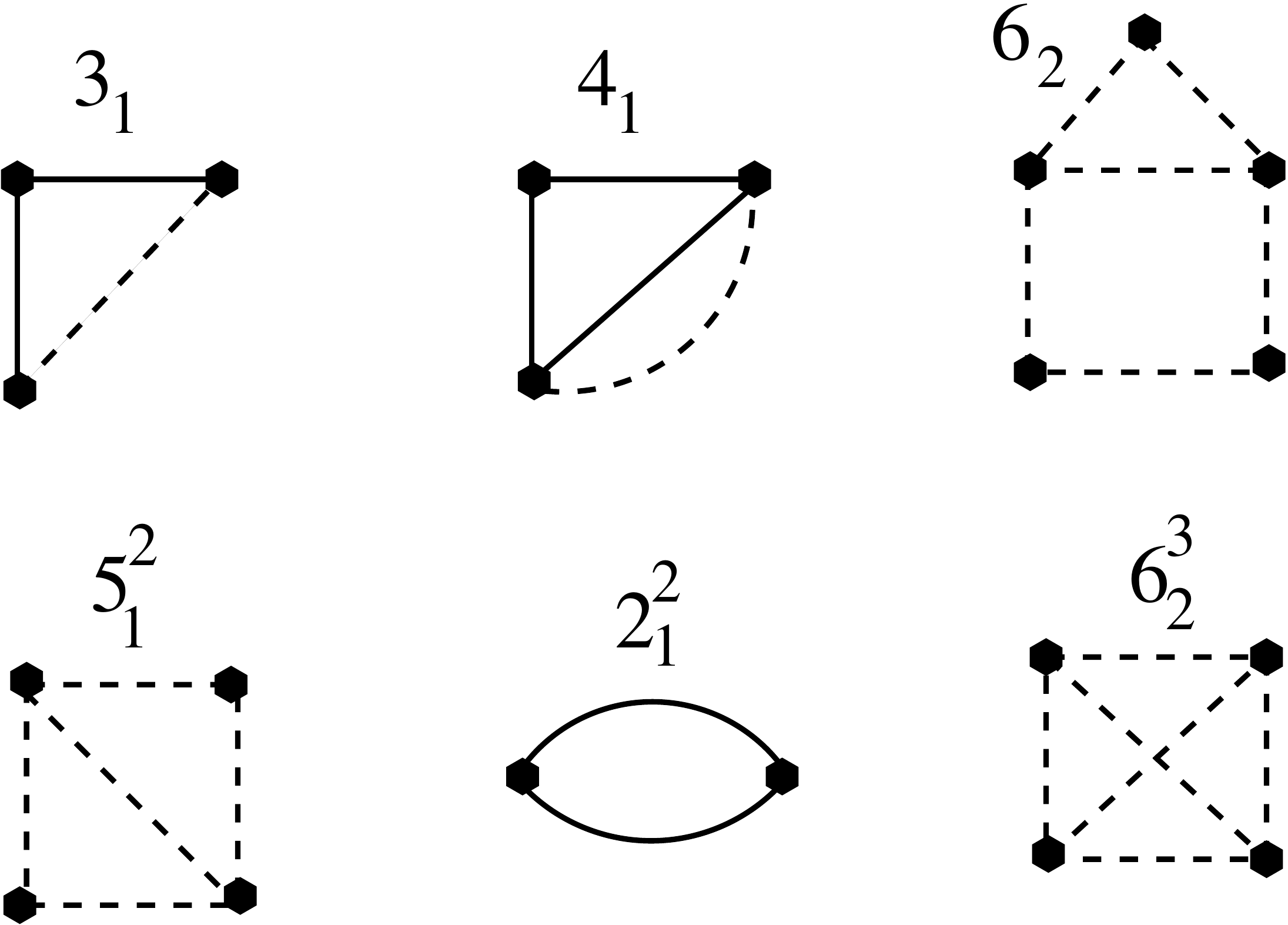}
\end{array}
$
\label{sixknots}
\caption{Planar diagrams for some examples of primary knots and links (top) and the associated partition functions (bottom). Of the two possible graphs for each knot (one for each choice of shading) we have chosen the less trivial one. All the knots lead to statistical mechanics models with nearest neighbor interactions except for the Borromean ring, $6^3_2$.}
\end{figure}

$Z_L$ is invariant under ambient isotopy provided the weights $\mathcal{W}_i$ satisfy certain conditions, the derivation of which is discussed in \cite{Wu}. It has been proven that the choice $\mathcal{W}^{\pm}_i=exp\left(\pm \beta \delta_{\sigma ,\sigma '}\right)$ where $\sigma=1, \ldots, q$ is compatible with these conditions if

\begin{equation}
\beta =\cosh^{-1}\left(\frac{q-2}{2}\right)\,,
\end{equation}
holds. In particular, the Potts partition function $Z_L$ for $q=1,2$,and $3$ at temperatures $\beta= i2\pi/3, i\pi/2$, and $i\pi /3$ respectively is a knot invariant.   Note that the existence of a quantum algorithm to compute the link invariants for these complex temperatures was already pointed out in Ref. \cite{Lidar}

We have determined the lattices $L$ for six examples of knots and links (see figure \ref{sixknots}) and computed $Z_L$ for a Potts model defined on $L$ with $q=1,2$ and $3$ for the values of $\beta$ where the partition function is a knot invariant (see table \ref{partitionvalues}). The invariant corresponding to the value $q=1,2$ are actually trivial. The case where $q=3$ is more interesting. A classical algorithm to compute this invariant exists which works in a time that scales polynomially with the number of crossings  \cite{Jaeger}. In turn, using a generalisation to three-level systems of the scheme presented in Section \ref{sect:ctpf} allows to estimate the quantum invariant $Z_L$ in \emph{constant} time with an additive error that scales like $1/\sqrt{R}$ where $R$ is the number of repetitions of the experiment, now independent of the number of crossings.

\begin{table}[h]
\begin{center}
    \begin{tabular}{|c|c|c|c|}
    \hline
          & $ q=1$ & $q=2$ & $q=3$   \\ \hline     
     $3_1$ & $e^{i\frac{5\pi}{6}}$ & $4e^{i\frac{5\pi}{8}}$ &  $\frac{3}{2}\left(7\sqrt{3}-i\right)e^{i\frac{1\pi}{4}}$    \\ \hline
     $4_1$ & $-e^{i\frac{\pi}{3}}$ & $4$ &  $-\frac{15}{2}\left(1-\sqrt{3}i\right)$  \\ \hline     
       $6_2$  & $-1$ & $-8e^{i\frac{\pi}{4}}$&$3\left(15-22i\right)$  \\ \hline   
      $5^2_1$  & $-e^{i\frac{5\pi}{6}} $ & $8e^{i\frac{3\pi}{8}}$ & $\frac{3}{2}\left(9\sqrt{3}+29i\right)e^{i\frac{3\pi}{4}}$  \\ \hline      
       $2^2_1$  & $-e^{i\frac{2\pi}{3}}$ & $0$ & $\frac{3}{2}\left(3+\sqrt{3}i\right)$ \\ \hline   
        $6^3_2$  & $-1$   & $8\sqrt{2}$& $-3\left(9\sqrt{3} + 4i\right)$  \\ \hline           
         \end{tabular}
         \caption{Knot invariants computed from the Potts model partition functions defined on the lattices in figure \ref{sixknots}. The temperatures at which the partition functions have been evaluated are given in the text.} 
        \label{partitionvalues}
\end{center}
      \end{table}

\section{Computational Power of Classical Models}

The analysis presented in Section \ref{sec:anacon} demonstrates how one can sample from a family of quantum circuits with fixed topology in $d$ dimensions to construct a partition function on a classical spin system with fixed topology in $d+1$ dimensions.  One could ask whether the reverse can be done, i.e. given a classical partition function can one then reconstruct the outcomes of a related quantum circuit for a family of coupling parameters?  Even more, is is possible that given the ability to compute the partition function of a suitably large classical system and for a suitable set of temperatures, one can reconstruct the outcome of measurements on arbitrary quantum computations of polynomial length in some fixed register input size?  This has been partially answered in Ref. \cite{vdN} where the authors show that the problem of computing the partition function of several classical spin models including the planar Ising model with magnetic fields all with \emph{complex} couplings is BQP-complete.  Such classical models do arise for some problems, e.g. the use of the Potts model with complex couplings to compute link invariants as discussed in Sec. \ref{sect:ki}.  In Ref. \cite{GeraciandLidarII} it was further shown that there is an equivalence between classical partition functions with \emph{real} couplings and quantum amplitudes for a certain certain class of quantum circuits known as Clifford circuits.  When this mapping exists the graph underling the classical theory is planar with no magnetic fields and can be estimated with a polynomial time classical algorithm \cite{Welsh}.  Also, deciding if a certain quantum circuit belongs to this equivalence class is classically easy.  These results are consistent with the Gottesman-Knill theorem which states that Clifford circuits admit classical simulations in polynomial time \cite{MikeandIke}.

It is desirable to obtain the connection between classical partition functions with real couplings and the output of any polynomial sized quantum circuit.  We do so in this section and also describe some applications:  one for investigating quantum phase transitions given the ability to compute classical partition functions, and another for computing partition functions given the ability to prepare and measure corner magnetisation on physically prepared classical thermal states.

\subsection{Estimating quantum computations from Ising model partition functions}
\label{sect:compu-power}
We show the following:

\begin{theorem}
\label{BQP}
Estimation of the partition function $Z(\beta)$ of a two dimensional ferromagnetic, consistent Ising model at inverse temperature $\beta$ on a square lattice of size $n\times m$ with $m=O(poly(n))$ with non uniform couplings and magnetic fields with additive error $\delta(n,m,\beta)<\exp(nm(49\beta-190)/2)$  is BQP-hard, i.e. it is at least as hard as simulating an arbitrary polynomial time quantum algorithm on $n$ qubits. By ferromagnetic we mean the couplings $J_{i,j}$ in Eq. \ref{eq:def-Ising} are all positive and by consistent the magnetic fields $h_i$ are all non-negative or all non-positive.  To simulate a quantum algorithm means to do the following:  
For a unitary $W$ built from a quantum circuit composed of $O(poly(n))$ one and two qubit gates on a length $n$ register provide an estimate of a complex scattering matrix element satisfying
\[
|\widehat{\bra{ +_x^{\otimes n}} W \ket{+_x^{\otimes n}}}-\bra{ +_x^{\otimes n}} W \ket{+_x^{\otimes n}}|\leq\frac{1}{O(poly(n))}
\]
with a probability that is exponentially in $n$ close to $1$.

\end{theorem}

\begin{proof}

The proof follows in several stages.  First we write an arbitrary polynomial sized quantum circuit in a convenient spatially translationally invariant form.  Then we show that the scattering matrix element is equivalent to a complex temperature classical Ising model on a square lattice.  Finally, we show that sampling the partition function over many real temperatures of a ferromagnetic Ising model, one can reconstruct the scattering matrix element.   

There are many possible equivalent quantum circuits which construct a given unitary.  We pick a quantum circuit with a coupling graph given by a one dimensional chain of qubits with open boundaries.  In order to perform a universal gate set, one needs a quantum circuit with gates either inhomogeneous in space or time or both.  We pick circuits which are homogenous in space only as they are simple to parameterize and it is pedagogically satisfying that each step in the quantum algorithm can be thought of as a Wick rotated transfer matrix generated by a spatially homogenous quantum Hamiltonian.   Several models exist for universal quantum computation which use 1D architectures with global interactions \cite{Benjamin, Fitzsimons}.  We pick a convenient one due to Raussendorf \cite{Raussendorf} which involves encoding quantum information in a 1D redundified data register, i.e. the data register is redundified in a second register which is spatially mirrored with respect to the first.  This method has the advantage that all gates acting on the system are translationally invariant and the initial state is translationally invariant, e.g. $\ket{+_x^{\otimes n}}$ .  The only requirements are uniform Ising interactions between nearest neighbours and global single qubit gates. Addressability is afforded by temporal addressing via judiciously chosen homogenous local operations.  Readout can be done again using global operations with the assistance of interspersed ancillary qubits or instead by using ancillary levels of each qubit \cite{PazSilva}.  The overall overhead incurred using global operations in this mirror encoded state is linear in $n$ \cite{Raussendorf}.

Consider a quantum register of an even number $n$ of logical qubits, encoded by a chain of $2n$ qubits. The encoding has a mirror structure, i.e. the wave function of the system is at all times of the form $\ket{\psi}_{1 \ldots n} \otimes \ket{\psi}_{2n \ldots n+1}$. The first ingredient in our proof of the BQP-hardness of the Ising model is the following lemma:

\begin{lemma}\label{thm:universal-gate-set} 

Let 
\beq
\begin{array}{lll}
&&\sigma^{\alpha}_{\rm{tot}}(\theta)=\prod_{j=1}^{2 n} e^{i \frac{\theta}{2} \sigma^{\alpha}_j},  \hspace{0.2cm} \alpha=x,y,z, \hspace{0.3cm}\\
&&\mathsf{CP}_{\rm{tot}}=\prod_{j=1}^{2n-1} \mathsf{CP}_{j,j+1}, \hspace{0.3cm}
\mathsf{Had}_{\rm{tot}}=\prod_{j=1}^{2 n} \mathsf{Had}_j,
\end{array}
\eeq

denote a set of translationally invariant (global) operations, where $\mathsf{Had}=e^{i\frac{\pi}{2\sqrt{2}}(\sigma^x_j+\sigma^z_j)}$ denotes a single qubit Hadamard gate and  $\mathsf{CP}=e^{i\pi \ket{11}\bra{11}}$ the controlled phase gate. The subset 
\beq\label{eq:uni-gate}
\mathfrak{G}= \{ \mathsf{CP}_{\rm{tot}}, \sigma^z_{\rm{tot}}(\pi/8),\mathsf{Had}_{\rm{tot}} \}
\eeq
 is universal for quantum computation.
\end{lemma}

\begin{proof}
This is proved in Appendix \ref{thm:universal-gate-set-proof}.
\end{proof}

This lemma implies that for any $\epsilon >0$, there exists a sequence of operators $\{ \mathcal{L}_t \in \mathfrak{G}: t=0 \ldots m-1 \}$, such that
\beq
|\widehat{\bra{ +_x^{\otimes n}} W \ket{+_x^{\otimes n}}}-\bra{ +_x^{\otimes 2 n}} \prod_{t=0}^{m-1} \mathcal{L}_t \ket{+_x^{\otimes 2 n}}|\leq \epsilon,
\eeq
where $m=O(\text{poly}(\log\frac{1}{\epsilon},n))$. Let $\sigma_{\text{tot}}$ label classical configurations for the $2n$-qubit chain (element of the computational basis). The action of $\sigma^z_{\text{tot}}(\pi/4)$ and $\mathsf{CP}_{\text{tot}}$ (up to a global phase) can be expressed as
\beq\label{eq:Ising-form-sigma-cp}
\begin{array}{lll}
\sigma^z_{\text{tot}}(\pi/8) \ket{\sigma_{\text{tot}}}&=& e^{i \frac{\pi}{16} \sum_{k=1}^n \sigma_k} \ket{\sigma_{\text{tot}}}, \hspace{0.3cm}\\
\mathsf{CP}_{\text{tot}} \ket{\sigma_{\text{tot}}}&=& 
e^{i \frac{\pi}{4} \big(\sum_{k=1}^{2n-1} (\sigma_k+\sigma_{k+1})+\sum_{k=1}^{2n-1} \sigma_k \sigma_{k+1} \big)}\\
&& \ket{\sigma_{\text{tot}}},
\end{array}
\eeq
while the matrix elements of a Hadamard gate (up to a global phase) read 
\beq\label{eq:Ising-form-had}
\bra{\sigma} \mathsf{Had} \ket{\sigma'}=\frac{1}{\sqrt{2}} e^{i \frac{\pi}{4}(\sigma+\sigma')} e^{i \frac{\pi}{4} \sigma \sigma'}.
\eeq

These expressions will help us to express the quantum amplitude $\bra{ +_x^{\otimes 2 n}} \prod_{t=0}^{m-1} \mathcal{L}_t \ket{+_x^{\otimes 2 n}}$ as an Ising partition function. It is convenient to introduce the following class of operators:
\bed
\mathcal{T}_s = 
\big( \mathsf{CP}_{\text{tot}} \big)^{e_0(s)}
\big( \sigma^z_{\text{tot}}(\pi/8) \big)^{e_1(s)}
2^n\mathsf{Had}^{1-\delta_{s,0}}_{\text{tot}},
\eed
where the exponents $e_0(s)$ and $e_1(s)$ take values in $\{0,1\}$. Up to constant factors, it is clear that the operators $\sigma^z_{\text{tot}}(\pi/8)$, $\mathsf{Had}_{\text{tot}}$ and $\mathsf{CP}_{\text{tot}}$ can each be expressed either as a single $\mathcal{T}$-type operator or as a product of at most 2 $\mathcal{T}$ operators. Consequently, we can write
\beq\label{eq:alt-expr-q-amplitude}
\bra{ +_x^{\otimes 2 n}} \prod_{t=0}^{m-1} \mathcal{L}_t \ket{+_x^{\otimes 2 n}}=
\frac{1}{2^{n M}} \bra{ +_x^{\otimes 2 n}} \prod_{s=0}^{M-1} \mathcal{T}_s \ket{+_x^{\otimes 2 n}},
\eeq
where $M\geq 1$.  If $M=1$ then the overlap is: $\bra{ +_x^{\otimes 2 n}} \prod_{t=0}^{m-1} \mathcal{L}_t \ket{+_x^{\otimes 2 n}}=2^{-n}Z_{1D}(\frac{i\pi}{16})$ where $Z_{1D}$ is the partition function for a classical Ising model in 1D with magnetic fields.  Since one dimensional Ising models are exactly solvable for any temperature, including complex temperatures, then so is the overlap.  Non exact estimations of scattering matrix element occur for $M>1$.    Since each layer operator $\mathcal{L}_t$ can be expressed as a product of at most two such operators $\mathcal{T}_s$, we see that $M$ is polynomial in $n$ (since we assume that $W$ is a polynomial depth quantum circuit). This last form of the quantum scattering amplitude, together with the identities  (\ref{eq:Ising-form-sigma-cp}, \ref{eq:Ising-form-had}) allow to express the quantum scattering amplitude as the partition function of an Ising model at imaginary temperature. Up to a global irrelevant phase, we have 
\beq\label{eq:q-amp-poly-form}
\bra{ +_x^{\otimes 2 n}} \prod_{t=0}^{m-1} \mathcal{L}_t \ket{+_x^{\otimes 2 n}}=\frac{1}{2^{n(M+2)}} \sum_{ \{ \sigma \}} e^{-\frac{i \pi }{16} H(\sigma)},
\eeq
where $H(\sigma)$ denotes the Hamiltonian of the form (\ref{eq:def-Ising}), defined on a square $(2n) \times M$ lattice. Simple inspection shows that all couplings (resp. fields) appearing in this Hamiltonian are positive integers, whose magnitude do not exceed 4 (resp. 17).

Let us now assume we are provided with the following resource:

$\mathtt{IsingEstimator}$:  Given an inverse temperature, $\beta$, and an inhomogeneous Ising Hamiltonian, defined on a two-dimensional square lattice of size $n_x \times n_y$,  a device provides an estimate $\widehat{Z}(\beta)$ for the partition function, $Z(\beta)$, that satisfies 
\beq
\textrm{Prob}[|\widehat{Z}(\beta)-Z(\beta)| \leq \epsilon \; \delta(n_x,n_y,\beta)]\geq \frac{3}{4},
\eeq
in a time that is polynomial in $n_x,n_y,\beta, 1/\epsilon$.

Our goal now is to study how we could design the function $\delta$ so that this resource allows for an efficient estimation of scattering amplitudes of quantum circuits. Since all magnetic fields and couplings appearing in the definition of the classical Hamiltonian associated with a quantum circuit are integers, the r.h.s of (\ref{eq:q-amp-poly-form}) can certainly be written as 
\bed
\frac{1}{2^{n(M+2)}} \sum_{ \{ \sigma \}} e^{-\frac{i \pi }{16} H(\{\sigma\})}= \frac{1}{2^{n(M+2)}} \sum_{k=-M'}^{+M'} c_k \; e^{i k \pi/16},
\eed
for some coefficients $c_k$. The value of the integer $M'$ is at most $\max_{\sigma} H(\sigma)$. The r.h.s. of the last equation can equivalently be written as
\bed
\frac{1}{2^{n(M+2)}} \sum_{ \{ \sigma \}} e^{-\frac{i \pi }{16} H(\sigma)}=\frac{e^{-iM' \pi/16}}{2^{n(M+2)}} \mathcal{P}(e^{i \pi/16}),  
\eed
where $\mathcal{P}$ is a degree-$(2 M')$ polynomial. For all $\beta \geq 0$, our resource allows to compute an estimate $\widehat{\mathcal{P}}(e^{-\beta}) \equiv e^{-\beta M'} \widehat{Z}(\beta)$ for $\mathcal{P}(e^{-\beta})$ that obeys $|\widehat{\mathcal{P}}(e^{-\beta})-\mathcal{P}(e^{-\beta})| \leq \epsilon e^{-\beta M'} \delta(2n,M,\beta)$. Using a Lagrange polynomial interpolation based on $K$ points $\{ \big(e^{-\beta_j},\widehat{\mathcal{P}} (e^{-\beta_j}) \big), j=0 \ldots K-1 \}$ ($K \geq 2M'+1$), we re-construct the polynomial $\mathcal{P}$ as 
\bed
\begin{array}{lll}
\widehat{\mathcal{P}}(z)&=& \sum_{j=0}^{K} \widehat{\mathcal{P}}(e^{-\beta_j}) \ell_j(z), \hspace{0.4cm} \\\ell_j(z)
&=&\prod_{k \neq j} \frac{z-e^{-\beta_k}}{e^{-\beta_j}-e^{-\beta_k}}, \hspace{0.2cm} z \in \mathbb{C}.
\end{array}
\eed  
This reconstructed polynomial is in turn used to estimate our quantum amplitude as $\widehat{\bra{ +_x^{\otimes 2 n}} \prod_{t=0}^{m-1} \mathcal{L}_t \ket{+_x^{\otimes 2 n}}}= \frac{e^{-iM' \pi/16}}{2^{n(M+2)}} \widehat{\mathcal{P}}(e^{i \pi/16})$. The error over this estimate can be bounded as
\bed
\begin{array}{ll}
&|\bra{ +_x^{\otimes 2 n}} \prod_{t=0}^{m-1} \mathcal{L}_t \ket{+_x^{\otimes 2 n}}-\widehat{\bra{ +_x^{\otimes 2 n}} \prod_{t=0}^{m-1} \mathcal{L}_t \ket{+_x^{\otimes 2 n}}}|\\
& \leq
\frac{1}{2^{n(M+2)}} \sum_{j=0}^{K-1} \epsilon \;  \delta(2n,M,\beta_j) e^{-M' \beta_j} |\ell_j(e^{i \pi/16})|.
\end{array}
\eed
It would be desirable to pick the integer $K$ and the temperatures $\beta_j$ in such a way that the r.h.s. of this last inequality is minimised. Presumably, calculus of variations might make this task doable. We have proceeded in a simpler way and made the choice 
\bed
e^{-\beta_j}=j/K, j=0 \ldots K-1.
\eed 
Then,
\bed
|\ell_j(e^{i \pi/8})|=\frac{\prod_{k \neq j} |K e^{i \pi/8}-k|}{\prod_{k \neq j} |j-k|}.
\eed
A closed form for the denominator on the r.h.s. of this expression can be easily worked out: 
\bed
|\prod_{k \neq j} (j-k)|= \prod_{k=0}^{j-1} (j-k) \times \prod_{k=j+1}^{K-1} (k-j)= j! \; (K-j-1)!
\eed
For the numerator, we observe that
\bed
\begin{array}{lll}
\prod_{k \neq j} |K e^{i \pi/16}-k|&=& \frac{K^K}{|K e^{i \pi/16}-j|}\\
&&\times \prod_{k=0}^{K-1} |e^{i \pi/16}-k/K|\\
&=&\frac{K^K}{|K e^{i \pi/16}-j|} \exp \big[\sum_{k=0}^{K-1}\\
&& \ln \sqrt{(\cos \frac{\pi}{16}-\frac{k}{K})^2+\sin^2 \frac{\pi}{16}}\big].
\end{array}
\eed
The argument of the last exponential is:
\bed
\begin{array}{ll}
&\frac{1}{2} \sum_{k=0}^{K-1} \ln \big[(\cos \frac{\pi}{16}-\frac{k}{K})^2+\sin^2 \frac{\pi}{16} \big] \times \frac{K}{K} \\
&<\frac{K}{2} \int_{0}^{1} \ln \big[ (\cos \frac{\pi}{16}-x)^2+\sin^2 \frac{\pi}{16} \big] \; dx\\
& < -0.744 K
\end{array}
\eed
Plugging these results in our bound for the error on the quantum amplitude, we find that, in the limit of large $K$,
\bed
\begin{array}{ll}
&|\bra{ +_x^{\otimes 2 n}} \prod_{t=0}^{m-1} \mathcal{L}_t \ket{+_x^{\otimes 2 n}}-
\widehat{\bra{ +_x^{\otimes 2 n}} \prod_{t=0}^{m-1} \mathcal{L}_t \ket{+_x^{\otimes 2 n}}}| \\
&<
\frac{\epsilon}{2^{n(M+2)}} \sum_{j=0}^{K-1} 
\frac{\delta(2n,M,\beta_j) (j/K)^{M'} K^K e^{-0.744 K} }{|K e^{i \pi/16}-j| j! (K-j-1)!}.
\end{array}
\eed

Considering the case where $K=2M'+1$, having an error 
\bed
\begin{array}{lll}
\delta(2n,M,\beta) &\leq& 
\sin \frac{\pi}{16} \; e^{(\beta+1.488)M'} 2^{n(M+2)} \\
&&\times\Gamma((2M'+1) e^{-\beta}+1) \Gamma((2M'+1) \\
&&\times~(1-e^{-\beta}))
(2M'+1)^{-2M'}
\label{deltabound}
\end{array}
\eed
is therefore sufficient for efficient reconstruction of quantum amplitudes. 
Note that the maximum energies from vertical and horizontal bonds in the lattice is $4(2n(M-1)+(2n-1)M)$ and the maximum local field energy is $2n(M-2)8+2n8+2nM+4M((2n-2)2+2)$.  Then we have the bound:  $M'\leq 50nM-12M-24n$. 
To work out how large $\delta(2n,M,\beta)$ is compared to the partition function, we can compute the the needed accuracy for a function of the error $\delta(2n,M,\beta')$ (we use a scaled temperature $\beta'=\beta/\ln 2$ to simplify the expression)
\[
\begin{array}{lll}
f_{\rm error}(2n,M,\beta')&\equiv& -\frac{\ln\delta(2n,M,\beta')}{\beta' 2nM}\\
&>& -\frac{M'\ln 2}{2nM}\\
&+&\frac{M'\ln 2 (0.7387+2\log_2(2M'+1)-2\log_2M')}{\beta' n M}.
\end{array}
\]
For large system sizes,
\[
f_{\rm error}(2n\gg 1,M\gg 1,\beta')>-25\ln 2+\frac{50\ln 2(2.7387)}{\beta'}.
\] 
Finally we get a bound for the permissible additive error in the estimation:
\begin{equation}
\delta(2n,M,\beta)<\exp(nM(49\beta-190))
\label{CtoQerrorbound}
\end{equation}

Writing $n_x=2n,n_y=M$, since estimating $Z(\beta)$ for a polynomial number (linear in $n_y$) of temperatures with additive error $\delta(n_x,n_y,\beta)$ on each provides the requisite estimate of the quantum scattering matrix element on a poly(n) sized quantum circuit, the complexity of the estimate of $Z(\beta)$ for an arbitrary temperature is BQP-hard.  This completes the proof of Theorem \ref{BQP}.
\end{proof}

We have found how much relative error we can tolerate in an estimation of a classical partition function and still accurately estimate quantum scattering amplitudes.  How does this compare to known accuracy of classical algorithms which provide estimates of these partition functions?    In Ref. \cite{JS} Jerrum and Sinclair construct a fully polynomial randomized approximation scheme (FPRAS) for computing the partition function of an arbitrary classical ferromagnetic Ising model that is consistent.  
Specifically they provide a classical algorithm that computes an estimate $\hat{Z}(\beta)$ of the partition function $Z(\beta)=\sum_{\{ \sigma \}}e^{-\beta H(\{ \sigma \})}$ for the ferromagnetic Hamiltonian $H(\{ \sigma \})$ on $N$ spins, with a multiplicative error $\epsilon$ and success probability
\[
{\rm Prob}\Big[|\hat{Z}(\beta)-Z(\beta)|\leq \epsilon Z(\beta)\Big]\geq\frac{3}{4}
\]
in a run time polynomial in $N,1/\epsilon$.  This probability of success can be boosted to $1-\delta$ in a number $\log(1/\delta)$ of repetitions \cite{JS}. Since the classical Hamiltonian in Eq. \ref{eq:q-amp-poly-form} is ferromagnetic, then when $\delta(2n,M,\beta)\geq Z(\beta)$,$\mathtt{IsingEstimator}$ is no more powerful than FPRAS.  In other words, if the tolerable error of $\mathtt{IsingEstimator}$ could be equal to or greater than $Z(\beta)$ for the relevant temperatures needed to reconstruct the scattering matrix element, then BQP-hard problems can be computed in polynomial time via FPRAS.  This is not expected to be the case so we almost certainly have the requirement that the inequality in Eq. \ref{deltabound} is $\delta(2n,M,\beta)< Z(\beta)$ over some significant range of temperatures and that it is smaller by an exponential in the problem size $M'$.  Note it 
is known that the problem of \emph{exactly} computing the partition function for even a ferromagnetic classical Ising model is $\#$P-complete \cite{JS}.  This complexity class is the same as that for counting the number of satisfying assignments of a Boolean function and counting optimal Traveling Salesman tours.  Approximating the partition function with multiplicative error for an anti-ferromagnetic Ising model on a square lattice is NP-hard and for the ferromagnetic model but with general fields is approximation preserving reducible to the complexity class $\#$BIS \cite{Goldberg}.  The latter is as hard as computing the number of independent sets, (an independent set is a set of vertices that does not contain both endpoints of any edge), in a bipartite graph which is thought to be of intermediate complexity between $\#$P and FPRAS.

\subsection{Ising models to compute quantum ground state overlaps}\label{sect:TM-spec}
We now consider an application of the mapping between classical partition functions and quantum scattering matrix applitudes:  measuring ground state wavefunction overlaps of quantum Hamiltonians.   It has been argued in Ref. \cite{Zanardi} that wave function overlaps, termed fidelity overlap, can be a good witness to quantum phase transitions when the ground states straddle a phase transition point.  In an ideal laboratory, this problem could be split in two: prepare two quantum registers in the desired states and measure the overlap using, for example, the protocols in Sec. \ref{sect:imp}. A possibility for the preparation step is to initialise the quantum system in the ground state $\ket{\Psi_0}$ of some simple hamiltonian $\hat{H}_0$, and to evolve this Hamiltonian to the target Hamiltonian $\hat{H}^{\star}$. A fundamental result of quantum mechanics, known as the adiabatic theorem, is that if the Hamiltonian is modified slowly enough, the state obtained at the end of the evolution will be very close to the true ground state $\ket{G}$ \cite{Amb-Reg}. Crucially, the time of the evolution need only grow \emph{polynomially} with the inverse of the minimum gap of the system, $\gamma$. 

The purpose of this section is to exhibit situations for which the adiabatic evolution need not be actually implemented. We are going to show that, in a precise sense, ``time can be replaced with space". Roughly speaking, we are going to show that, instead of performing measurements on a quantum system of a given size, say "$\mathsf{size}$" that has been \emph{evolved} for a time "$\mathsf{time}$", we can equivalently measure partition functions of classical Ising models prepared on a system of size $O(\mathsf{size} \times \mathsf{time})$. 

To make things precise, we will focus on the quantum transverse Ising model, described by the Hamiltonian\footnote{We use the hat on the operator to emphasise that this is a quantum Hamiltonian.}:
\beq\label{eq:q-Ising}
\hat{H}^{\star}=-h_{\perp} \sum_{i \in \Lambda} \sigma^x_i \; -  J \sum_{\langle i, j \rangle \\ \in E(\Lambda)} \sigma_i^z \; \sigma_j^z-h \sum_{i \in \Lambda} \sigma_i^z,
\eeq
where $\Lambda$ denotes some $d$-dimensional lattice, and $E(\Lambda)$ denotes the set of edges of $\Lambda$. We are going to view this Hamiltonian as a particular member of a family of time-dependent operators labelled by some time index, $t$. This family is 
\beq\label{eq:ham-family}
\hat{H}(t)=\hat{H}_0+\hat{H}_1(t), \hspace{0.5cm} t \in [0:T]
\eeq
where 
\beq
\begin{array}{lll}
\hat{H}_0&=&-h_{\perp} \sum_{i \in \Lambda} \sigma^x_i, \hspace{0.5cm} \\
\hat{H}_1(t)&=&-\frac{t}{T} J \hspace{-0.1cm} \sum_{\langle i, j \rangle \in E(\Lambda)} \sigma_i^z \; \sigma_j^z- \frac{t}{T} h \sum_{i \in \Lambda} \sigma_i^z.
\end{array}
\eeq
Without loss of generality, we will assume that $h_{\perp} >0$. In that case, $\ket{\Phi_0} \equiv \ket{+_x^{\otimes |\Lambda|}}$ is of course the (unique) ground state of $\hat{H}_0$. Evidently, $\hat{H}(T)=\hat{H}^\star$. The starting point of our construction is a discretisation of an adiabatic evolution

\begin{theorem}\label{thm:Ising-adiab-combined} 

Let $T$ satisfy the inequality
\beq\label{eq:adiab1}
T \geq T_*(\hat{H},\delta)= \frac{10^5}{\delta^2} \frac{\big( |h| \cdot |\Lambda|+ |J| \cdot |E(\Lambda)|\big)^3}{\gamma^4},
\eeq
where $\gamma=\text{min}_{t \in [0:T]} \text{gap} \; \hat{H}(t)$, where $\text{gap} \; \hat{H}(t)$ denotes the difference between the two lowest eigenvalues of $\hat{H}(t)$. Let $L$ denote a positive integer, and let us define the discretisation step as 
\beq\label{eq:def-tau}
\tau\equiv T/L.
\eeq 
The quantity by which the state $U_{L-1} U_{L-2} \ldots U_0 \ket{+_x^{\otimes |\Lambda|}}$ deviates from the true ground state $\ket{G}$ of $H^\star$ is at most 
\beq\label{eq:approx-final}
\begin{array}{lll}
\Delta &=& \delta+ T \sqrt{\frac{2 \big( |h| \cdot |\Lambda|+ |J_{\parallel}| \cdot |E(\Lambda)|\big)}{L}}\\
&&+ K L  \big( |h| \cdot |\Lambda|+ |J_{\parallel}| \cdot |E(\Lambda)|\big) \cdot |h_{\perp}| \cdot |\Lambda| \tau^2,
\end{array}
\eeq
where $K$ is some constant. Each unitary $U_k$ is defined as 
\beq
U_k=e^{-i \tau \hat{H}_0} e^{-i \tau \hat{H}_1(k\tau)}.
\eeq

\end{theorem}

This theorem, whose proof is given in Appendix \ref{app:proof-adiab-combined}, will help us to study  fidelity overlaps, 
\[
 f=\bra{\tilde{G}} G \rangle,
 \]
 where $\ket{G}$ is the ground state of $\hat{H}^{\star}$ and $\ket{\tilde{G}}$ is the ground state of some other Hamiltonian $\hat{\tilde{H}}^{\star}$.  For $T(T')$ and $L(L')$ large enough to build an approximation to $\ket{G}(\ket{\tilde{G}})$, $f$ can be replaced in good approximation with 
\beq\label{eq:overlap}
f \simeq \bra{+_x^{\otimes |\Lambda|}}
W^{\dagger}_{0} W^{\dagger}_{1} \ldots W^{\dagger}_{L'-1}\times~U_{L-1} U_{L-2} \ldots U_0 \;
\ket{+_x^{\otimes |\Lambda|}}.
\eeq
For the transverse Ising model exemplified here the fidelity estimate which gives witness to a quantum phase transition is for the the case $\ket{G}$ being the ground state of $\hat{H}^\star$ with couplings $h=0$ and $\ket{\tilde{G}}$ being the ground state of $\hat{\tilde{H}}^\star$ with couplings $h'=0$, $J'=J$, and $h'_{\perp}=h_{\perp}+\delta h_{\perp}$.  Near the critical point, $h_{\perp}=J$, there is a strong dip in the fidelity especially pronounced for $\delta_{\perp} h/J\sim 0.2$  \cite{Gu}.
\begin{figure}[ht]
\begin{centering}
\includegraphics[scale=0.14]{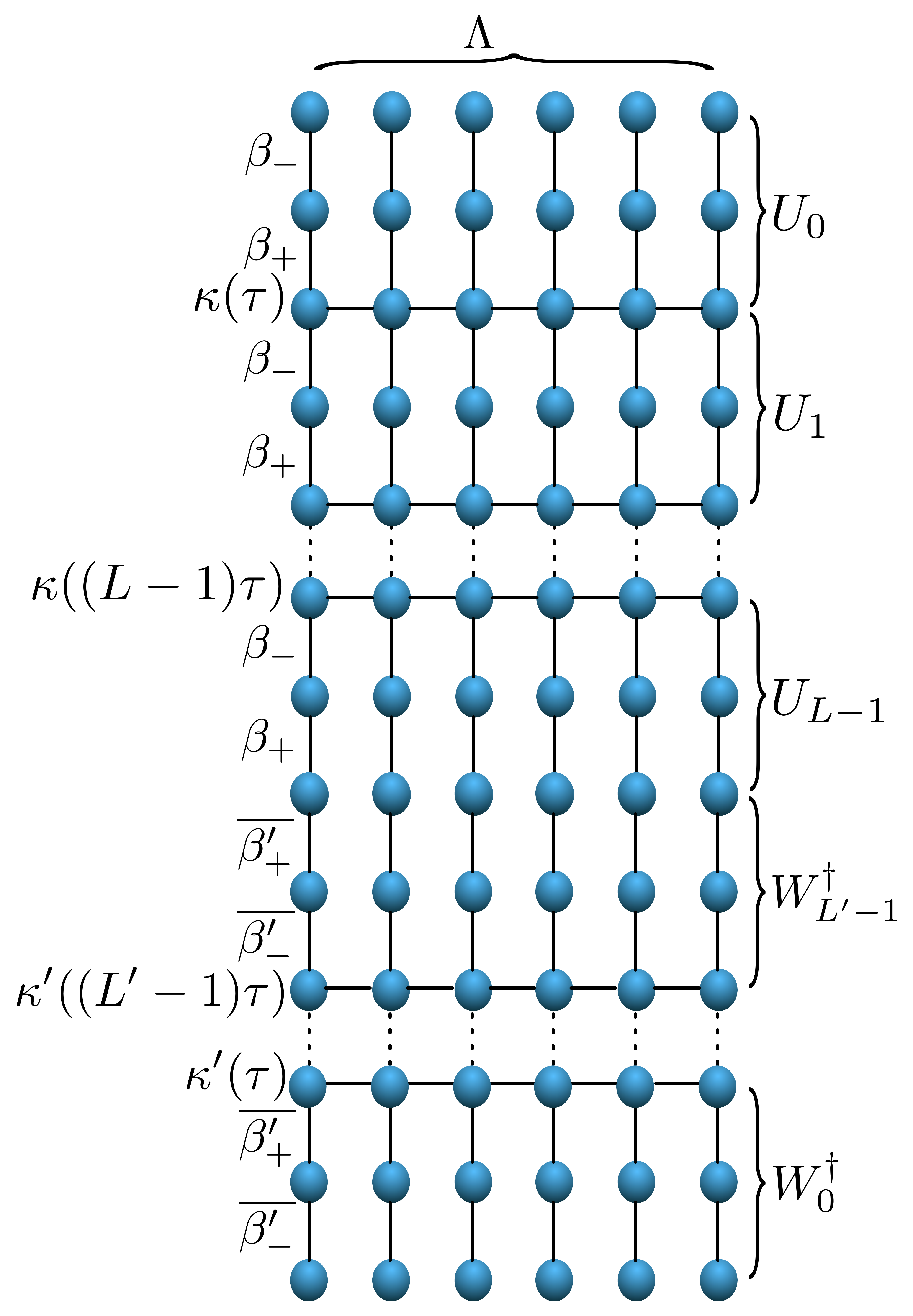}
\caption{Representation of the $d+1$ dimensional classical Ising spin lattice with couplings that encode information of the wavefunction overlap on a $d$ dimensional quantum spin lattice $\Lambda$.  Here the overlap is $\langle \tilde{\Psi}^{\star}| \Psi^\star\rangle$ which is an approximation to the fidelity $f=\langle \tilde{G}| G\rangle$, where $\ket{G}$ is the ground state of a Hamiltonian $\hat{H}$ and $\ket{\tilde{G}}$ is the ground state of $\hat{\tilde{H}}^\star$.  The sequence $U_{L-1}\ldots U_0$ provides for adiabatic evolution, in small time steps $\tau$, of a time dependent Hamiltonian $\hat{H}(t)$ from the product state $\ket{+_x^{\otimes |\Lambda|}}$ to $\ket{\Psi^{\star}}$ (which is an approximation to $\ket{G}$), and similarly for the sequence $W_{L'-1}\ldots W_0$, in steps $\tau'$, which builds an approximation $\ket{\tilde{\Psi}^{\star}}$ of $\ket{\tilde{G}}$ from $\ket{+_x^{\otimes |\Lambda|}}$.    Note that the number of gates $L$ and $L'$ to reach target ground states could differ as will the couplings generically.  Each gate is a composition of diagonal gates with dimensionless coupling $\kappa(t)$ and two non diagonal gates with dimensionless couplings $\beta_{\pm}$.  The temporal evolution of quantum gates can be represent on a classical spin lattice of one extra dimension with bond couplings as indicated on the left.  For the Hamiltonian in Eq. \ref{eq:ham-family} $\kappa(k\tau)$ means dimensionless row couplings $\beta k$ between nearest neighbour spins, local fields of strength $\beta k h/J$, and couplings $\beta_{\pm}$between rows. The parameters for adiabatic evolution to $\ket{\tilde{\Psi}^\star}$ are indicated with primes.
}\label{qsimsfig}
\end{centering}
\end{figure}

We are going to use a classical argument of quantum field theory \cite{Kogut}, in a simple form adapted to our purposes, and show that the overlap (\ref{eq:overlap}) can be expressed as a partition function for a $d+1$-dimensional many body system at finite (complex) temperature, described by a suitable \emph{classical} Ising Hamiltonian.  The operator $e^{-i \tau  \hat{H}_0}$ can be expressed as the transfer matrix of a \emph{classical} system, using the identity \cite{Frad-Suss}
\beq\label{eq:two-spin-tm}
T(\beta)=\sum_{\sigma \sigma'} e^{\beta \sigma \sigma'} \ket{\sigma} \bra{\sigma'}=e^{\beta} (\gras{1}+e^{-2\beta} \sigma^x). 
\eeq
Since on the other hand,
\bed
e^{-i \tau h_{\perp} \sigma^x}=\cos(\tau h_{\perp}) (\gras{1}- i \tan{(\tau h_{\perp})} \sigma^x),
\eed
it would be natural to make the identification $e^{-2 \beta}=-i \tan{\tau h_{\perp}}$, giving $\beta=i \frac{\pi}{4}-\frac{1}{2} \ln \tan{(\tau h_{\perp})}$, in order to relate the quantities to a classical model. Rather we are going to express the single-site unitary operator $e^{-i \tau h_{\perp} \sigma^x}$ in terms of \emph{two} operators $T$. For $\epsilon >0$, let us define $\beta_{\pm}(\epsilon)$ through
\beq\label{eq:cons:betapm}
e^{-2\beta_{\pm}(\epsilon)}=\mp i(1 \pm \epsilon).
\eeq
One checks that 
\beq
\begin{array}{ll}
T(\beta_+(\epsilon)) \; T(\beta_-(\epsilon))=&(2-\epsilon^2) e^{(\beta_{+}(\epsilon)+\beta_{-}(\epsilon))} \\&\times[\gras{1}-i \frac{2 \epsilon}{2-\epsilon^2} \sigma^x].
\end{array}
\eeq
This choice of using two transfer matrices is not strictly necessary but it guarantees that the amount by which $\beta_+$ and $\beta_-$ need deviate from the imaginary axis is small which makes the connection to the traditional classical to quantum mappings \cite{Kogut} more transparent.



So, for 
\beq\label{eq:eps-tau}
\frac{2 \epsilon}{2-\epsilon^2}=\tan(\tau h_{\perp}),
\eeq
we see that the operator $e^{-i \tau \hat{H}_0}$ can be expressed as a product of two \emph{classical} Ising transfer matrices: 
\bed
e^{-i \tau \hat{H}_0}= \Bigg[\sqrt{\frac{1-\epsilon^2}{\epsilon^4+4}}\Bigg]^{|\Lambda|} \prod_{x \in \Lambda} 
T_x(\beta_+(\epsilon)) \; \prod_{y \in \Lambda} T_y(\beta_-(\epsilon)).
\eed

This latter identity allows to express each operator $U_k$ in terms of classical Ising transfer matrices. Introducing closure relations and bearing in mind that the operator $\hat{H}_1(t)$ is diagonal in computational basis, the matrix elements of each operator $U_k$ can now be expressed as a sum over paths on three copies of the lattice $\Lambda$:
\bed
U_k= \sum_{\sigma(k)}
\sum_{\sigma(k+1)}
\sum_{\sigma(k+2)}
e^{\mathscr{L}_{\sigma(k),\sigma(k+1),\sigma(k+2)}}\ket{\sigma (k+2)}
\bra{\sigma(k)},
\eed
Here $\sigma(k)$ denotes a classical spin configurations over one copy of $\Lambda$, and the interaction $\mathcal{L}$, defined over a lattice $\Lambda \times \Lambda \times \Lambda$, is
\bed
\begin{array}{lll}
\mathcal{L}_{\sigma(k),\sigma(k+1),\sigma(k+2)}&=&
\beta_- \sum_{j \in \Lambda} \sigma_j(k) \sigma_j(k+1)\\
&+&\beta_+ \sum_{j \in \Lambda} \sigma_j(k+1) \sigma_j(k+2)\\
&-&i\frac{k\tau^2}{T}
\big(J  \sum_{\langle i, j \rangle \in E(\Lambda)} \sigma_i(k) \sigma_j(k) \\
&+& h \sum_{j \in \Lambda} \sigma_j(k) \big).
\end{array}
\eed
This interaction looks like a classical spin interaction with alternating complex couplings $\beta_+,\beta_-$ in the ``time' direction which transfers between different copies of the lattice $\Lambda$ and complex coupling within the lattice $\Lambda$.  We would like to be able to chose variable couplings along the ``space" and ``time" directions so we define a new interaction (assuming $J\neq 0$)
\bed
\begin{array}{lll}
\mathcal{H}_{\sigma(k),\sigma(k+1),\sigma(k+2)}&=&
\beta_- \sum_{j \in \Lambda} \sigma_j(k) \sigma_j(k+1)\\
&+&\beta_+ \sum_{j \in \Lambda} \sigma_j(k+1) \sigma_j(k+2)\\
&+&\beta k (\sum_{\langle i, j \rangle \in E(\Lambda)} \sigma_i(k) \sigma_j(k) \\
&+& \frac{h}{J} \sum_{j \in \Lambda} \sigma_j(k) \big).
\end{array}
\eed
A similar Hamiltonian can be written to represent evolution by gates $W_k^{\dagger}$:
\bed
\begin{array}{lll}
\mathcal{H'}_{\sigma(k),\sigma(k+1),\sigma(k+2)}&=&
\beta'_+ \sum_{j \in \Lambda} \sigma_j(k) \sigma_j(k+1)\\
&+&\beta'_- \sum_{j \in \Lambda} \sigma_j(k+1) \sigma_j(k+2)\\
&+&\beta' (L'+L-1-k)\\
& \times& (\sum_{\langle i, j \rangle \in E(\Lambda)} \sigma_i(k+3) \sigma_j(k+3) \\
&+& \frac{h'}{J'} \sum_{j \in \Lambda} \sigma_j(k+3) \big).
\end{array}
\eed
Now we can write a Hamiltonian on the \emph{enlarged} lattice $\hat{\Lambda}=\{1, \ldots, 2(L+L')+1 \} \times \Lambda$,
\bed
\begin{array}{lll}
-H(\{\sigma \})&=&\sum_{k=0}^{L-1} \mathcal{H}_{\sigma(2k+1),\sigma(2k+2),\sigma(2k+3)}\\
&+&\sum_{k=L}^{L'+L-1} \mathcal{H'}_{\sigma(2k+1),\sigma(2k+2),\sigma(2k+3)},
\end{array}
\eed
which takes exactly the form of a classical $(d+1)$-dimensional Ising Hamiltonian but with complex couplings.  The associated partition function depends on the vector of couplings $\vec{\beta}\equiv\{\beta_+,\beta_-,\beta'_+,\beta'_-,\beta,\beta'\}$
\[
Z(\vec{\beta})=\sum_{\{\sigma\}}e^{-H(\{\sigma\})}
\]
and is a sum over classical configurations defined over $\hat{\Lambda}$.
 
Substituting this expression in Eq.(\ref{eq:overlap}), we see that the fidelity overlap can be approximated by 
\beq\label{eq:corr-action}
f= \frac{1}{2^{|\Lambda |}}\Big[\sqrt{\frac{1-\epsilon^2}{\epsilon^4+4}}\Big]^{L |\Lambda|} \Big[\sqrt{\frac{1-\epsilon'^2}{\epsilon'^4+4}}\Big]^{L' |\Lambda|} \; Z(\vec{\beta}^\star).
\eeq
where, $\epsilon'$ is a solution to $2 \epsilon'/(2-\epsilon'^2)=\tan(\tau' h'_{\perp})$ appropriate for the quantum Hamiltonian $\hat{\tilde{H}}^\star$ and the vector of complex variables $\vec{\beta^\star}\equiv\{\beta^\star_+,\beta^\star_-,\beta'^\star_+,\beta'^\star_-,\beta^\star,\beta'^\star\}$ is
\begin{equation}
\begin{array}{lll}
\beta_{\pm}^{\star}&=&\pm\Bigg(\frac{i\pi}{4}+\frac{1}{2}\log\Big(1\pm \Big(\sqrt{3-\cos(2 \tau h_{\perp})}\\&&\csc(\tau h_{\perp})/\sqrt{2}-\cot(\tau h_{\perp})\Big)\Big)\Bigg)\\
\beta_{\pm}^{'\star}&=&\mp\Bigg(\frac{i\pi}{4}+\frac{1}{2}\log\Big(1\pm \Big(\sqrt{3-\cos(2 \tau' h'_{\perp})}\\&&\csc(\tau' h'_{\perp})/\sqrt{2}-\cot(\tau' h'_{\perp})\Big)\Big)\Bigg)\\
\beta^{\star}&=&\frac{iJT}{L^2}\\
\beta'^{\star}&=&-\frac{iJ'T'}{L'^2},
\label{betavariables}
\end{array}
\end{equation}
which were obtained by solving Eqs. \ref{eq:cons:betapm},\ref{eq:eps-tau}.  Note for $\tau h_{\perp}\ll 1$, $\beta_{\pm}^{\star} = \pm i \frac{\pi}{4}\pm \frac{\tau h_{\perp}}{2}- O(\tau^2 h_{\perp}^2).$ which is the statement that the analytic continuation is performed to nearly purely imaginary couplings strengths along the ``time" direction.

As described in Sec. \ref{sec:anacon} we can write the partition function in a power series in exponentials of the coupling parameters.  For simplicity we assume $J=J'=1, h'=h=0$ in which case:
\begin{equation}
\begin{array}{lll}
Z(\vec{\beta})&=&\sum_{g_1=-m_1}^{m_1}\sum_{g_2=-m_2}^{m_2}\sum_{g_3=-m_3}^{m_3}\sum_{g_4=-m_4}^{m_4}\\
&&\sum_{g_5=-m_5}^{m_5}\sum_{g_6=-m_6}^{m_6}c_{g_1,g_2,g_3,g_4,g_5,g_6}\\
&&\times~e^{\beta g_1}e^{\beta' g_2}e^{\beta_+ g_3}e^{\beta_- g_4}e^{\beta'_+ g_5}e^{\beta'_- g_6},
\label{Zpower}
\end{array}
\end{equation}
where:
\begin{equation}
\begin{array}{lll}
&&m_1=m_2= L|\Lambda|,\quad m_3 =m_4= L'|\Lambda|,\quad \quad \\
&&m_5 =  L(L-1)(|\Lambda|-1)/2,\quad \\
&&m_6 =  L'(L'-1)(|\Lambda|-1)/2.
\end{array}
\label{mss}
\end{equation}
In Appendix \ref{sect:overlapreconstruct} it is shown how the coefficients $c_{g_1,g_2,g_3,g_4,g_5,g_6}$ can be obtained by sampling the partition function for $O(poly(L^2|\Lambda|,L'^2|\Lambda|))$ number of \emph{real} coupling strengths $\vec{\beta}$ which then gives an estimate $\widehat{Z}(\vec{\beta})$.
In order to obtain an estimate $\hat{f}$ of the fidelity overlap, we then need to perform an analytic continuation:
\begin{equation}
\hat{f}=\frac{1}{2^{|\Lambda |}}\Big[\sqrt{\frac{1-\epsilon^2}{\epsilon^4+4}}\Big]^{L |\Lambda|} \Big[\sqrt{\frac{1-\epsilon'^2}{\epsilon'^4+4}}\Big]^{L' |\Lambda|} \; \widehat{Z}(\vec{\beta}^\star).
\end{equation}
Suppose we demand the error in the estimation of the fidelity to be $\epsilon=O(1/poly(|\Lambda|)$:
\[
|\hat{f}-f|\leq \epsilon,
\]
and that the additive error in the estimation of the classical partition function satisfies
\beq
\textrm{Prob}[|\widehat{Z}(\vec{\beta})-Z(\vec{\beta})| \leq \epsilon \; \delta(\vec{\beta})]\geq \frac{3}{4}.
\eeq
Then it is shown in Appendix \ref{sect:overlapreconstruct} that the following precision will suffice:
\begin{equation}
\begin{array}{lll}
\delta(\vec{\beta})&\leq& 16 TT'L (L-1)L'(L'-1) |\Lambda|^6 \\
&&\times~e^{(\beta_++\beta_--6.4)L|\Lambda|}e^{(\beta'_++\beta'_--6.4)L'|\Lambda|}\\
&&\times~e^{(\frac{\beta}{2}-1.6)L(L-1)(|\Lambda|-1)}\\
&&\times~e^{(\frac{\beta'}{2}-1.6)L'(L'-1)(|\Lambda|-1)}.
\end{array}
\end{equation}
To summarize, the required precision in the partition function estimation shows an exponential dependence on quadratic and cubic quantities in the system size. The origin of this dependence lies on the number of  Fourier frequencies needed to reconstruct the partition function.  This  number is obtained by summing the amplitude of the bonds in the lattice (represented in  Fig. \ref{qsimsfig}) associated with each coupling. The interactions corresponding to vertical (horizontal) bonds need a number of Fourier frequencies which is quadratic (cubic) in the system size. This different behaviour  ultimately comes from the chosen adiabatic time dependence on the total Hamiltonian of the system (Eq. \ref{eq:ham-family}).

The reconstruction of the fidelity overlap described above required sampling over a large range for the six ``temperatures" $\vec{\beta}=\{\beta^1,\beta^2,\beta^3,\beta^4,\beta^5,\beta^6\}$.  We can further ask how precisely we need to sample the classical partition function if we only sample over a finite intervals.  Since, for finite systems, the partition function is analytic this is indeed possible.  Defining the interval for each temperature as
\begin{equation}
\Delta_j=\frac{e^{-\beta^j_{\rm min}}-e^{-\beta^j_{\rm max}}}{2m_j}.
\end{equation}
 it is shown in Appendix \ref{sect:overlapreconstruct} that for $\Delta_j\ll 1$ the required precision scales as
\begin{equation}
\delta(\vec{\beta})\le 16TT'L(L-1)L'(L'-1)  \prod_{j=1}^6 e^{(\frac{\beta^j}{2}-1+\log{\frac{\Delta_j}{2}})2m_j}.
\end{equation}
Hence one incurs a penalty exponential in the system size to sample only over a small temperature interval.

We comment that for the sake of simplicity we have restricted our analysis to some homogeneous quantum Ising models in $d$ dimensions, models which already have a well known correspondence to the $d+1$ dimensional classical Ising model \cite{Wu}.  Indeed one may wonder why go through this laborious reconstruction technique involving sampling classical partition functions over six temperature parameters when the quantum phase transition in the $d$ dimensional quantum  transverse Ising model can be simply probed by computing the classical partition function on a $d+1$ dimensional lattice around the critical temperature.  However our construction is more general and allows analysis of quantum models which do not have a well defined classical correspondence.    For example, extensions to disordered quantum spin Hamiltonians of the form, say,
\[
\begin{array}{lll}
 \hat{H}&=&-\sum_{i} (h_i^x \sigma_i^x+h_i^y \sigma_i^y+h_i^z \sigma_i^z)\\
 &&- \sum_{\langle i,j \rangle} (J^x_{i,j} \sigma^x_i \sigma^x_j+J^y_{i,j} \sigma^y_i \sigma^y_j+J^z_{i,j} \sigma^z_i \sigma^z_j)
 \end{array}
 \]
is straightforward \footnote{Again, one could use the Baker-Campbell-Hausdorff expansion to decompose the evolution operator associated with this Hamiltonian. Then, it would be enough to express $\sigma^y$ in terms of $\sigma^x$ and $\sigma^z$ operators using an Euler angle decomposition.}.

Finally, while fidelity overlaps could be estimated using the method of mapping to a generic quantum circuit presented in Sec. \ref{sect:compu-power}, the method described in this section is much more efficient in resource scaling since the gates are applied directly using the transfer matrix formalism rather than mapping to a fixed library of quantum gate in an encoded circuit.  Furthermore, the required accuracy of estimation of the partition function is exponentially better than the bound computed in that case (Eq. \ref{CtoQerrorbound}).

%

 \subsection{Corner magnetisation and estimating partition functions}
 \label{cornermag}
 
The foregoing analysis illustrates the computational power of accurate evaluation of Ising partition functions. We can wonder what is the computational power of more modest tasks, such as estimating the mean values of specific observables. We have studied a simple instance of this problem. As it turns out, very simple tasks already have computional power. For instance, the ability to accurately estimate single site magnetisations on random Ising models lead to random approximation schemes for partition functions. This is the content of the following theorem.

\begin{theorem}

Consider the Ising model on a two-dimensional square lattice $\Lambda$, described by the Hamiltonian:
\beq
H(\sigma)=-J\sum_{\langle i, j \rangle} \sigma_i \sigma_j-h\sum_{i \in \Lambda} \sigma_i.
\eeq
For any $\epsilon$, inverse temperature $\beta$, and magnetic field strength $h$ it is possible to provide an estimate $\hat{Z}(\beta,h)$ for the Ising partition function $Z(\beta,h)$ satisfying
\beq
{\rm Prob}[|\hat{Z}(\beta,h)-Z(\beta,h)| \leq \epsilon \; Z(\beta,h)] \geq 3/4,
\eeq
in a time that scales at most polynomially with $\beta,\epsilon^{-1}, |h|$, and the size of the system if we are able to perform corner magnetisation measurements on specific non-homogeneous Ising systems with a relative precision that need not be lower than the inverse of some polynomial in $|h|$, $\epsilon^{-1}$ and the size of the system.

\end{theorem}
\begin{proof}
The proof is given in Appendix \ref{sect:FPRAS}.
\end{proof}

This result might appear surprising since it applies to even to antiferromagnetic Ising models whereas,  as discussed above, a multiplicative approximation of the partition function in that case is an NP-hard problem.  However, corner measurement is a quantum process which assumes the thermal state of the classical Hamiltonian has been prepared.  Some earlier work \cite{Yung, LidarBiham, vandenNestII} provides quantum algorithms to simulate thermal states of classical spin models.  However as mentioned in Sec. \ref{sec:anacon}, generically these algorithms scale exponentially in the system size, and given the complexity of multiplicative approximations of antiferromagnetic partition functions we would not expect a drastic improvement in thermal state preparation by quantum algorithms in that case.  Whether efficient quantum algorithms exist for preparing ferromagnetic thermal states is as far as we know an open problem but if so than corner magnetisation measurement could prove a useful diagnostic for such algorithms since classical FPRAS is available.  Finally, we add that recently quantum algorithms for FPRAS were found which exhibit a quadratic speed up over the classical counterparts \cite{Wocjan}.  These algorithms are rather different in spirit from measuring corner magnetisation as instead of using mixed states they use a combination of Grover search and phase estimation to prepare pure states of many qubit systems which coherently encode probability distributions of various classical spin configurations.

\section{Conclusions}

In conclusion, we have presented schemes allowing for the measurement of partition functions and mean values of classical many-body systems, at complex temperatures. Although we have mainly focused on Ising Hamiltonians, these schemes can be generalised to other systems, such as the $q$-state Potts model for instance. We have presented two applications of these schemes. 

First, we have studied the possibility to use it in order to compute real temperature partition functions. Although our findings yielded results as poor as previous attempts made by other authors, it is interesting to have found similar results using a different route, in particular one that involves reconstructing partition functions for all temperatures as opposed to a single temperature. We have also seen how experimental data allow to \emph{a posteriori} sharpen error estimates, through a central-limit theorem. This theorem has a validity that extends beyond the present context. Some of its implications will be discussed elsewhere \cite{IB}. To the best of our knowledge, the problem of determining whether quantum mechanics can be used (or not) to efficiently compute partition functions of classical models, or even FPRAS thereof, is still open.  As a second application, we have seen how some link invariants could be deduced from the ability to detect imaginary temperature partition functions, again using constant depth quantum circuits.

These applications all rely on two kinds of schemes, one whose implementation could, in principle, only require a constant time, another involving a time evolution. All schemes translate naturally into global operations and measurements supplemented by edge addressability.  This is natural for certain architectures such as cold trapped atoms in optical lattices \cite{Greiner}, or superconducting qubit arrays \cite{Nori}.  Furthermore, this kind of quantum processing can be made fault tolerant without demanding more addressability as shown in \cite{GlobalQC}. 

We have considered the dual of the first application mentioned, and studied the possibility to efficiently simulate a quantum computer, given the ability to estimate \emph{real temperature} disordered Ising partition functions. We have found that quantum amplitude of a depth-$D$ quantum circuit, acting on $n$ qubits, could be reliably estimated if suitably associated disordered Ising models could evaluated with a precision that essentially grows exponentially with $D$ and $n$. The problem of simulating quantum circuits from statistical mechanical partition functions, estimated with a looser precision (polynomial, say) is, just as open its dual.  One implication is that given the power to compute classical partition functions in $d+1$ dimensions, in certain cases one can compute quantities relevant to quantum phase transitions in $d$ dimensions.  This argument involved viewing the overlap of two ground states of a quantum Hamiltonian as the scattering matrix element for a quantum computation which can then be estimated by computing classical Ising model partition functions with real couplings.  The method was illustrated for the particular case of the quantum transverse Ising model in one dimension and while that model already has a well know classical correspondence, the technique extends to a variety of other quantum spin Hamiltonians in a straightforward manner.  This mapping could provide new ways to perform quantum simulation, via either quantum or classical algorithms for estimating Ising model partition functions.  Given some of the difficulties that beset fault tolerant implementations of quantum simulations \cite{Deutsch, Brown} new approaches are certainly desirable.

Finally, we have seen how the ability to prepare thermal states and perform single qubit measurements immediately implies the existence random approximation schemes. This observation naturally leads to wonder what is the quantum complexity of the preparation of classical thermal state. In view of recent inapproximability results \cite{Galanis}, it would be very interesting to solve this question in the case of the anti-ferromagnetic Ising model for instance.

\section{Acknowledgements}

We would like to thank A. Riera, M. Bremner, T. Cubitt, G. De Las Cuevas, J.I. Latorre, D. P\'erez-Garc\'ia, J. Twamley, and M. van den Nest for discussions. S.I. acknowledges financial support from the Ramon y Cajal program (RYC-2009-04318).  G.K.B. received support from the European Community's Seventh Framework Programme (FP/2007Ð2013) under grant agreement no. 247687 (Integrating Project AQUTE).  G.K.B., M.C., and J.T. received support through the ARC via the Centre of Excellence in Engineered Quantum Systems (EQuS), project number CE110001013.
\appendix

\section{Disordered Systems}\label{appendix_a-priori}

\emph{Preliminary}: We found it convenient to use a slight variation of the detection schemes described in Section \ref{sect:ctpf} and consider single qubit gates described by conjugation of a phase gate by the Hadamard gate:
\begin{equation}
G(\theta)= \mathsf{Had} \left(\begin{array}{cc}1&0\\0&e^{i\theta}\end{array}\right) \mathsf{Had} \;\;.
\end{equation}
For $\theta^\star=-i\log\tanh\beta J$, this single qubit gate turns out to be equal to $T(\beta J)/2 \cosh(\beta J)$, where $T(\beta J)$ is the two-spin Ising transfer matrix introduced in Eq.(\ref{eq:two-spin-tm}).

In this appendix, we are interested in two-dimensional Ising models, of size $n \times m$, with random bond interactions having strengths taking values in $\{ -1,+1 \}$. The magnetic field felt by each spin is also assumed to be random and takes value in $\{-1,0,+1\}$. For a \emph{fixed} configuration of bonds and magnetic fields, the partition function can be evaluated for a specific range of complex temperature.  This is done via instantaneous measurements on a two-dimensional lattice of quantum particles, or through the time evolution of a one-dimensional quantum system. 

The one-step protocol doesn't pose any particular problem for disordered systems. From quantum amplitudes of the form given by Eq.(\ref{eq:ima-pf-simple}) evaluated at specific angles, one can reconstruct the partition function through analytic continuation. Namely,
\beq\label{eq:simple-reconst}
Z(\beta)=A(i \beta)=\sum_{j_1=0}^{2N_1}  w^{(N_1)}( i \beta- \alpha_{j_1}) A(\frac{2 j_1 \pi}{N_1}),   
\eeq
where $w^{(N_1)}$ is defined by Eq.(\ref{def:omega}), and where $N_1$ is polynomial in $n$ and in $m$. 

The case of the time evolved scheme is slightly more complicated than in Section \ref{sect:ctpf}. Reproducing the reasoning presented in that section, one can find an appropriate sequence of controlled phase gates (\ref{eq:def-controlled-phase}) and $G-$gates that provides relevant quantum amplitudes. The real partition functions are again obtained after Fourier transform and analytic continuation. It turns out that three parameters are enough for that. One, $\alpha$, takes into account constant-time interactions and magnetic fields. The two others, $\theta^+$ and $\theta^-$, are respectively related to ferromagnetic and antiferromagnetic interactions between particles corresponding to consecutive time-slices. More precisely, one can see that the kind of partition functions we wish to consider can be written as 
\beq\label{eq_reconstructed}
\begin{array}{lll}
Z(\beta)&=&2^m(e^{\beta}+e^{-\beta})^{N_2^++N_2^-}\sum_{{\nu_1}=-N_1}^{N_1}\sum_{{\nu^+_2}=0}^{N^+_2}\\
&&\sum_{{\nu^-_2}=0}^{N^-_2}c_{{\nu_1} {\nu^+_2}{\nu^-_2}} e^{{\nu_1}\beta}(\tanh{\beta})^{{\nu^+_2}+{\nu^-_2}}\;\;,
\end{array}
\eeq
where $N_1,N^+_2,N^-_2$ are again polynomial in $n$ and in $m$. Actually, $N_1=2 nm-n$ represents a bound on the total number of ``horizontal" bonds plus the number of sites, while $N_2^+$ (resp. $N_2^-$) represents the number of ferromagnetic (resp. antiferromagnetic) edges connecting spins at different time-slices ("vertical" bonds) ($N_2^++N_2^-=m(n-1)$). The coefficients $c_{{\nu_1} {\nu^+_2}{\nu^-_2}}$ are essentially Fourier transforms of quantum amplitudes $A(\alpha,\theta^+,\theta^-)$ detected at selected angles $\alpha,\theta^+,\theta^- \in (0,2\pi]$:
\beq\label{eq:petit-Fourier}
\begin{array}{lll}
c_{{\nu_1} {\nu^+_2}{\nu^-_2}}&=&\frac{(-1)^{{\nu^-_2}} }{(2N_1+1)(N^+_2+1)(N^-_2+1)}\\
&&\times\sum_{j_1=0}^{2N_1}\sum_{j_2^+=0}^{N^+_2}\sum_{j_2^-=0}^{N^-_2}\\
&&\times~e^{-2\pi i ( \frac{\nu_1 j_1}{2N_1+1}+\frac{\nu^+_2  j^+_2}{N^+_2+1}+\frac{\nu^-_2 j^-_2}{N^-_2+1})} \\
&&\times~A(\frac{2 j_1 \pi}{2N_1+1}, \frac{ 2 j^+_2 \pi}{N^+_2+1},\frac{ 2 j^-_2 \pi}{N^-_2+1}).
\end{array}
\eeq
One can note how the particular form of the $G-$gate (which does not involve terms of the form $e^{-i\theta}$) allows for the Fourier transform in $\theta^\pm$ to be restricted to positive frequencies. Again, plugging Eq.(\ref{eq:petit-Fourier}) into Eq.(\ref{eq_reconstructed}) allows one to express the partition function as a function of the ``experimental" data:
\beq\label{eq:composite-reconst}
\begin{array}{lll}
Z(\beta)&=& \frac{2^m (e^\beta+e^{-\beta})^{N_2^++N_2^-}}{(2N_1+1)(N_2^++1)(N_2^-+1)} \\
&&\times\sum_{j_1=0}^{2N_1}\sum_{j_2^+=0}^{N^+_2}\sum_{j_2^-=0}^{N^-_2} A(j_1,j_2^+,j_2^-)\\
&&\times(e^\beta e^{-\frac{2 i \pi j_1}{2N_1+1}})^{-N_1} S^{(2N_1)}(e^\beta e^{-\frac{2 i \pi j_1}{2N_1+1}})\\
&&\times~S^{(N)}(\tanh\beta e^{-\frac{2 i \pi j_2^+}{N_2^++1}})\\
&&\times~S^{(N)}(-\tanh\beta e^{-\frac{2 i \pi j_2^-}{N_2^-+1}}),
\end{array}
\eeq
where $S^{(N)}(q) \equiv (1-q^{N+1})/(1-q)$.

The restricted set of possible values for the couplings and magnetic fields implies that the partition function of the disordered Ising model we are considering can be written as
\beq
Z(\beta)=\sum_{k=-N}^{N} \xi_k e^{-k \beta}
\label{eq_coefficients_xi}
\eeq
where again $N$ scales polynomially with the system size, and where each $\xi_k$ is positive integer whose magnitude is at most equal to the number of possible configurations for the system, i.e. $\xi_k \leq 2^{n m}, \forall k$.  This implies  they can be represented exactly with $nm$ bits. Thus, the estimation of each coefficient $\xi_k$ with $nm$ bits of accuracy, i.e. with a variance $\mathsf{E}_2(\xi_k)$ lower than one would allow for an exact reconstruction of the partition function for \emph{all} temperature. Yet another Fourier transform shows that 
\beq\label{eq_coefficients}
\xi_k=\frac{1}{2N+1} \sum_{j=0}^{2N+1} Z(i \frac{2 j \pi}{2N+1})e^{-i \frac{2 j \pi}{2N+1}}. 
\eeq 
Combining this latter relation with Eq.(\ref{eq:simple-reconst}) for instance, it is possible to see that in order to get $\xi_k$ with $nm$ bits of accuracy, one would need to estimate the quantum amplitudes themselves with $O(nm)$ bits of accuracy. Unfortunately, we do not know how to do that efficiently. In our scheme, the quantum amplitudes are obtained from repeated Bernoulli trials. It therefore seems that $O(2^{nm})$ trials are then necessary. A similar conclusion is reached when the time evolution protocol in one lower dimension is used.\\
\begin{figure}
\begin{center}
\includegraphics[width=\columnwidth]{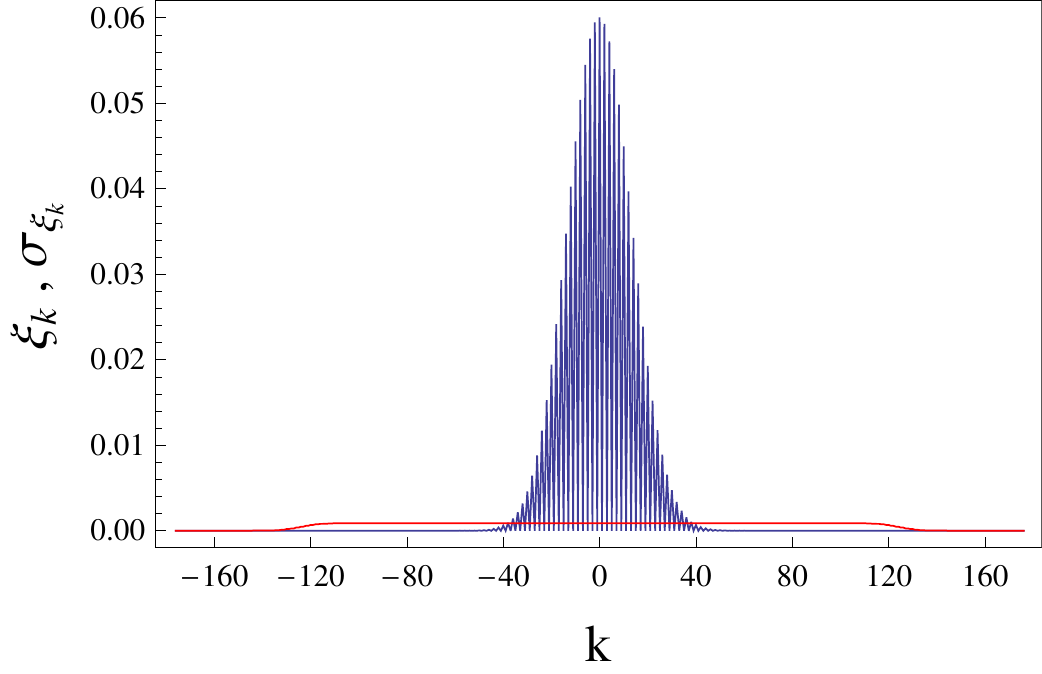}
\end{center}
\caption{(color online) Plot of the coefficients $\xi_k$ (in blue, see Eq. \ref{eq_coefficients_xi}) and an upper bound on their standard deviation $\sigma_{\xi_k}$ (in red) as a function of $k$ for an $8\times 8$ Ising model with $50\%$ positive/negative bonds and uniform magnetic field (set to $1$). From this plot we can qualitatively justify the performances of the algorithm in the small and high temperature limits. The low temperature limit behaviour has to be found in the range of $k$ where the coefficients $\xi_k$ start to be non-zero. This range does not correspond to the maximum possible value of $k$ due to the fact that, in the present model, spin configurations cannot minimize each local hamiltonian. This does not allow to take advantage of the enhanced precision of the protocol for big $k$ and it is the reason for the poor performances of the algorithm at small temperatures. As the temperature increases, the whole  range of $k$ starts to become important, so that we can focus on the intermediate values of $k$, where the bigger coefficients $\xi_k$ are. As evident from the plot, in this regime the relative error is quite small explaining the good high temperatures performances of the protocol.
The value of $k$ where the standard deviation is equal to the relative coefficient $\xi_k$ sets the limit for a possible estimate of an upper bound on the ground state energy.}
\label{plot_XIk}
\end{figure}
We now give some more qualitative insight on the performance of the protocol by analyzing a particular instance of the reconstruction (through the time evolving algorithm) of the coefficients $\xi_k$ (Eq. \ref{eq_coefficients_xi}) for an $8\times 8$ Ising model with $50\%$ positive/negative bonds and uniform magnetic field (set to $1$). In Fig. \ref{plot_XIk} we plot the coefficients $\xi_k$ together with an upper bound on their standard deviation as a function of $k$.

For small temperatures only  coefficients $\xi_k$ with big $k$ are important as it is evident from the series in Eq. \ref{eq_coefficients_xi}.
As shown in the plot, in the ``big $k$'' range, two facts are evident: the standard deviation goes to zero and the coefficients $\xi_k$ are exactly zero. The reason behind the behaviour of the standard deviation is found by algebrically expanding equation \ref{eq_reconstructed} and noticing that the coefficients (responable for the amplification the experimental errors) multiplying big powers of $e^\beta$  are small. On the other hand, the behaviour of the coefficients $\xi_k$ for big $k$ is a natural feature of the disorderd model consider here. More specifically, it simply reflects the impossibility for the ground state spin configuration to minimize each local term of the Hamiltonian, namely, to satisfy each bond and align with the  magnetic field everywhere. The low temperature properties of the model then appear around the values of $k$ where the coefficients $\xi_k$ start to be non-zero. Unfortunately, in that regime the error is no longer approaching zero, explaining why, in this case, the protocol does not perform well at low temperatures. Conversely, for a uniform Ising model, the coefficient $\xi_k$ would be nonzero for the biggest possible $k$.  This explains why we could obtain good results in the low temperature limit for the uniform case (see Fig. \ref{fig:simulations}). 

By inspecting Eq. \ref{eq_coefficients_xi} one is easily convinced that the coefficients $\xi_k$ for smaller $k$ become more important as the temperature increases. In this regime, the standard deviation is basically constant owitnessing properties of the counting process needed to calculate the coefficients $\xi_k$, again obtained by expanding eq \ref{eq_reconstructed} in powers of $e^\beta$. As one can infer by the plot, the relative error is quite small for these intermediate values of $k$, justifying the better high temperatures performances of the protocol.

\begin{figure}
\begin{center}
\includegraphics[width=\columnwidth]{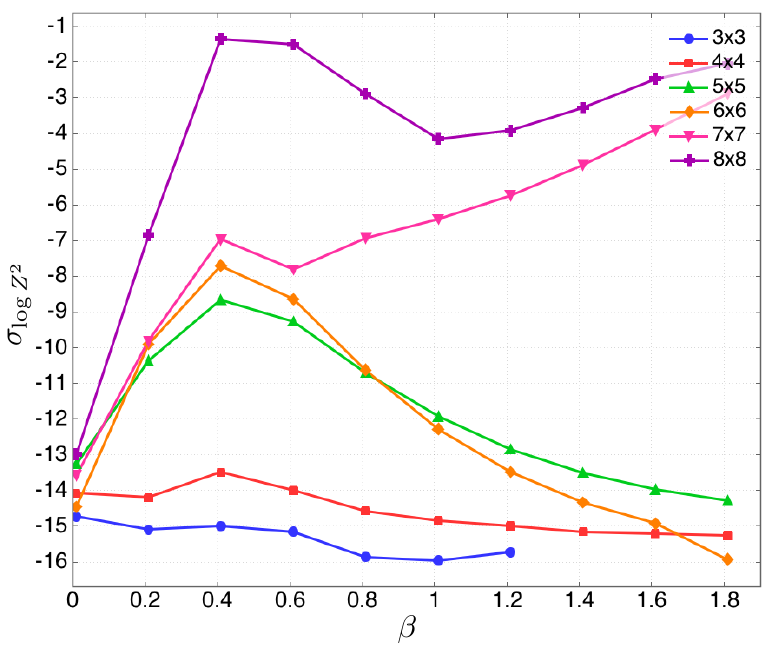}
\end{center}
\caption{Performance of the reconstruction of the squared partition function for the classical ferromagnetic Ising model with open boundary conditions using Protocol 1.  Relative errors in the logarithm of the reconstruction of the squared partition function are plotted as a function of the inverse temperature. Each curve is the relative error for a different system size. For all system sizes the error shows a peak near the critical point. However, for systems larger than $7\times 7$ the error grows quickly as the temperature approaches zero. }
\label{plotZ2}
\end{figure}

Now it is natural to consider the possibility to use our Fourier sampling scheme to estimate an upper bound of the ground state energy. This relies on restating the problem of finding the ground state energy as the problem of finding the maximum $k$ for which $\xi_k\neq 0$. Following this statement,  in order to find an estimate for the upper bound for the ground state energy, we want, roughly speaking, to look at the condition by which the standard deviation on the coefficients $\xi_k$ is not bigger than the coefficients themselves. In the plot presented here, an upper bound on the ground state energy is then obtained by looking at the point where the two curves intersect. As one can see, the result for this instance is very good, but, generically speaking, the impossibility to rule out worst cases scenarios does not allow us to give more quantitative results.


Finally we elaborate on the statement made in Sec. \ref{sect:imp} that two different measurement protocols can be used to calculate partition functions.  Indeed in the same way Protocol 2 can be used to estimate the real temperature partition function via measurements of the quantum overlap $\braket{\Phi}{\Psi}$ and an analytic continuation, Protocol 1 can be used to estimate the square of the real temperature partition function via measurements of the square of the overlap $|\braket{\Phi}{\Psi}|^2$. The only difference is that as the function to reconstruct is squared, the frequencies of the modes in the Fourier series that we construct from experimental data is doubled. Hence, for Protocol 1 more measurements are needed, double the amount needed in Protocol 2. We have reconstructed the square of the partition function of a classical Ising model in 2D and performed the analytic continuation.  For a study of the errors in the reconstruction see Fig.\ref{plotZ2}.

%
%


\section{Proof of theorem \ref{thm:error-stat} }\label{thm:proof-central}
Let us start with the following classical result \cite{BE}

\begin{theorem}[Berry-Ess\'een] 

Let $W_0 \ldots W_{L-1}$ denote $L$ independent random variables such that $\mean{W_j}=0, 0<\mean{W_j^2} < \infty, \mean{|W_j|^3}  < \infty$,  $j \in \{0 \ldots L-1\}$. The cumulative distribution function $\mathcal{F}_W$ of 
\bed
W \equiv \frac{W_0+ \ldots + W_{L-1}}{(\mean{W^2_0}+ \ldots + \mean{W^2_{L-1}})^{1/2}}
\eed 
satisfies the inequality
\beq\label{eq:BE-bound}
|| \mathcal{F}_W-\mathcal{F}_{*}||_{\infty} \;  \leq \; C_{BE} \; \sum_{l=0}^{L-1} \mean{|W_l|^3} / \big( \sum_{l=0}^{L-1} \mean{W_l^2} \big)^{3/2},
\eeq
where $\mathcal{F}_{*}$ denotes the cumulative distribution of a zero-mean unit-variance Gaussian. The value of the constant $C_{BE}$ is at most $0.56$ \cite{Shev}.

\end{theorem}

We are going to use this theorem to study the behaviour of the random variable $dA$, defined as
\beq\label{eq:redef-dA}
dA =\sum_{j=0}^{2 N} \sum_{k=1}^{M}  W_j(k).
\eeq
where 
\beq
W_j(k)=\frac{1}{M} \Re w^{(N)}( i \beta- \alpha_j) \big( \Re A(\alpha_j)-  X_j(k) \big).
\eeq

If the quantum experiments are perfect, then $\Re A(\alpha_j)=\mean{X_j(k)}, \forall k=1 \ldots M$ and $\mean{W_j(k)}=0$ indeed. Let us assume that $0 < \mean{W_j(k)^2} \; \forall j=0 \ldots 2N$, as in the assumptions appearing in the Berry-EssŽen theorem. From a physical point of view, we expect this assumption to be generically satisfied. Indeed, $\mean{W_j(k)^2}=0$ means that $p_j=0$ or that $p_j=1$. In that case, the contribution $\sum_k \frac{1}{M} \Re w^{(N)}( i \beta- \alpha_j) \big( \Re A(\alpha_j)-  X_j(k) \big)$ is always strictly zero, and can therefore not be a source of errors. We will therefore assume that
\beq\label{eq:cond-pj}
\exists \; \delta_*>0 \; \textrm{s.t.} \; \delta_* < p_j < 1-\delta_* \; \forall j=0 \ldots 2N.
\eeq 
Let us introduce the quantity
\beq
\lambda_M=1/\big( \sum_{j,k} \mean{W_j(k)^2} \big)^{1/2}.
\eeq
The random variable $\lambda_M dA$ can certainly be identified with the random variable $W$ appearing in the Berry-Ess\'een theorem and $\forall \Delta >0$,

\bed
\begin{array}{ccc}
&& \textrm{Proba}[|dA| < \Delta ]= \textrm{Proba}[|\lambda_M dA| < \lambda_M \Delta ]= \\
&& \big(\mathcal{F}_{\lambda_M dA}(\lambda_M \Delta)-\mathcal{F}_{*}(\lambda_M \Delta)\big)+ \\
&& \big(\mathcal{F}_{*}(-\lambda_M \Delta)-\mathcal{F}_{\lambda_M dA}(-\lambda_M \Delta) \big) 
+ \big(\mathcal{F}_{*}(\lambda_M \Delta)-\mathcal{F}_{*}(-\lambda_M \Delta)\big) \\
&& \geq 1-2 \mathcal{F}_{*}(-\lambda_M \Delta)-2 || \mathcal{F}_{\lambda_M dA}-\mathcal{F}_{*}||_{\infty} \\
&& \geq 1-2 \mathcal{F}_{*}(-\lambda_M \Delta)-2 C_{BE} D_M
\end{array}
\eed
\bed
\geq 1-2 \mathcal{F}_{*}(-\lambda_M \Delta)-1.12 D_M,
\eed

where
\beq
D_M=
\frac{\sum_{j,k} \mean{|W_j(k)|^3}}{\big( \sum_{j,k} \mean{W_j(k)^2} \big)^{3/2}}.
\eeq

This latter bound is not useful as such because the quantities $\mean{|W_j(k)|^3}$ and $\mean{W_j(k)^2}$, on which $\lambda_M$ and $D_M$ depend, are unknown. For that reason, we will seek to replace $\mean{|W_j(k)|^3}$ and $\mean{W_j(k)^2}$ by appropriate estimates, constructed from experimental observations. In order to lighten a bit the notation, we introduce the (shifted) Bernoulli random variable $B_j(k) \equiv \Re A(\alpha_j)-X_j(k)$. By assumption, for a fixed value of $j$, all $X_j(k)$ are i.i.d. and $\mean{B_j(k)}=0$. Clearly, $\mean{B_j(k)^2} = \mathsf{E}_2(p_j)$ and $\mean{|B_j(k)|^3} = \mathsf{E}_3(p_j)$. If we denote by $p_j$ the probability that $X_j(k)=-1$, it is clear that
\bed
\Re A(\alpha_j)=\mean{X_j(k)}=-p_j+(1-p_j)=1-2 p_j,
\eed
and that 
\bed
\mathsf{E}_2(p_j) =4 p_j (1-p_j). 
\eed
Similarly, 
\bed
\mathsf{E}_3(p_j) =8 p_j (1-3 p_j+4 p_j^2-2 p_j^3).
\eed

Let $\widehat{p}_j$ denote an estimate for $p_j$ constructed from observations as:
\beq\label{eq:def-pjhat}
\frac{1}{M} \sum_{k=1}^M X_j(k)=1-2\widehat{p}_j.
\eeq
Applying Hoeffding's inequality to the case of $M$ identical Bernoulli trials shows that
\bed
\textrm{Proba}[|\widehat{p}_j-p_j| \leq \epsilon] \geq 1-2 e^{-2 \epsilon^2 M}  \; \forall \epsilon >0.
\eed
$\widehat{p_j}$ can be used to construct estimates for $\mean{B_j(k)^2}$ and $\mean{|B_j(k)|^3}$ as 
\bed
\mathsf{E}_2(\widehat{p}_j) \equiv 4 \widehat{p}_j (1- \widehat{p}_j),
\eed
\bed
\mathsf{E}_3(\widehat{p}_j) \equiv 8 \widehat{p}_j (1- 3 \widehat{p}_j+4 \widehat{p}^2_j-2 \widehat{p}^3_j).
\eed

Since $\mathsf{E}_2$ and $\mathsf{E}_3$ are continuous differentiable functions over $[0,1]$, we have that, whenever $|\widehat{p}_j-p_j| \leq \epsilon$, then 
\bed
\big |\widehat{\mathsf{E}}_2(\widehat{p}_j)-\mathsf{E}_2(p_j) \big| \leq \max_{0 \leq p \leq 1} \big| \frac{d}{dp} \mathsf{E}_2(p)  \big| \times \epsilon= 4 \epsilon,
\eed
and 
\bed
|\widehat{\mathsf{E}}_3(\widehat{p}_j)-\mathsf{E}_3(p_j)| \leq \max_{0 \leq p \leq 1} \big| \frac{d}{dp} \mathsf{E}_3(p)  \big| \times \epsilon= 8 \epsilon.
\eed

Let $\epsilon_j$ denote a set of $2N+1$ positive numbers. We see that whenever $|\widehat{p}_j-p_j| \leq \epsilon_j \; \forall j=0 \ldots 2N$, which occurs with probability at least 
\bed
\prod_{j=0}^{2N+1}  \big( 1-2 e^{-2 \epsilon_j^2 M}\big), 
\eed
the numerator of $D_M$ is upper bounded by the quantity
\bed
\sum_{j,k} |\Re w^{(N)}( i \beta- \alpha_j)|^3 \big( \widehat{\mathsf{E}}_3(\widehat{p}_j)+8 \epsilon_j \big),
\eed
while the quantity $\sum_{j,k} \mean{W_j(k)^2}$, appearing in the denominator of $D_M$, is lower bounded by
\bed
\mathcal{V}_M=\sum_{j,k} |\Re w^{(N)}( i \beta- \alpha_j)|^2 \big( \widehat{\mathsf{E}}_2(\widehat{p}_j)-4 \epsilon_j \big)  
\eed
So, whenever this latter quantity is strictly positive and $|\widehat{p}_j-p_j| \leq \epsilon_j \; \forall j=0 \ldots 2N$, the quantity
\bed
\widetilde{D}_M(\{\epsilon_j\})=
\frac{\sum_{j,k} |\Re w^{(N)}( i \beta- \alpha_j)|^3 \big( \widehat{\mathsf{E}}_3(\widehat{p}_j)+8 \epsilon_j \big)}
{ \big(\sum_{j,k} |\Re w^{(N)}( i \beta- \alpha_j)|^2 \big( \widehat{\mathsf{E}}_2(\widehat{p}_j)-4 \epsilon_j \big)  \big)^{3/2}}
\eed
\bed
=\frac{1}{\sqrt{M}}
\frac{\sum_{j=0}^{2N} |\Re w^{(N)}( i \beta- \alpha_j)|^3 \big( \widehat{\mathsf{E}}_3(\widehat{p}_j)+8 \epsilon_j \big)}
{ \big(\sum_{j=0}^{2N} |\Re w^{(N)}( i \beta- \alpha_j)|^2 \big( \widehat{\mathsf{E}}_2(\widehat{p}_j)-4 \epsilon_j \big)  \big)^{3/2}}
\eed
upper bounds $D_M$.

Also, whenever $|\widehat{p}_j-p_j| \leq \epsilon_j \; \forall j=0 \ldots 2N$, the quantity 
\bed
\widetilde{\lambda}_M(\{\epsilon_j\})=
\frac{M}{\sqrt{ \sum_{j,k} |\Re w^{(N)}( i \beta- \alpha_j)|^2 \big( \widehat{\mathsf{E}}_2(\widehat{p}_j)+4 \epsilon_j \big)}}
\eed
\bed
=\frac{\sqrt{M}}{\sqrt{ \sum_{j=0}^{2N} |\Re w^{(N)}( i \beta- \alpha_j)|^2 \big( \widehat{\mathsf{E}}_2(\widehat{p}_j)+4 \epsilon_j \big)}}
\eed
lower bounds $\lambda_M$. Of course, whenever $\widetilde{D}_M(\{\epsilon_j\}) \geq D_M$ and $\widetilde{\lambda}_M(\{\epsilon_j\}) \leq \lambda_M$, we have that 
\bed
1-2 \mathcal{F}_{*}(-\lambda_M \Delta)-2 C_{BE} D_M \geq
\eed
\beq
1-2 \mathcal{F}_{*}(-\widetilde{\lambda}_M(\{\epsilon_j\}) \Delta)-2 C_{BE} \widetilde{D}_M(\{\epsilon_j\})
\eeq

One possibility to ensure that $\mathcal{V}_M \geq 0$ is to pick 
\beq\label{def:epsilon-j}
\epsilon_j=\frac{1}{4+s} \mathsf{E}_2(\widehat{p}_j),
\eeq
where $s>0$ is a constant we are free to choose at our convenience. It is not possible to ensure that $\mathcal{V}_M$ is always \emph{strictly} positive. Indeed, from Eq.(\ref{eq:def-pjhat}), we see that in the event where $X_j(1)=\ldots=X_j(M) \; \forall j=0 \ldots 2N$, we have that $\widehat{p}_j=0$ or $\widehat{p}_j=1$, implying that $\mathsf{E}_2(\widehat{p}_j)=0 \; \forall j$ and that $\mathcal{V}_M=0$. Then $\widehat{D}_M$ would be infinite, a situation where we are not able to construct a useful estimator. For that reason, we define our estimator for $D_M$ as follows: 
\beq
\widehat{D}_M(\{\epsilon_j\})= \left\{
\begin{array}{rl}
& \widetilde{D}_M(\{\epsilon_j\})
\; \text{if} \; \mathcal{V}_M \neq 0, \\
& 0  \; \text{if} \; \mathcal{V}_M=0.
\end{array} \right.
\eeq
Our estimator for $\lambda_M$ is defined as 
\beq
\widehat{\lambda}_M(\{\epsilon_j\})= \left\{
\begin{array}{rl}
& \widetilde{\lambda}_M(\{\epsilon_j\})
\; \text{if} \; \mathcal{V}_M \neq 0, \\
& -\infty  \; \text{if} \; \mathcal{V}_M=0.
\end{array} \right.
\eeq

Fortunately, the probability of a pathological situation,
\bed
\Proba{\mathcal{V}_M=0}=\prod_{j=0}^{2N} \big((p_j)^M+(1-p_j)^M\big).
\eed
is exponentially small in $M$ whenever $0 < p_j < 1$ for at least some $j$. 

Let us estimate \emph{the probability to get a valid and useful bound} $\mathcal{L}$. We consider the following four events:
\bed
\mathcal{A}= \{\mathcal{V}_M \neq 0 \}.
\eed
\bed
\mathcal{B}= \{ |\widehat{p}_j-p_j| \leq \epsilon_j  \forall j \}.
\eed
\bed
\mathcal{C}= \{ D_M \leq \widehat{D}_M(\{\epsilon_j\}) \}.
\eed
\bed
\mathcal{D}= \{ \lambda_M \geq \widehat{\lambda}_M(\{\epsilon_j\}) \}.
\eed
We are interested in the event $\mathcal{A} \cap \mathcal{C} \cap \mathcal{D}$. Obviously, 
\bed
\Proba{\mathcal{A} \cap \mathcal{C} \cap \mathcal{D}}=
\Proba{\mathcal{C} \cap \mathcal{D}}
\eed
\bed
-\Proba{ \mathcal{C} \cap \mathcal{D} | \textrm{not} \mathcal{A}} \Proba{\textrm{not} \mathcal{A}}
\eed
and
\bed
\Proba{\mathcal{C} \cap \mathcal{D}}
\geq \Proba{ \mathcal{C} \cap \mathcal{D} \cap \mathcal{B} }. 
\eed
Therefore,
\bed
\Proba{\mathcal{A} \cap \mathcal{C} \cap \mathcal{D}} 
\geq \prod_{j=0}^{2N} \big( 1-2 e^{-\epsilon_j^2 M} \big)-\prod_{j=0}^{2N} \big( p_j^M+(1-p_j)^M  \big),
\eed

which tends to $1$ exponentially as $M$ grows large.


\emph{In summary, the random variable
\bed
1-2 \mathcal{F}_{*}(-\widehat{\lambda}_M(\{\epsilon_j\}) \Delta)-2 C_{BE} \widehat{D}_M(\{\epsilon_j\}),
\eed
with $\epsilon_j$ defined by Eq.(\ref{def:epsilon-j}), lower bounds the quantity $\Proba{|dA| < \Delta}$ with probability at least}
\bed
\mathcal{P}(\{\epsilon_j \},M,N) \equiv \prod_{j=0}^{2N} \big( 1-2 e^{-\epsilon_j^2 M} \big)-\prod_{j=0}^{2N} \big( p_j^M+(1-p_j)^M  \big),
\eed

\section{Proof of Lemma \ref{thm:universal-gate-set}}\label{thm:universal-gate-set-proof}

We begin with the discrete gate set 
\beq\label{eq:inter-gate-set}
\begin{array}{lll}
\mathfrak{G}_0&=&\{ \mathsf{Z}_k(\pi/4), \mathsf{Had}_k, k=1 \ldots n \ \} \; \\
&&\cup
\{ \mathsf{CNOT}_{k,k+1}, k=1 \ldots n-1 \}
\end{array}
\eeq
acting on an $n$ qubit register that is universal for quantum computation \cite{Mor}.  By the Solovay-Kitaev \cite{nielsen} theorem an arbitrary polynomial sized quantum circuit can be efficiently approximate from this gate set with a polynomial overhead. To realize this using global operations in the mirror encoding of Raussendorf, one makes frequent use of the global shift operator $\mathsf{G}_{\text{tot}}= \sigma^z_{\text{tot}}(\pi) \sigma^y_{\text{tot}}(\pi/2) \mathsf{CP}_{\textrm{tot}}$, with the property that $\mathsf{G}_{\text{tot}}^{2n+1}$ is a reflection of the state of the chain about its middle. An arbitrary $Z$ rotation on logical qubit $k$ can be physically implemented as \cite{Raussendorf}
\beq\label{eq:indi-gate-1}
\begin{array}{lll}
\mathsf{Z}_k^{\text{logi}}(\alpha)&=&e^{i \frac{\alpha}{2} (\sigma^z_k+\sigma^z_{n-k+1})}\\
&=&\mathsf{G}_{\text{tot}}^{n+1-k} \; \sigma^{y}_{\text{tot}}(\pi) \; \mathsf{G}\sigma^{y}_{\text{tot}}(\pi)\mathsf{G}^{k-1}\\
&&\times\sigma^{z}_{\text{tot}}(-\alpha/2)
\mathsf{G}_{\text{tot}}^{n+1-k} \; \sigma^{y}_{\text{tot}}(\pi) \; \mathsf{G} \; \sigma^{y}_{\text{tot}}(\pi) \; \\&&\times\mathsf{G}^{k-1} \; \sigma^{z}_{\text{tot}}(\alpha/2).
\end{array}
\eeq

Similarly, an $X$ rotation on logical qubit $k$ is
\bed\label{eq:indi-gate-2} 
\begin{array}{lll}
\mathsf{X}_k^{\text{logi}}(\alpha)&=&e^{i \frac{\alpha}{2} (\sigma^x_k+\sigma^x_{n-k+1})}\\
&=&\mathsf{G}_{\text{tot}}^{n-k} \;
 \sigma^{y}_{\text{tot}}(\pi) \;
  \mathsf{G} \; 
  \sigma^{y}_{\text{tot}}(\pi) 
  \; \mathsf{G}^{k} \; 
  \sigma^{z}_{\text{tot}}(-\pi/2) \;\\
 &&\times~ \sigma^{y}_{\text{tot}}(\alpha/2) \;
    \sigma^{z}_{\text{tot}}(\pi/2) \;
     \mathsf{G}_{\text{tot}}^{n-k} \;
  \sigma^{y}_{\text{tot}}(\pi) 
  \mathsf{G}\sigma^{y}_{\text{tot}}(\pi) \\
 &&\times~ \mathsf{G}^k
    \sigma^{z}_{\text{tot}}(-\pi/2) \;  
  \sigma^{y}_{\text{tot}}(-\alpha/2)
    \sigma^{z}_{\text{tot}}(\pi/2).
    \end{array}
 \eed

Finally, an entangling gate between logical qubits $k$ and $k+1$ can be implemented as
\beq\label{eq:indi-gate-3}
\begin{array}{lll}
\mathsf{V}_{k,k+1}^{\text{logi}}(\alpha)&=&e^{i \alpha \big( \sigma^z_k \otimes \sigma^x_{k+1}+ \sigma^z_{k+n} \otimes \sigma^x_{k+n-1} \big)}\\
&=&\mathsf{G}^k \mathsf{X}_k^{\text{logi}}(\alpha) \mathsf{G}^{\dagger k}.
\end{array}
\eeq
Since $V_{k,k+1}(\pi/4)\textsf{Had}_{k+1}Z_k(\pi/2)Z_{k+1}(\pi/2)\textsf{Had}_{k+1}=\textsf{CNOT}_{k,k+1}$ then the gate set
\beq\label{eq:inter-gate-set}
\begin{array}{lll}
\mathfrak{G}_1&=&\{ \mathsf{Z}^{\text{logi}}_k(\pi/4), \mathsf{Had}^{\text{logi}}_k, k=1 \ldots n \ \} \; \\
&&\cup
\{ V_{k,k+1}(\pi/4), k=1 \ldots n-1 \}
\end{array}
\eeq
is universal for quantum computation. 
Now the Hadamard gate can be related to $X$ and $Z$ rotations through the identity $\textsf{Had}=\sigma^z(\pi/2) \sigma^x(\pi/2) \sigma^z(\pi/2)$. 
Also we note the following relations: $[\sigma^z(\pi/8)]^{31}=\sigma^z(-\pi/8)$, and $\sigma^{y}(\pm \pi/4)=\sigma^{x}(-\pi/2)\sigma^{z}(\mp \pi/4)\sigma^{x}(\pi/2)$ and also $\sigma^{x}(\pm \pi/2)=\sigma^{z}(\pm\pi/2)\textsf{Had}\sigma^{z}(\pm \pi/2)$.  Then from Eqs.(\ref{eq:indi-gate-1},\ref{eq:indi-gate-2},\ref{eq:indi-gate-3}), we see that it is enough to be able to implement
\bed
\mathfrak{G}=\{\mathsf{CP}_{\text{tot}}, \sigma^z_{\text{tot}}(\pi/8), \mathsf{Had}_{\text{tot}} \}
\eed
in order to achieve universal quantum computation. 

\section{Proof of Theorem \ref{thm:Ising-adiab-combined}}\label{app:proof-adiab-combined}

Our starting point is the following direct consequence of the adiabatic theorem, as stated in \cite{Amb-Reg}.

\begin{lemma}
\label{gap}
Let $\gamma=\text{min}_{t \in [0:T]} \text{gap} \; \hat{H}(t)$, where $\text{gap} \; \hat{H}(t)$ denotes the difference between the two lowest eigenvalues of $\hat{H}(t)$, and let $\ket{\Phi'}$ denote the quantum state obtained by the \emph{continuous} evolution induced on $\ket{\Phi_0}$ by the Hamiltonian family (\ref{eq:ham-family}). Let also $|\Lambda|$ and $|E(\Lambda)|$ denote respectively the number of sites and edges of the lattice $\Lambda$. The distance between $\ket{\Phi'}$ and the true ground state $\ket{G}$ is at most $\delta$ whenever T satisfies
\beq\label{eq:adiab1}
T \geq T_*(\hat{H},\delta)= \frac{10^5}{\delta^2} \frac{\big( |h| \cdot |\Lambda|+ |J| \cdot |E(\Lambda)|\big)^3}{\gamma^4}.
\eeq

\end{lemma}
\begin{proof}
Let us introduce the parameter $s=t/T$. Theorem 2.1 of Ref.\cite{Amb-Reg} provides the following sufficient condition for adiabaticity\footnote{In the following $||A||_{\infty}$ will denote the operator norm of an operator $A$, i.e. $||A||_{\infty}=\sup_x \frac{||A x||_2}{||x||_2}$.}:
\beq\label{eq:gen-adiab-cond}
\begin{array}{lll}
T& \geq& T_*(\hat{H},\delta)\\
&=&\frac{10^5}{\delta^2} 
\max_{0 \leq s \leq 1} \max \{ \frac{|| \frac{d}{ds} \hat{H} ||_{\infty}^3}{\gamma^4}, \frac{||\frac{d}{ds} \hat{H} ||_{\infty} \cdot ||\frac{d^2}{ds^2} \hat{H} ||_{\infty}}{\gamma^3}\}
\end{array}
\eeq
valid for any time-dependent hamiltonian $\hat{H}(t)$. Adapting this condition to the special case of Hamiltonians (\ref{eq:ham-family}), we see the r.h.s of (\ref{eq:adiab1}) certainly upper bounds the r.h.s of (\ref{eq:gen-adiab-cond}).
\end{proof}

We wish to discretise the time evolution of our system. Instead of considering the time-dependent evolution associated with the Hamiltonians $\hat{H}(t)$, we will deal with $L$ consecutive \emph{constant} unitary operators, $\mathcal{U}_k=\text{Exp}\big(-i \; \tau \; \hat{H}_0-i \; \tau \; \hat{H}_1(k \tau) \big), k=0 \ldots L-1$, where we define the discretisation step as
\beq
\tau\equiv T/L.
\eeq 
We wish to work with the state $\ket{\Phi^\star}=\mathcal{U}_{L-1} \ldots \mathcal{U}_0 \ket{+_x^{\otimes |\Lambda|}}$ rather than with the state $\ket{\Phi'}$. Of course when $L$ grows large we expect this substitution to have negligible effect. But we need to be precise and quantify the induced error. The following lemma addresses this issue.
\begin{lemma}
\label{thm:Ising-adiab} 
The distance between $\ket{G}$ and $\ket{\Phi^\star}$ is bounded as
\beq\label{eq:adiab2}
|| \ket{\Phi^\star}-\ket{G} || \leq \delta+ T \sqrt{\frac{2 \big( |h| \cdot |\Lambda|+ |J| \cdot |E(\Lambda)|\big)}{L}},
\eeq
whenever $T \geq T_*(\hat{H},\delta)$.
\end{lemma}
\begin{proof}
The triangular inequality yields
\beq\label{eq:vd-mosca}
|| \; \ket{\Phi^\star}-\ket{G} || \leq || \; \ket{\Phi'}-\ket{G} || +|| \; \ket{\Phi^\star}-\ket{\Phi'} ||.
\eeq
The first term of the r.h.s of this expression is of course bounded by $\delta$. To bound the second, we use Lemma 1 of \cite{vanDam}, which states that if two time-dependent Hamiltonians $H_a(t), H_b(t), 0 \leq t \leq T$ differ at most by $\epsilon$ in operator norm for every $t$, then the difference between the unitary evolutions they induce, $\mathcal{U}_a(T),\mathcal{U}_b(T)$ satisfy $||\mathcal{U}_a(T)- \mathcal{U}_b(T)||_{\infty} \leq \sqrt{2 T \epsilon} $. For every $t \in [0,T]$, let $k(t) \in \{0, \ldots, L-1 \}$ such that $k(t) \tau \leq t \leq (k(t)+1) \tau$. Clearly, $|| \hat{H}(t)- \hat{H}( k(t) \tau) || \leq \tau (|h| \cdot |\Lambda|+ |J| \cdot |E(\Lambda)| )$. Identifying the r.h.s. of this inequality with $\epsilon$ and bearing in mind the definition of $\tau$, one bounds the second term of the r.h.s. of (\ref{eq:vd-mosca}) in the desired way.
\end{proof}

Next, we split each unitary $\mathcal{U}_k$ into a part that depends only on $\hat{H}_0$ and a part that depends only on $\hat{H}_1(k \tau)$: for $\tau$ small enough, each unitary $\mathcal{U}_k$ can be safely replaced by the operator
\beq
U_k=e^{-i \tau \hat{H}_0} e^{-i \tau \hat{H}_1(k\tau)}.
\eeq
Indeed, the Baker-Campbell-Hausdorff identity \cite{vanDam} implies that
\beq
\label{BCH}
|| \mathcal{U}_k-U_k ||_{\infty} \leq K \big( |h| \cdot |\Lambda|+ |J| \cdot |E(\Lambda)|\big) \cdot \big( |h_{\perp}| \cdot |\Lambda| \big) \tau^2,
\eeq
for some \emph{constant} $K$. Then we arrive at the following:
\begin{lemma}

The quantity by which the state $U_{L-1} U_{L-2} \ldots U_0 \ket{+_x^{\otimes |\Lambda|}}$ deviates from the true ground state of $H^\star$ is at most 
\beq\label{eq:approx-final}
\begin{array}{lll}
\Delta &=& \delta+ T \sqrt{\frac{2 \big( |h| \cdot |\Lambda|+ |J_{\parallel}| \cdot |E(\Lambda)|\big)}{L}}\\
&&+ K L  \big( |h| \times |\Lambda|+ |J_{\parallel}| \times |E(\Lambda)|\big) \times |h_{\perp}| \cdot |\Lambda| \tau^2.
\end{array}
\eeq

\end{lemma}
\begin{proof}
The result follows by combining the inequality in Eq. \ref{BCH} with the Lemmata 
\ref{gap},\ref{thm:Ising-adiab}.
\end{proof}

\section{Approximation of fidelity overlaps}\label{sect:overlapreconstruct}

In this section we describe how to reconstruct fidelity overlap which is proportional to a partition function with complex couplings by sampling from partition functions with real couplings.  We begin by rewriting Eq. \ref{Zpower} using more compact notation:
\begin{align}
\begin{array}{ll}
Z(\vec{\beta})=B(\vec{\beta})&\displaystyle{\sum_{g_1=-n_1}^{0}\sum_{g_2=-n_2}^{0}\sum_{g_3=-n_3}^{0}\sum_{g_4=-n_4}^0\sum_{g_5=-n_5}^{0}\sum_{g_6=-n_6}^{0}   }\\
&\tilde{c}_{g_1,g_2,g_3,g_4,g_5,g_6}e^{\sum_{j=1}^6 \beta^j g_j},
\end{array}
\end{align}
where:
\[
\begin{array}{lll}
\vec{\beta}&=&\{\beta^1,\beta^2,\beta^3,\beta^4,\beta^5,\beta^6\}\equiv\{\beta_+,\beta_-,\beta'_+,\beta'_-,\beta,\beta'\},\\
n_1&=&n_2= 2 L|\Lambda|,\quad n_3 =n_4= 2 L'|\Lambda|,\quad \quad \\
n_5 &=&  L(L-1)(|\Lambda|-1),\quad n_6 =  L'(L'-1)(|\Lambda|-1).
\end{array}
\]
and where $B(\vec{\beta})=\prod_{j=1}^6 B_j(\beta^j)$ with $B_j(\beta^j)=e^{\frac{1}{2}n_j \beta^j}$ and $\tilde{c}$ is just a relabeling of $c$ with each index $g_j$ ranging from $[-n_j,0]$ rather than $[-n_j/2,n_j/2]$ (recall $n_j=2m_j$ defined in Eq. \ref{mss}).

Let us define the polynomial:
\begin{equation}
\begin{array}{lll}
p(\vec{x})&=&p(x_1,x_2,x_3,x_4,x_5,x_6)\\
&=&\displaystyle{\sum_{i_1=0}^{n_1}\sum_{i_2=0}^{n_2}\sum_{i_3=0}^{n_3}\sum_{i_4=0}^{n_4}\sum_{i_5=0}^{n_5}\sum_{i_6=0}^{n_6}  }\tilde{c}_{i_1,i_2,i_3,i_4,i_5,i_6} x_1^{i_1}x_2^{i_2}x_3^{i_3}x_4^{i_4}x_5^{i_5}x_6^{i_6},
\end{array}
\end{equation}
where $\vec{x}=\{x_1,x_2,x_3,x_4,x_5,x_6\}\in\mathbb{R}^6$.
Introducing the notation:
\[
x(\cdot)=e^{-(\cdot)}\;\;,
\]
one has the trivial relation:
\begin{equation}
\label{eq_PolyPartition}
\begin{array}{ll}
p(x_1(\beta_+),x_2(\beta_-),x_3(\overline{\beta}'_+),x_4(\beta'_-),x_5(\beta),x_6(\beta'))=B^{-1}(\vec{\beta})Z(\vec{\beta}).
\end{array}
\end{equation}
Note that, for physical temperatures, the domain of the polynomial is such that $x_j>0$ and $||x_j||\leq 1$ for $j=1,\dots,6$.
We now want to reconstruct the polynomial $p(\vec{x})$ from a set of $N$ data values $p(\vec{x}_{\vec{i}})$ with $\vec{x}_{\vec{i}}\equiv\{x_{1,i_1},x_{2,i_2},x_{3,i_3},x_{4,i_4},x_{5,i_5},x_{6,i_6}\}\in\Gamma$ where $\Gamma$ is a certain lattice of points in $\mathbb{R}^6$. Although several options are available \cite{Chung}, in our case the polynomial is such that the simplest possible option can be used: a rectangular mesh lattice as:
\[
\Gamma=\{x_{1,i_1=1},\dots,x_{1,i_1=n_1+1}\}\times\cdots\times\{x_{6,i_6=1},\dots,x_{6,i_6=n_6+1}\}.
\]
This is justified by the fact that, as we constructed it, the polynomial $p(\vec{x})$ has degree at most $n_j$ in $x_j$ ($j=1,\dots,6$). This means that $p(\vec{x})$ actually lies in the product space $\Pi_{n_1}\times\cdots\times\Pi_{n_6}$, where $\Pi_n$ indicates the space of univariate polynomials of degree at most $n$.
Explicitly, the data values are written as:
\[
\begin{array}{lll}
p(\vec{x}_{\vec{i}})&\equiv& p(x_{1,i_1},x_{2,i_2},x_{3,i_3},x_{4,i_4},x_{5,i_5},x_{6,i_6})\\
&\equiv& p_{i_1 i_2 i_3 i_4 i_5 i_6}\;\;.
\end{array}
\]
The reconstructed polynomial can then be written as:
\begin{align}
\label{eq_coeff_poly}
p(\vec{x})=\sum_i p_{i_1 i_2 i_3 i_4 i_5 i_6} l_{i_1 i_2 i_3 i_4 i_5 i_6}(\vec{x})\;\;,
\end{align}
where:
\[
\begin{array}{lll}
l_{i_1 i_2 i_3 i_4 i_5 i_6}(\vec{x})&=&l_{1,i_1}(x_1)l_{2,i_2}(x_2)l_{3,i_3}(x_3)\\
&&\times~ l_{4,i_4}(x_4)l_{5,i_5}(x_5)l_{6,i_6}(x_6)\;\;,
\end{array}
\]
with:
\begin{align}
l_{j,i_j}(x)=\prod_{\scriptsize\begin{array}{c}k_j=1\\ k_j\neq i_j\end{array}}^{n_j+1}\frac{x-x_{j,k_j}}{x_{j,i_j}-x_{j,k_j}}\;\;.
\end{align}
We now suppose to have a device that provides an estimate $\widehat{Z}(\beta)$ for the partition function, $Z(\beta)$, that satisfies 
\begin{equation}
\label{eq_oracle}
|\widehat{Z}(\beta)-Z(\beta)| \leq \delta\;\;,
\end{equation}
and, from this, we want to see how well we can estimate the previously defined overlaps. Since the overlaps depend on the analytically continued partition function $Z(\vec{\beta}^\star)$, we are going to show how to  reconstruct it. From Eq. \ref{eq_PolyPartition} and Eq. \ref{eq_coeff_poly} we can write:
\begin{equation}
\begin{array}{lll}
Z(\vec{\beta}^\star)&=&B(\vec{\beta}^\star)p(x_1(\beta_+^\star),x_2(\beta_-^\star),x_3({{\beta}'}^\star_+),\\
&&\times~ x_4({\beta}'^\star_-),x_5(\beta^\star),x_6({\beta}'^\star))\\
&=&B(\vec{\beta}^\star)\sum_i p_{i_1 i_2 i_3 i_4 i_5 i_6}\prod_{j=1}^6 l_{j,i_j}(\vec{x}^j\left(\beta^{\star j})\right).
\end{array}
\end{equation}
Now the coefficients $p_{i_1 i_2 i_3 i_4 i_5 i_6}$ are the values of the polynomial evaluated at the lattice points $\vec{x}_i$, and we can use the real temperature version of the partition function in order to write:
\begin{align}
Z(\vec{\beta}^\star)&=B(\vec{\beta}^\star)\sum_{\vec{i}} B^{-1}(\vec{\beta}_{\vec{i}})Z(\vec{\beta}_{\vec{i}})\prod_{j=1}^6 l_{j,i_j}(\vec{x}^j\left(\beta^{\star j})\right),
\end{align}
where $\vec{\beta}_{\vec{i}}$ represents the lattice $\Gamma$ transformed in ``$\beta$ coordinates'': 
\begin{equation}
\begin{array}{lll}
\vec{\beta}_{\vec{i}} &\equiv&\{\beta_{+,i_1},\beta_{-,i_2},{\beta}'_{+,i_3},{\beta}'_{-,i_4},\beta_{i_5},{\beta}'_{i_6}\}\\
&\equiv& \{-\log{x_{1,i_1}},-\log{x_{2,i_2}},-\log{x_{3,i_3}},\\
&&-\log{x_{4,i_4}},-\log{x_{5,i_5}},-\log{x_{6,i_6}}\}.
\end{array}
\end{equation}
We also want to make an explicit choice for this lattice:
\begin{align}\label{eq:lattice}
\vec{x}^j_{i_j}\equiv\frac{i_j}{n_j+1}~~\text{with:}~~i_j=1,\dots,n_j+1\;\;,
\end{align}
which clearly satisfies the properties of the rectangular mesh $\Gamma$ we stated before. Explicitly the mapping of this lattice in the ``temperature domain'' reads:
\begin{align}\label{eq:lattice_beta}
\beta^j_{i_j}=-\log{\frac{i_j}{n_j+1}},
\end{align}
and henceforth we use the notation $\vec{\beta}_{\vec{i}}=\{\beta^1_{\vec{i}},\beta^2_{\vec{i}},\beta^3_{\vec{i}},,\beta^4_{\vec{i}},\beta^5_{\vec{i}},\beta^6_{\vec{i}}\}$.
Now, we can write the final formula for the overlap as a function of the estimation of the partition function at real temperatures as:
\begin{equation}\label{eq:corr-action}
\begin{array}{lll}
f&=& \frac{1}{2^{|\Lambda |}}\Big[\sqrt{\frac{1-\epsilon^2}{\epsilon^4+4}}\Big]^{L |\Lambda|} \Big[\sqrt{\frac{1-\epsilon'^2}{\epsilon'^4+4}}\Big]^{L' |\Lambda|} \;B(\vec{\beta}^\star)\\
&&\times\sum_{\vec{i}} B^{-1}(\vec{\beta}_{\vec{i}})Z(\vec{\beta}_{\vec{i}})\prod_{j=1}^6 l_{j,i_j}(\vec{x}^j\left(\beta^{\star j})\right)\;\;.
\end{array}
\end{equation}
We are interested in studying how the variance on this quantity scales. We have:
\begin{equation}
\begin{array}{lll}
\sigma^2_{f}&\leq& \frac{1}{2^{2(L+L'+1)|\Lambda |}}\;|B(\vec{\beta}^\star)|^2\sum_{\vec{i}} |B^{-1}(\vec{\beta}_{\vec{i}})|^2\sigma^2_{Z(\vec{\beta}_{\vec{i}})}\\
&&\times\prod_{j=1}^6 |l_{j,i_j}(\vec{x}^j\left(\beta^{\star j})\right)|^2\;\;.
\end{array}
\end{equation}
From Eq. \ref{eq_oracle} we have:
\begin{equation}
\label{eq_sigmaOverlap}
\begin{array}{lll}
\sigma^2_{f}&\leq& \frac{1}{2^{2(L+L'+1)|\Lambda |}}\;|B(\vec{\beta}^\star)|^2\sum_{\vec{i}} \delta_{\vec{i}}^2 |B^{-1}(\vec{\beta}_{\vec{i}})|^2\\
&&\times\prod_{j=1}^6 |l_{j,i_j}(\vec{x}^j\left(\beta^{\star j})\right)|^2\;\;.
\end{array}
\end{equation}
We now study the term by term the quantities in this expression.  First,
\begin{equation}
\begin{array}{lll}
|B(\vec{\beta}^\star)|&=&\prod_{j=1}^6 |B_j(\beta^{\star j})|\\
&=&\prod_{j=1}^6 |e^{\frac{1}{2}n_j\beta^{\star j}}|\\
&=&\left|\left(\frac{1}{\sqrt{-i(1+\epsilon)}}\right)^{n_1/2}\left(\frac{1}{\sqrt{i(1-\epsilon)}}\right)^{n_2/2}\right.\\
&&\times\left.\left(\frac{1}{\sqrt{i(1+\epsilon)}}\right)^{n_3/2}\left(\frac{1}{\sqrt{-i(1-\epsilon)}}\right)^{n_4/2}\right|\\
&=&\frac{1}{(1-\epsilon^2)^{(L+L')|\Lambda|}},
\end{array}
\end{equation}
and
\begin{equation}
\begin{array}{lll}
|B(\vec{\beta}_{\vec{i}})^{-1}|&=&\prod_{j=1}^6 |B_j(\vec{\beta}^{j}_{i_j})^{-1}|\\
&=&\prod_{j=1}^6 |e^{-\frac{1}{2}n_j\vec{\beta}^{j}_{i_j}}|\\
&=&\prod_{j=1}^6 x_{j,i_j}^{n_j/2}\\
&=&\prod_{j=1}^6 \left(\frac{i_j}{n_j+1}\right)^{n_j/2}.
\end{array}
\end{equation}
We now turn to each term $ |l_{j,i_j}\left(\vec{x}^j(\beta^{\star j})\right)|$ for each fixed $j$:
\begin{equation}
\begin{array}{lll}
l_{j,i_j}(\vec{x}^j(\beta^{\star j}))&=&\prod_{\scriptsize\begin{array}{c}k_j=1\\ k_j\neq i_j\end{array}}^{n_j+1}\frac{|\vec{x}^j(\beta^{\star j})-x_{j,k_j}|}{|x_{j,i_j}-x_{j,k_j}|}\\
&=&\frac{(n_j+1)^{n_j}}{\prod_{k_j =1}^{i_j-1} (i_j-k_j)\prod_{k_j =i_j+1}^{n_j+1} (k_j-i_j)}\\
&&\times\frac{\prod_{k_j=1}^{n_j+1}\sqrt{(\rho_j \cos{\theta_j}-\frac{k_j}{n_j+1})^2+\rho_j^2 \sin{\theta_j}^2}}{\sqrt{(\rho_j \cos{\theta_j}-\frac{i_j}{n_j+1})^2+\rho_j^2 \sin{\theta_j}^2}}\\
&&\\
&\leq& \frac{(n_j+1)^{n_j}}{i_j!(n_j+1-i_j)!}\frac{e^{\frac{n_j+1}{2}I(\rho_j,\theta_j)}}{\sqrt{(\rho_j \cos{\theta_j}-\frac{i_j}{n_j+1})^2+\rho_j^2 \sin{\theta_j}^2}},
\end{array}
\end{equation}
where $\vec{x}^j(\beta^{\star j})\equiv \rho_j e^{i\theta_j}$ as can be deduced by looking at Eqs. \ref{eq:cons:betapm}, \ref{betavariables} and:
\begin{equation}
\begin{array}{llll}
I(\rho_j,\theta_j)&=&\int_{0}^1 dx&\log{[(\rho_j \cos{\theta_j}-(x-\frac{1}{n_j+1}))^2}+\rho_j^2 \sin{\theta_j}^2]\\
&<&-\frac{1}{4}&.
\end{array}
\end{equation}
The last inequality holds for the cases considered by Eq. \ref{betavariables}, for $n_j\geq10$.  Note that 
$\theta_1=-\theta_2=-\theta_3=\theta_4=\frac{\pi}{4}$ and $\theta_5,\theta_6\ll 1$ since $J\tau/L,J\tau'/L'\ll 1$ by assumption in the adiabatic mapping.
Reassembling everything and using Eq. \ref{eq_sigmaOverlap} we get:
\begin{align}
\sigma^2_{f}&\leq \sum_{\vec{i}} A_{\vec{i}}^2 \delta_{\vec{i}}^2   \;\;,
\end{align}
with:
\begin{equation}
\begin{array}{lll}
A_{\vec{i}}^2 &=&  \frac{\prod_{j=1}^6 \left[\frac{\left(\frac{i_j}{n_j+1}\right)^{n_j}    (n_j+1)^{2 n_j}e^{(n_j+1)I(\rho_j,\theta_j)}}{\left((\rho_j \cos{\theta_j}-\frac{i_j}{n_j+1})^2+\rho_j^2 \sin{\theta_j}^2\right)\left( i_j!(n_j+1-i_j)!\right)^2}\right]}{(1-\epsilon^2)^{2(L+L')|\Lambda|}2^{2(L+L'+1)|\Lambda |}}\\
&&\\
&\leq&\frac{\prod_{j=1}^6 \left[\frac{\left(\frac{i_j}{n_j+1}\right)^{n_j}    (n_j+1)^{2 n_j}e^{-(n_j+1)/4}}{\left( i_j!(n_j+1-i_j)!\right)^2}\right]}{(1-\epsilon^2)\theta_5^2\theta_6^2\sin^8{\frac{\pi}{4}}(1-\epsilon^2)^{2(L+L')|\Lambda|}2^{2(L+L'+1)|\Lambda |}}.
\end{array}
\end{equation}
In arriving at the inequality above we used the fact that ${(\rho_j \cos{\theta_j}-\frac{i_j}{n_j+1})^2+\rho_j^2 \sin^2{\theta_j}}\geq \rho_j^2\sin^2{\frac{\pi}{4}}$ for $j=1,2,3,4$ and  ${(\rho_j \cos{\theta_j}-\frac{i_j}{n_j+1})^2+\rho_j^2 \sin^2{\theta_j}}\geq \theta_j^2$ for $j=5,6$, supposing that $\theta_5,\theta_6\rightarrow 0$ as  is the case.  In the temperature domain this formula reads:
\begin{align}
\begin{array}{ll}
&A^2(\beta^j_{i_j})=
\frac{\prod_{j=1}^6 \left[\frac{e^{-n_j\beta^j_{i_j}}   (n_j+1)^{2 n_j}e^{-(n_j+1)/4}}{ \Gamma^2((n_j+1)e^{-\beta^j_{i_j}}+1)\Gamma^2((n_j+1)(1-e^{-\beta^j_{i_j}})+1)}\right]}{\theta_5^2\theta_6^2\sin^8{\frac{\pi}{4}}(1-\epsilon^2)^{2(L+L'+1)|\Lambda|}2^{2(L+L'+1)|\Lambda |}}\;\;.
\end{array}
\end{align}
From this we can get the following condition for the variance on the overelap to be polynomially bounded in the system size expressed for generic temperatures:
\begin{align}
\delta(\vec{\beta})\leq \frac{1}{A(\vec{\beta})}\;\;.
\end{align}
Explicitly we have:
\begin{align}
\begin{array}{ll}
\delta(\vec{\beta})\leq\prod_{j=1}^6\frac{\theta_5 \theta_6\sin^4{\frac{\pi}{4}}e^{\frac{n_j+1}{8}} \Gamma((n_j+1)e^{-\beta^{j}}+1)\Gamma({(n_j+1)(1-e^{-\beta^{j}})+1)} }{(n_j+1)^{n_j} e^{-\frac{n_j}{2}\beta^{j}}},
\end{array}
\end{align}
where we used  $1-\epsilon^2\geq\frac{1}{2}$.

\begin{figure}
 \begin{centering}
 \includegraphics[width=\columnwidth]{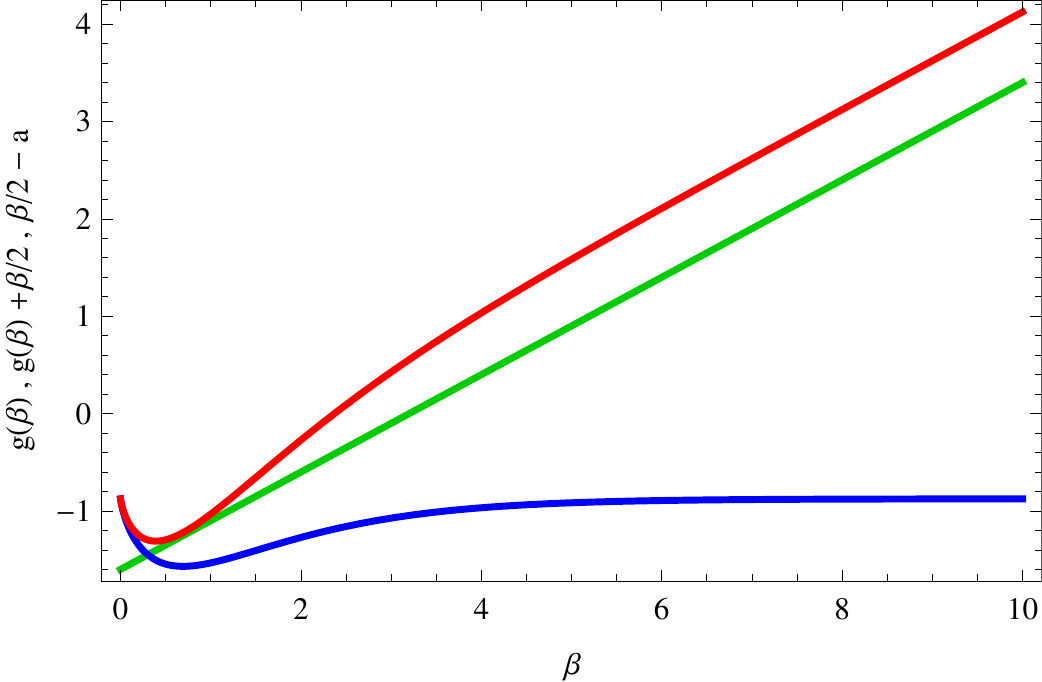}
   \caption{Plots of the functions $g(\beta)$ (blue) and $g(\beta)+\frac{\beta}{2}$ (red) and $\frac{\beta}{2}-a$ (green) with $a = \min_\beta{g(\beta)}= -1.6$ as defined in Eq. \ref{gfunction}.}
   \end{centering}\label{fig_delta}
\end{figure}

Using Stirling approximation we then obtain:
\begin{align}
\label{eq_delta_result}
\begin{array}{lll}
\log{\delta(\vec{\beta})}&\leq& \sum_{j=1}^6 n_j \left(g(\beta^j)+ \frac{\beta^j}{2}\right) \\
&+& \sum_{j=1}^6 \log{(n_j+1)} + \sum_{j=1}^6 g(\beta^j)+K\;\;,
\end{array}
\end{align}
where:
\begin{align}
\label{gfunction}
g(\beta_j)&=(1-e^{-\beta^j})\log{(1-e^{-\beta^j})}-\beta^j e^{-\beta^j}-\frac{7}{8}\\
K&= 4\log{\sin{\frac{\pi}{4}}}+\log{\theta_5\theta_6}.
\end{align}
This result is telling us how much error we can tolerate in the sampling of the classical partition function in order to be able to reconstruct certain quantum overlaps with a precision that scales polynomially in the system size. All error values satisfying Eq. \ref{eq_delta_result} allow for such a reconstruction. For this reason, if we want to obtain a weaker but more compact result we can chose to state a smaller threshold.
We can do this by substituting the functions appearing in Eq.  \ref{eq_delta_result}  with their minimum (see Fig. \ref{fig_delta}):
\begin{align}
\log{\delta(\vec{\beta})}\leq \sum_{j=1}^6  \left(\frac{\beta^j}{2}+a\right) n_j + \sum_{j=1}^6 \log{(n_j+1)} + b\;\;,
\end{align}
where:
\begin{align}
a &\leq \min_\beta{g(\beta)}\sim -1.6\\
b &= K-10\;\;.
\end{align}
In the thermodyamic limit the above formula can be further approximated by:
\begin{equation}
\begin{array}{lll}
\log{\delta(\vec{\beta})}&\leq& \sum_{j=1}^6 \left(\frac{\beta^j}{2}+a\right)  n_j + \sum_{j=1}^{6}\log (n_j) +\log{\theta_5\theta_6} \;\;,
\end{array}
\end{equation}
so that:
\begin{align}
\begin{array}{lll}
\delta(\vec{\beta})&\leq& 2^4 TT'L (L-1) L'(L'-1) |\Lambda|^6\prod_{j=1}^6 e^{\left(\frac{\beta^j}{2}-1.6\right)  n_j }.
\end{array}
\end{align}

This is the central result of this section.  It  has been obtained by supposing we sample the classical partition function in the (inverse) temperature lattice of points as given by Eq. \ref{eq:lattice_beta}. One can note that, in the thermodynamic limit, we are effectively sampling on a domain which ranges over all possible temperatures. We want now to address the question of how the result changes if we instead sample only on a limited temperature domain inside $[\beta^j_{\rm max},\beta^j_{\rm min}]$ for each of the six variables $\beta^j$. To do this we introduce the interval (recall we have set $J=J'=1$)
\begin{equation}
\Delta_j=\frac{e^{-\beta^j_{\rm min}}-e^{-\beta^j_{\rm max}}}{n_j}.
\end{equation}
and slightly modify the definition of the lattice given by Eq. \ref{eq:lattice} to have:
\begin{align}
\vec{x}_{j,i_j}\equiv\frac{\Delta_j}{n_j}(i_j-1)+x^E_j~~\text{with:}~~i_j=1,\dots,n_j+1
\end{align}
where $0< x^E_j \le 1$ and $\Delta_j$ is contrained so that $0< x_{j,i_j}\le 1$. In the following we want to focus on the case when $\Delta_j\rightarrow 0$.\\
The only difference with respect to the previous case lies in the terms $ |l_{j,i_j}(\vec{x}^j\left(\beta^{\star j})\right)|$ for each fixed $j$. We now have:
\begin{equation}
\begin{array}{lll}
l_{j,i_j}(\vec{x}^j(\beta^{\star j}))&=&\prod_{\scriptsize\begin{array}{c}k_j=1\\ k_j\neq i_j\end{array}}^{n_j+1}\frac{|\vec{x}^j(\beta^{\star j})-x_{j,k_j}|}{|x_{j,i_j}-x_{j,k_j}|}\\
&=&\left(\frac{n_j}{\Delta_j}\right)^{n_j}\frac{\frac{\prod_{k_j=1}^{n_j+1}\sqrt{(\rho_j \cos{\theta_j}-\frac{\Delta_j}{n_j}(k_j-1)-x^E_j)^2+\rho_j^2 \sin{\theta_j}^2}}{\sqrt{(\rho_j \cos{\theta_j}-\frac{\Delta_j}{n_j}(i_j-1)-x^E_j)^2+\rho_j^2 \sin{\theta_j}^2}}}{\prod_{k_j =1}^{i_j-1} |i_j-k_j|\prod_{k_j =i_j+1}^{n_j+1} |i_j-k_j|}\\
&&\\
&=&\left(\frac{n_j}{\Delta_j}\right)^{n_j}\frac{\frac{e^{\frac{1}{2}\sum_{k_j=1}^{n_j+1}{(\rho_j \cos{\theta_j}-\frac{\Delta_j}{n_j}(k_j-1)-x^E_j)^2+\rho_j^2 \sin{\theta_j}^2}}}{\sqrt{(\rho_j \cos{\theta_j}-\frac{\Delta_j}{n_j}(i_j-1)-x^E_j)^2+\rho_j^2 \sin{\theta_j}^2}}}{\prod_{k_j =1}^{i_j-1} |i_j-k_j|\prod_{k_j =i_j+1}^{n_j+1} |i_j-k_j|}\\
&&\\
&\leq &\left(\frac{n_j}{\Delta_j}\right)^{n_j}\frac{\frac{e^{\frac{n_j}{2\Delta_j}I(\rho_j,\theta_j,\Delta_j,S_j,x^E_j)}}{\sqrt{(\rho_j \cos{\theta_j}-\frac{\Delta_j}{n_j}(i_j-1)-x^E_j)^2+\rho_j^2 \sin{\theta_j}^2}}}{i_j!(n_j+1-i_j)!}
\end{array}
\end{equation}
where:
\begin{equation}
\begin{array}{lll}
I(\rho_j,\theta_j,\Delta_j,S_j,x^E_j)&=&\int_{x^E_j}^{x^E_j+\Delta_j}dx~\log{\left[(\rho_j\cos{\theta_j}-(x-S_j))^2\right.}\\
&&{\left.+\rho_j^2\sin^2{\theta_j}\right]}
\end{array}
\end{equation}
with:
\begin{equation}
S_j=\left\{\begin{array}{c} \frac{\Delta_j}{n_j}~~\text{if}~~x^E_j,x^E_j+\Delta_j < \rho_j \cos{\theta_j}  \\
0 ~~\text{if}~~x^E_j,x^E_j+\Delta_j > \rho_j\cos{\theta_j}.\end{array}\right.
\end{equation}
The value of the additional variable $S_j\ll 1$ introduced here depends on whether both $x^E$ and $x^E+\Delta_j$ lie on the same side of the domain of $x$ split by the position of the minimum of the function we want to integrate $f(x)=(\rho_j\cos{\theta_j}-x)^2+\rho_j^2\sin{\theta_j}$. Since in the end we want to work with $\Delta_j\ll 1$ this is not such a restrictive hypothesis but it allows for the following inequality (used to get the bounds on the quantities $l_{j,i_j}$) to be true:
\begin{equation}
\begin{array}{lll}
I(\rho_j,\theta_j,\Delta_j,S_j,x^E_j)&\ge& \sum_{k_j=1}^{n_j+1}{(\rho_j \cos{\theta_j}-\frac{\Delta_j}{n_j}(k_j-1)-x^E_j)^2}\\
&&+\rho_j^2 \sin{\theta_j}^2
\end{array}
\end{equation}
%
The integral $I(\rho_j,\theta_j,\Delta_j,S_j,x^E_j)$ can be computed to first order in $\Delta$:
\begin{equation}
\begin{array}{lll}
I(\rho_j,\theta_j,\Delta_j,S_j,x^E_j)&=&\frac{n_j\Delta_j}{2}\log{ P_j}
\end{array}
\end{equation}
where $P_j=(x_E-S_j)^2+\rho_j^2-2(x_E-S_j)\rho\cos{\theta_j}$. Note that: $0<P_j<(x_E-S_j+\rho_j)^2$. In the thermodynamic limit and by supposing $x_E<1-\epsilon$ we have $0<P_j<4$. The last equality defining $I$ holds at the first order in $\Delta_j$. Analogously to what was done before we write:
\begin{equation}
\begin{array}{lll}
\delta(\vec{\beta})\leq\theta_5\theta_6\sin^4{\frac{\pi}{4}} \prod_{j=1}^6\frac{\Gamma((n_j+1)e^{-\beta^{j}}+1)\Gamma({(n_j+1)-(n_j+1)e^{-\beta^{j}}+1)} }{(\frac{n_j}{\Delta_j})^{n_j} e^{-\frac{n_j}{2}\beta^{j}}P_j^{\frac{n_j}{2}}}
\end{array}
\end{equation}
and then take advantage of the Stirling approximation:
\begin{equation}
\begin{array}{lll}
\log{\delta(\vec{\beta})}&\le& K+\sum_{j=1}^6(n_j+1)\log{(n_j+1)}-n_j\log{n_j}\\
&&+n_j (g'(\beta^j)+\frac{1}{2}\beta^j-\frac{1}{2}\log{P_j}+\log{\Delta_j})+g'(\beta^j)
\end{array}
\end{equation}
where:
\begin{equation}
\begin{array}{lll}
g'(\beta^j)&=&(1-e^{-\beta^j})\log{(1-e^{\beta^j})}-\beta^j e^{-\beta^j}-1\\
K&=& 4\log{\sin{\frac{\pi}{4}}}+\log{\theta_5\theta_6}\;\;.
\end{array}
\end{equation}
In the thermodynamic limit this result becomes:
\begin{equation}
\begin{array}{lll}
\log{\delta(\vec{\beta})}&\le& K+\sum_{j=1}^6 1+\log{n_j}+n_j (g'(\beta^j)+\frac{1}{2}\beta^j-\frac{1}{2}\log{ P_j}\\
&&+\log{\Delta_j})+g'(\beta^j),
\end{array}
\end{equation}
or
\begin{equation}
\begin{array}{lll}
\delta(\vec{\beta})\le 2^4TT'L(L-1)L'(L'-1)\prod_{j=1}^6{\left(\frac{\Delta_j}{\sqrt{P_j}}\right)^{n_j}}  e^{(\frac{\beta^j}{2}-1)n_j}.
\end{array}
\end{equation}
A less conservative result takes advantage of the upper bound for $P_j$ so that:
\begin{equation}
\delta(\vec{\beta})\le2^4TT'L(L-1)L'(L'-1)  \prod_{j=1}^6 e^{(\frac{\beta^j}{2}-1+\log{\frac{\Delta_j}{2}})n_j}.
\end{equation}
Hence the price for allowing the classical partition function to be estimated only in a small temperature window is an overhead exponential in the system size.

\section{Magnetisation and Approximation Schemes}\label{sect:FPRAS}

The first part of our construction closely follows a general argument presented in Ref. \cite{JS}, and establishes a connection between partition function evaluations and the ability to draw samples from Boltzmann probability distributions. Some adaptations were made, though. We felt that indicating only these adaptations would have resulted in an awkward presentation. This is why, for the sake of clarity, we have chosen to reproduce this argument, with these adaptations included, in a concise but self-contained manner. In the second part of our construction, we show how measurements of magnetisation on specific non-homogeneous Ising models allow to draw from Boltzmann distributions. 

Let us thus consider the Ising model on a two-dimensional square lattice $\Lambda$, described by the Hamiltonian:
\beq
H(\sigma)=-J\sum_{\langle i, j \rangle} \sigma_i \sigma_j-h\sum_{i \in \Lambda} \sigma_i.
\eeq
For $h=0$, the model is solvable and $Z(h=0)$ is known exactly (see e.g. \cite{Wu}). We wish to evaluate the partition function at a fixed temperature\footnote{Change of notations: Since we will work at constant temperature, we will from now drop $\beta$ and simply write $Z(h)$ instead of $Z(\beta,h)$.} $\beta$, $Z(h)$, for $h >0$, say\footnote{The case $h <0$ is treated similarly.}. For that purpose, we express $Z(h)$ as 
\beq
Z(h)=\frac{Z(h_L)}{Z(h_{L-1})} \times \frac{Z(h_{L-1})}{Z(h_{L-2})} \times \ldots \times \frac{Z(h_{1})}{Z(h_{0})} \times Z(h_{0}),
\eeq 
where $0=h_0 < h_1 < \ldots < h_L=h$. These values $h_k$ are chosen to be equally spaced, and we will denote the spacing $h_k-h_{k-1}$ by $\delta h$. Each ratio $\varrho_k=Z(h_k)/Z(h_{k-1})$ can be expressed as 
\beq
\begin{array}{lll}
\varrho_k&=&\sum_{\sigma} \frac{e^{-\beta H_{k-1}(\sigma)}}{Z(h_{k-1})} e^{\beta \delta h |\Lambda| M(\sigma)}\\
& \equiv &\sum_{\sigma} \pi_{k-1}(\sigma) \; e^{\beta \delta h |\Lambda| M(\sigma)},
\end{array}
\eeq
where $M(\sigma)$ denotes the mean magnetisation of the system when the lattice is in configuration $\sigma$, $|\Lambda|$ denotes again the size of the lattice $\Lambda$, and where $H_{k-1}$ is a shorthand notation for the hamiltonian when the magnetic field is set to $h_{k-1}$.

In order to evaluate $Z(h)$, we will use a collection of estimators for the quantities $\varrho_k$, each involving $n$ sample configurations. These estimators are defined as
\beq
\begin{array}{lll}
\hat{\varrho}_k: \{ \sigma_k^{(1)}, \ldots, \sigma_k^{(n)}\} &\to& \hat{\varrho}_k(\sigma_k^{(1)}, \ldots, \sigma_k^{(n)})\\
&=&\frac{1}{n} \sum_{j=1}^n e^{\beta |\Lambda| \delta h M(\sigma_k^{(j)})},
\end{array}
\eeq  
where each sample $\sigma_k^{(j)}$ is drawn according to some probability distribution $\pi'_{k-1}$. Our estimator for $Z(h)$ is 
\bed
\hat{Z}(h) \equiv \prod_{k=1}^L \hat{\varrho}_k \; Z(h_0).
\eed

Let $\bar{\varrho}_k$ denote the mean value of $\hat{\varrho}_k$, i.e.
\bed
\begin{array}{ll}
\bar{\varrho}_k=\sum_{\sigma_k^{(1)}} \ldots \sum_{\sigma_k^{(n)}} &\pi'_{k-1}(\sigma_k^{(1)}) \ldots \pi'_{k-1}(\sigma_k^{(n)}) \; \\
&\times~\hat{\varrho}_k (\sigma_k^{(1)}, \ldots, \sigma_k^{(n)}).
\end{array}
\eed
Since all $\hat{\varrho}_k$ are independent random variables, we find that the mean value of $\hat{Z}(h)$ is given by $\bar{Z}(h)=\prod_{k=1}^L \bar{\varrho}_k Z(h_0)$. Now let us assume that 
\beq\label{eq:first-general-bound}
|Z(h)-\bar{Z}(h)| \leq \epsilon' Z(h),
\eeq
and that
\beq
|\bar{Z}(h)-\hat{Z}(h)| \leq \delta \; \bar{Z}(h),
\eeq
with probability at least, $3/4$ say \footnote{This value is somewhat arbitrary. As explained in Ref.\cite{JS}, any level of confidence strictly above $1/2$ can be efficiently boosted to arbitrarily close to 1.}. Then  
\bed
(1-\delta)(1-\epsilon') Z(h) \leq \hat{Z}(h) \leq (1+\delta)(1+\epsilon') Z(h),
\eed
with probability at least $3/4$. Thus 
\beq
(1-\epsilon) Z(h) \leq \hat{Z}(h) \leq (1+\epsilon) Z(h)
\eeq
with probability at least $3/4$ whenever $\epsilon \geq \delta+\epsilon'+\delta \epsilon'$. 

Clearly,
\bed
e^{-\beta \delta h |\Lambda|} \leq e^{\beta |\Lambda| \delta h M(\sigma)} \leq e^{\beta \delta h |\Lambda|} \hspace{0.2 cm} \forall \sigma.
\eed
Plugging these inequalities into Hoeffding's inequality \cite{BE}, we find that
\beq
\text{Prob}[ |\hat{\varrho}_k-\bar{\varrho}_k| \leq \zeta ] \geq 1-2 e^{-2 n \zeta^2/\sinh(|\Lambda| \beta \delta h)^2}.
\eeq
Let us use this latter relation in order to construct an upper bound on $|\hat{Z}(h)-\bar{Z}(h)|$ valid with tunable probability.  We will use the following Lemma:

\begin{lemma}
\label{approxmagZ}
\beq
|\hat{Z}(h)-\bar{Z}(h)| \leq |\prod_{k=1}^L (1+\frac{\zeta}{\bar{\varrho}_k})-1| \; \bar{Z}(h)
\eeq
with probability at least $(1-2 e^{-2 n \zeta^2/\sinh(|\Lambda| \beta \delta h)^2})^L$.
\end{lemma}

\begin{proof}
We start with the following identity
\bed
\begin{array}{lll}
|\hat{Z}(h)-\bar{Z}(h)|&=&|\prod_{k=1}^L \hat{\varrho}_k-\prod_{k=1}^L \bar{\varrho}_k| \; Z(h_0)\\
&=&|\prod_{k=1}^L (1+\frac{\hat{\varrho}_k-\bar{\varrho}_k}{\bar{\varrho}_k})-1| \; \bar{Z}(h)
\end{array}
\eed
Next, we have the inequality
\beq\label{eq:useful-ineq}
|\prod_{k=1}^L (1+x_k)-1| \leq |\prod_{k=1}^L (1+|x_k|)-1|, \hspace{0.6cm} \forall x_k \in \mathbb{R}.
\eeq
Let us consider two cases: (i) $\prod_{k=1}^L (1+x_k)-1\geq 0$, (ii) $\prod_{k=1}^L (1+x_k)-1 < 0$. The inequality is trivial in case (i). In case (ii), we need to prove that
\bed
1-\prod_{k=1}^L (1+x_k) \leq \prod_{k=1}^L (1+|x_k|)-1,
\eed
or $2 \leq \prod_{k=1}^L (1+|x_k|)+\prod_{k=1}^L (1+x_k)$. The r.h.s. of this last inequality can certainly be written as 
\bed
2+\sum_{i_1} \ldots \sum_{i_L} \varkappa_{i_1 \ldots i_L} (|x_1|^{i_1} \ldots |x_{L} |^{i_L}+x_1^{i_1} \ldots x_{L}^{i_L}), 
\eed
where each coefficient $\varkappa_{i_1 \ldots i_L}$ is non-negative. It is also clear that each quantity $(|x_1|^{i_1} \ldots |x_{L} |^{i_L}+x_1^{i_1} \ldots x_{L}^{i_L})$ is non-negative.
Inequality (\ref{eq:useful-ineq}) implies that 
\bed
|\hat{Z}(h)-\bar{Z}(h)| \leq |\prod_{k=1}^L (1+\frac{| \hat{\varrho}_k-\bar{\varrho}_k |}{\bar{\varrho}_k})-1| \; \bar{Z}(h).
\eed
The r.h.s of this relation is lower than $|\prod_{k=1}^L (1+\frac{\zeta}{\bar{\varrho}_k})-1| \; \bar{Z}(h)$ with probability at least 
$(1-2 e^{-2 n \zeta^2/\sinh(|\Lambda| \beta \delta h)^2})^L$ (Hoeffding's inequality).
\end{proof}

 We will pick the spacing between two consecutive magnetisations to be $\delta h=\frac{\eta}{\beta |\Lambda|}$, where $\eta$ is some positive constant we are free to choose at our convenience. $\delta h$ fixes the value of $L$ to
\beq\label{eq:cond-L}
L=(h-h_0) \beta |\Lambda|/\eta.
\eeq
With a given choice for $\delta h$, we have that $\bar{\varrho}_k \geq e^{-\eta}$ and
\beq\label{eq:bias}
|\hat{Z}(h)-\bar{Z}(h)| \leq | (1+ e^{\eta} \zeta)^L -1 | \; \bar{Z}(h),
\eeq
with probability at least $(1-2 e^{-2 n \zeta^2/\sinh(|\Lambda| \beta \delta h)^2})^L$. How should we pick $\zeta$ in order to ensure that the l.h.s. of (\ref{eq:bias}) is smaller than $\delta \bar{Z}(h)$ for some fixed $\delta$ ? Since $(1+ e^{\eta} \zeta)^L \leq e^{L \zeta e^{\eta}}$, it is enough that 
\bed
\zeta \leq \frac{\ln (1+\delta)}{L e^{\eta}}.
\eed
We also wish to know how, for fixed values of $\zeta, L, \eta$, we should choose $n$ in order to guarantee a level of confidence at least equal to $3/4$. Direct substitution shows that the condition
\bed
(1-2 e^{-2 n \zeta^2/\sinh(\eta)^2})^L \geq 3/4
\eed
is satisfied if
\beq\label{eq:cond-n}
n \geq - \frac{\sinh{\eta}^2 e^{2 \eta} L^2}{2(\ln (1+\delta))^2} \ln \left[\frac{1}{2}\left(1-\left(\frac{3}{4}\right)^{1/L}\right)\right].
\eeq
To summarise, for $L$ satisfying (\ref{eq:cond-L}) and $n$ satisfying (\ref{eq:cond-n}), the partition function estimator satisfies 
\beq
\text{Prob}[|\hat{Z}(h)-\bar{Z}(h)| \leq \delta \bar{Z}(h) ] \geq 3/4. 
\eeq
Next we wish to establish a condition that guarantees that Inequality (\ref{eq:first-general-bound}) holds. We start by observing that
\bed
|\bar{Z}(h)-Z(h)| \leq |\prod_{k=1}^L (1+\frac{|\varrho_k-\bar{\varrho}_k|}{\varrho_k})-1| \; Z(h).
\eed
Let 
\bed
\begin{array}{lll}
\Delta \pi_{k-1} &\equiv& \max_{S} | \pi_{k-1}(S)-\pi'_{k-1}(S) |\\
&=&\frac{1}{2} \sum_{\sigma} |\pi_{k-1}(\sigma)-\pi'_{k-1}(\sigma)|
\end{array}
\eed
denote the total variation\footnote{To obtain the last equality, one observes that if an event $S_*$ achieves the maximum, so does the complementary event $S_*^c$.} between the probability distributions $\pi_{k-1}$ and $\pi'_{k-1}$. Let us also denote $\Delta \pi^*=\max_k \Delta \pi_{k-1}$. We see that 
\bed
|\varrho_k-\bar{\varrho}_k| \leq e^{\eta} \Delta \pi^*, \hspace{0.5cm} \rho_k \geq e^{-\eta} \hspace{0.5cm} \forall k.
\eed
Thus
\beq
\begin{array}{lll}
|Z(h)-\bar{Z}(h)| &\leq& [(1+ e^{2 \eta} \Delta \pi^*)^L-1] Z(h) \\
&\leq& (e^{L e^{2 \eta} \Delta \pi^*}-1) Z(h).
\end{array}
\eeq
So it is enough that 
\bed
\Delta \pi^* \leq \frac{e^{-2 \eta}}{L} \ln(1+\epsilon').
\eed
On another hand, $\Delta \pi_{k-1}$ satisfies the inequality
\bed
\Delta \pi_{k-1} \leq \frac{1}{2} \max_{\sigma} |1-\frac{\pi'_{k-1}(\sigma)}{\pi_{k-1}(\sigma)}|.
\eed

\begin{figure}[h]
\begin{center}
\includegraphics[width=\columnwidth]{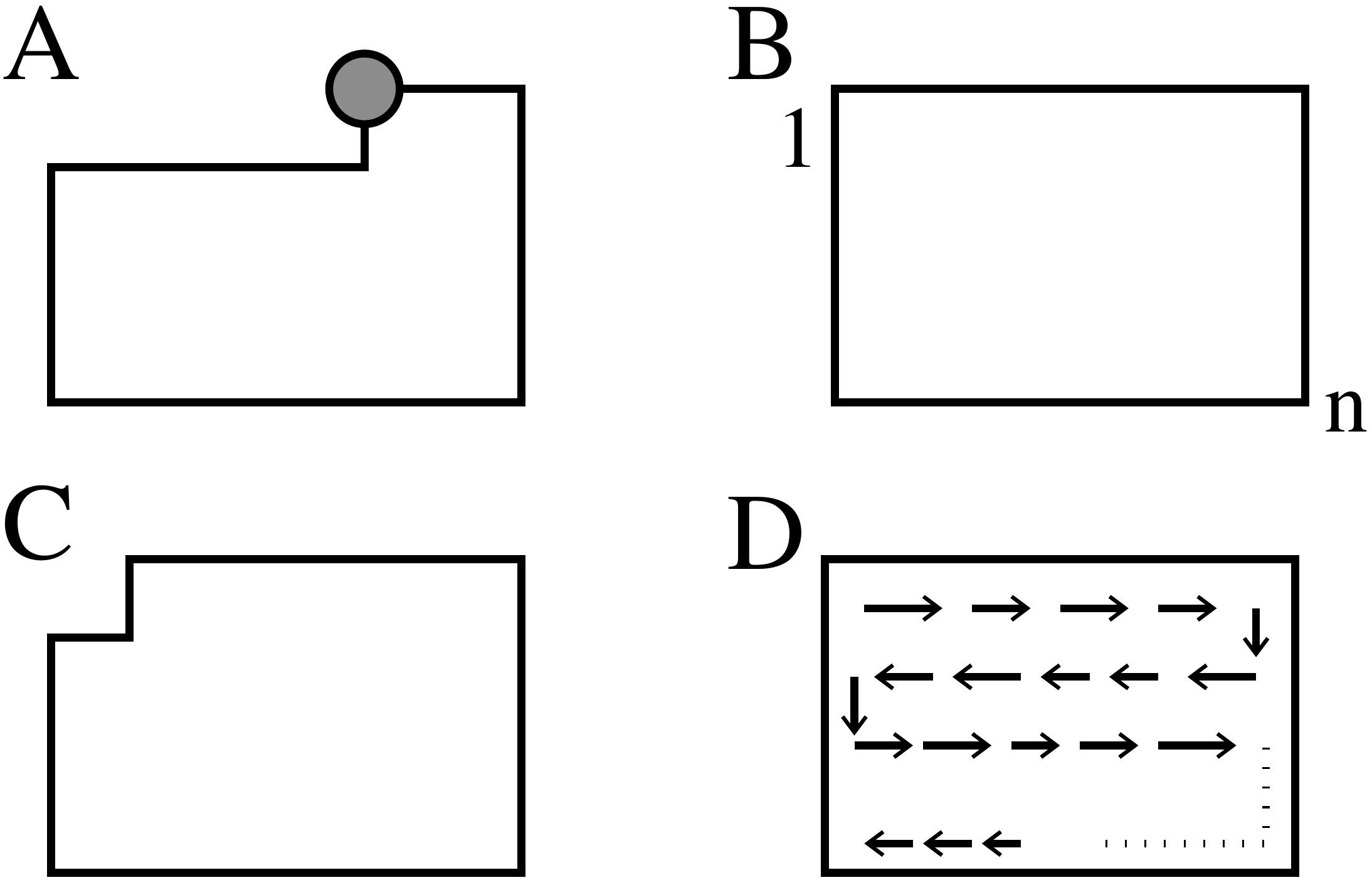}
\end{center}
\caption{A. Typical corner on which magnetisations need to be measured in order to get an approximation for the partition function of the Ising model defined on a square lattice. B. Labelling of spins of the original lattice. C. Lattice obtained after the spin on one corner has been fixed. D. Cartoon for a possible choice to run over all spins of the original lattice}
\label{fig:BQP-DRAW}
\end{figure}

Using Bayes' theorem, to express $\pi_{k-1}$ in terms of marginal and conditional probability distributions,
\beq
\begin{array}{lll}
\pi_{k-1}(\sigma_1 \ldots \sigma_{|\Lambda|})&=&\pi_{k-1}^{(1)}(\sigma_1) \pi_{k-1}^{(2)}(\sigma_2|\sigma_1) \ldots \\
&&\times~\pi_{k-1}^{(|\Lambda|)}(\sigma_{|\Lambda|} |\sigma_1 \ldots \sigma_{|\Lambda|-1}),
\end{array}
\eeq
the r.h.s of the latter inequality can be written as
\bed
\frac{1}{2} \max_{\sigma} |\prod_{l=1}^{|\Lambda|} \frac{\pi'^{(l)}_{k-1}(\sigma_l | \sigma_1 \ldots \sigma_{l-1})}{\pi^{(l)}_{k-1}(\sigma_l | \sigma_1 \ldots \sigma_{l-1})}-1|.
\eed
If we use the \emph{finesse}  
\beq
\mathfrak{f} \equiv 
\max_{k,l,\sigma} \frac{|\pi'^{(l)}_{k-1}(\sigma_l | \sigma_1 \ldots \sigma_{l-1})-\pi^{(l)}_{k-1}(\sigma_l | \sigma_1 \ldots \sigma_{l-1})|}{\pi^{(l)}_{k-1}(\sigma_l | \sigma_1 \ldots \sigma_{l-1})}
\eeq
to quantify the accuracy with which the distributions $\{ \pi'_{k-1} \}$ approach the distributions $\{ \pi_{k-1} \}$, we see that $\Delta \pi^* \leq \frac{1}{2} |(1+\mathfrak{f})^{|\Lambda|}-1| \leq \frac{1}{2} (e^{\mathfrak{f} |\Lambda|}-1)$. So $|Z(h)-\bar{Z}(h)| \leq \epsilon' Z(h)$ whenever the finesse satisfies
\beq\label{eq:finesse-requirement}
\mathfrak{f} \leq \frac{1}{|\Lambda|} \ln [1+\frac{2 e^{-2 \eta}}{|\Lambda|} \ln(1+\epsilon')].
\eeq

We now turn to the second part of our construction and explain how it is possible to get samples for the estimators $\hat{\varrho}_k$ from corner single site magnetisation estimates, as indicated on Fig.\ref{fig:BQP-DRAW}-A. Assume that the $|\Lambda|$ particles of the lattice are numbered as indicated on Fig.\ref{fig:BQP-DRAW}-B. For fixed external field $h_{k-1}$, It is clear that the magnetisation on the corner $'1'$ is given by
\bed
\begin{array}{lll}
m_{k-1}(1)&=&\frac{1}{Z(h_{k-1})}\sum_{\sigma} e^{-\beta H_{k-1}(\sigma )} \sigma_1\\
&=& \pi^{(1)}_{k-1}(\uparrow)-\pi^{(1)}_{k-1}(\downarrow).
\end{array}
\eed
From an estimate $m'_{k-1}(1)$, we construct $\pi'^{(1)}_{k-1}(\sigma_1)$ as
\beq
\begin{array}{lll}
\pi'^{(1)}_{k-1}(\uparrow)&=&\frac{1+m'_{k-1}(1)}{2} \hspace{0.3cm}\\
\pi'^{(1)}_{k-1}(\downarrow)&=&\frac{1-m'_{k-1}(1)}{2}.
\end{array}
\eeq
It is certainly possible to draw \emph{exactly} according to this distribution $\pi'^{(1)}_{k-1}$; it is a \emph{known} two-outcome probability distribution. Let us imagine we do it and obtain an outcome $\mathsf{x}_1$. Then we consider another Ising system, identical to the original apart from the fact that the spin labelled '1' is now fixed to $\mathsf{x}_1$. This new system is now defined on the geometry indicated by Fig.\ref{fig:BQP-DRAW}-C ($|\Lambda|-1$ spins), and governed by the Ising Hamiltonian:
\bed
H^{(2)}(\sigma_2 \ldots \sigma_{|\Lambda|})=H_{k-1}(\mathsf{x}_1 \sigma_2 \ldots \sigma_{|\Lambda|}),
\eed
and its Boltzmann weights obey
\bed
\begin{array}{lll}
\frac{e^{-\beta H^{(2)}(\sigma_2 \ldots \sigma_n)}}{Z^{(2)}}
&=&\pi_{k-1}^{(2)}(\sigma_2|\mathsf{x}_1) \ldots \\
&&\times~\pi_{k-1}^{(|\Lambda|)}(\sigma_{|\Lambda|} |\mathsf{x}_1 \ldots \sigma_{|\Lambda|-1}).
\end{array}
\eed
If we now measure the magnetisation at corner '2' for this new system, we get
\bed
m'_{k-1}(2|\mathsf{x}_1) \simeq 
m_{k-1}(2|\mathsf{x}_1)= \pi_{k-1}^{(2)}(\uparrow|\mathsf{x}_1)-\pi_{k-1}^{(2)}(\downarrow|\mathsf{x}_1)
\eed
The magnetisation $m'_{k-1}(2|\mathsf{x}_1)$ allows to construct
\beq
\begin{array}{lll}
\pi'^{(2)}_{k-1}(\uparrow |\mathsf{x}_1)&=&\frac{1+m'_{k-1}(2|\mathsf{x}_1)}{2} \hspace{0.3cm}\\
\pi'^{(2)}_{k-1}(\downarrow |\mathsf{x}_1)&=&\frac{1-m'_{k-1}(2|\mathsf{x}_1)}{2}.
\end{array}
\eeq
Again, this known probability distribution is simple enough that it is possible to draw exactly a sample $\mathsf{x}_2$ according to it. Repeating this reasoning, running along the lattice in the order indicated by the cartoon on Fig.\ref{fig:BQP-DRAW}-D, we see that the ability to estimate corner magnetisations combined with Bayes' theorem allows to draw sequentially \footnote{The order we have chosen has no particular meaning. The reasoning is of course valid for any labelling of the sites of the lattices.} according to 
\bed
\begin{array}{lll}
\pi'_{k-1}(\sigma_1 \ldots \sigma_{|\Lambda|})&=&\pi'^{(1)}_{k-1}(\sigma_1) \pi'^{(2)}_{k-1}(\sigma_2|\sigma_1) \ldots \\
&&\times~\pi'^{(|\Lambda|)}_{k-1}(\sigma_{|\Lambda|} |\sigma_1 \ldots \sigma_{|\Lambda|-1}).
\end{array}
\eed

Finally, we observe that 
\bed
\begin{array}{ll}
&\frac{|\pi'^{(l)}_{k-1}(\sigma_l | \sigma_1 \ldots \sigma_{l-1})-\pi^{(l)}_{k-1}(\sigma_l | \sigma_1 \ldots \sigma_{l-1})|}{\pi^{(l)}_{k-1}(\sigma_l | \sigma_1 \ldots \sigma_{l-1})|} \\
&\leq
\frac{|m'_{k-1}(l |\sigma_1 \ldots \sigma_{l-1})-m_{k-1}(l |\sigma_1 \ldots \sigma_{l-1})|}{|1-|m_{k-1}(l |\sigma_1 \ldots \sigma_{l-1})| \; |}.
\end{array}
\eed
So the condition (\ref{eq:finesse-requirement}) leads simply to a condition on the \emph{relative} precision over the magnetisation.

Summarising, \emph{for any $\epsilon >0$, temperature $\beta$ and magnetic field $h$, it is possible to provide an estimate $\hat{Z}(h)$ for the Ising partition function $Z(h)$ satisfying
\beq
\text{Prob}[|\hat{Z}(h)-Z(h)| \leq \epsilon \; Z(h)] \geq 3/4,
\eeq
in a time that scales at most polynomially with $\beta,\epsilon^{-1}$, $|h|$, and the size of the system if we are able to perform corner magnetisation measurements on related non-homogeneous Ising systems. The required relative precision need not be lower than the inverse of some polynomial in $|h|, \beta,\epsilon^{-1}$ and the size of the system.}

\end{document}